%% file: topoc.tex
\documentclass[acmjacm]{acmsmall}
 
\usepackage{latexsym,xspace,calc,amsthm}
\usepackage{amsmath,amssymb} %
\usepackage{nicefrac} %

\newcommand\greyprint{}

\usepackage{lmodern}
\usepackage{microtype}
\usepackage{tgtermes}
\usepackage{fix-cm}

\usepackage{bm} %

\usepackage{cancel}

\usepackage{savesym}
\savesymbol{subcaption}
\usepackage{subcaption}

\usepackage{graphbox} %

\makeatletter
\def\@biblabel#1{[#1]} %
\def\thebibliography#1{%
    \footnotesize
    \refsection*{{\refname}
        \@mkboth{\uppercase{\refname}}{\uppercase{\refname}}%
    }
    \list{\@biblabel{\@arabic\c@enumiv}}%
       {\settowidth\labelwidth{\@biblabel{#1}}%
        \leftmargin\labelwidth
        \advance\leftmargin\bibindent
        \itemindent-\bibindent
        \itemsep2pt
        \parsep \z@
        \usecounter{enumiv}%
        \let\p@enumiv\@empty
        \renewcommand\theenumiv{\@arabic\c@enumiv}%
    }%
    \let\newblock\@empty
    \sloppy
    \sfcode`\.=1000\relax
}
\makeatother

 \renewenvironment{thebibliography}[1]{%
   \begin{odlthebibliography}{#1}%
     \setlength{\parskip}{0ex}%
     \setlength{\itemsep}{3pt}%
     \fontsize{9.5}{9.5} %
     \selectfont
}%
 {%
   \end{odlthebibliography}%
 }

\usepackage{tikz}
\usepackage{tikz-cd}
\usetikzlibrary{cd,decorations.markings}
\tikzset{
    dharrow/.style={
        <->,
        postaction={decorate,-},
        }
}
\tikzset{
    dhdashedarrow/.style={
        <->,
        dashed,
        postaction={decorate,-},
        }
    }
\tikzset{
    lrharpoonarrow/.style={
        <[harpoon]->[harpoon],
        postaction={decorate,-},
        }
}
\tikzset{
    lrharpoondashedarrow/.style={
        <[harpoon]->[harpoon],
        dashed, %
        postaction={decorate,-},
        }
}
\usetikzlibrary{arrows}
\usetikzlibrary {arrows.meta} 
\usetikzlibrary{calc}
\usetikzlibrary{positioning}
\usetikzlibrary{snakes,automata,chains}
\usetikzlibrary{graphs}

\usepackage{binarytree}

\usepackage{amssymb,stmaryrd,amsmath}
\usepackage{mdwlist} %
\usepackage{float}   %
\usepackage{centernot} %

\usepackage[colorlinks]{hyperref}
\usepackage{breakurl}  %

\newcommand\jamiesection[1]{\section{#1}}
\newcommand\jamiesubsection[1]{\subsection{#1}}
\newcommand\jamiesubsubsection[1]{\subsubsection{#1}}

\newtheoremstyle{jamiestyle}%
  {4pt}%
  {0pt}%
  {\it}%
  {0pt}%
  {\sc}%
  {.}%
  { }%
  {}%
\theoremstyle{jamiestyle}
\newtheorem{thrm}{Theorem}[subsection]
\newtheorem{prop}[thrm]{Proposition}
\newtheorem{lemm}[thrm]{Lemma}
\newtheorem{corr}[thrm]{Corollary}

\newtheoremstyle{jamienfstyle}%
  {4pt}%
  {0pt}%
  {\normalfont}%
  {0pt}%
  {\sc}%
  {.}%
  { }%
  {}%
\theoremstyle{jamienfstyle}
\newtheorem{nttn}[thrm]{Notation}
\newtheorem{defn}[thrm]{Definition}
\newtheorem{xmpl}[thrm]{Example}
\newtheorem{rmrk}[thrm]{Remark}

\usepackage{color}
\definecolor{mygreen}{rgb}{0,0.6,0}
\definecolor{mygray}{rgb}{0.5,0.5,0.5}
\definecolor{mymauve}{rgb}{0.58,0,0.82}

\usepackage{listings}
 
\definecolor{gray}{RGB}{128, 128, 128}
\definecolor{lightgray}{RGB}{200, 200, 200}
\definecolor{cyan}{RGB}{0, 255, 255}
\definecolor{blue}{RGB}{0, 0, 255}
\definecolor{red}{RGB}{255, 0, 0}
\definecolor{pink}{RGB}{255, 128, 128}
\definecolor{green}{RGB}{0, 128, 0}
\definecolor{lightyellow}{RGB}{255, 255, 200}
\definecolor{purple}{RGB}{128, 0, 128}

\lstdefinestyle{all}
    {basicstyle=\ttfamily\scriptsize,
     keywordstyle=\color{blue}\ttfamily\scriptsize,
     commentstyle=\color{green}\ttfamily\scriptsize,
     stringstyle=\color{red}\ttfamily\scriptsize}

\lstdefinelanguage{hask}{%
    frame=none,
    xleftmargin=2pt,
    belowcaptionskip=\bigskipamount,
    captionpos=b,
    tabsize=2,
    numbers=left,
    numberstyle=\tiny\color{gray},
    emphstyle={\bf},
	morecomment=[s][\color{green}]{\{-}{-\}},
    stringstyle=\mdseries\rmfamily,
    commentstyle=\color{green},
    keywords={},
    keywords=[1]{case, of, data, if, then, else, where, let, in, do},
    keywords=[2]{Chip, Config, CurrencySymbol, TokenName, PubKeyHash, Integer, Value, State, Action, Text, Maybe, Void, TxConstraints,  Contract},
    keywords=[3]{HasNative},
    keywords=[4]{=>},
    keywords=[5]{Just, Nothing, MkChip, MkConfig, SetPrice, Buy},
    keywordstyle=[1]\mdseries\sffamily\color{red},
    keywordstyle=[2]\mdseries\sffamily\color{blue},
    keywordstyle=[3]\mdseries\sffamily\color{green},
    keywordstyle=[4]\mdseries\sffamily,
    keywordstyle=[5]\mdseries\sffamily\color{purple},
    columns=flexible,
    basicstyle=\small\sffamily,
    showstringspaces=false,
    breaklines=false,
    showspaces=false,
    escapeinside={--}{\^^M},escapebegin={\color{green}--},escapeend={},
    literate= {+}{{$+$}}1 {/}{{$/$}}1 {*}{{$*$}}1 {=}{{$=$}}1
              {>}{{$>$}}1 {<}{{$<$}}1 {\\}{{$\lambda$}}1
              {\\\\}{{\char`\\\char`\\}}1
              {->}{{$\rightarrow$}}2 {>=}{{$\geq$}}2 {<-}{{$\leftarrow$}}2
              {<=}{{$\leq$}}2 {=>}{{$\Rightarrow$}}2
              {\ .}{{$\circ$}}2 {\ .\ }{{$\circ$}}2
              {>>}{{>>}}2 {>>=}{{>>=}}2
              {|}{{$\mid$}}1
              {\_}{{\underline{\hspace{2mm}}}}2
}

\lstdefinelanguage{solidity}{%
    frame=none,
    xleftmargin=2pt,
    belowcaptionskip=\bigskipamount,
    captionpos=b,
    tabsize=2,
    numbers=left,
    numberstyle=\tiny\color{gray},
    emphstyle={\bf},
	morecomment=[s][\color{green}]{\{-}{-\}},
    stringstyle=\mdseries\rmfamily,
    commentstyle=\color{green},
    keywords={},
    keywords=[1]{pragma, solidity, contract, event, constructor, require, function, return, emit},
    keywords=[2]{address, uint, mapping},
    keywords=[3]{public, payable, external, view, returns},
    keywords=[4]{=>, +=, -=, =, <=, ==},
    keywords=[5]{msg, sender, transfer, value},
    keywordstyle=[1]\mdseries\sffamily\color{red},
    keywordstyle=[2]\mdseries\sffamily\color{blue},
    keywordstyle=[3]\mdseries\sffamily\color{green},
    keywordstyle=[4]\mdseries\sffamily,
    keywordstyle=[5]\mdseries\sffamily\color{purple},
    columns=flexible,
    basicstyle=\small\sffamily,
    showstringspaces=false,
    breaklines=false,
    showspaces=false,
    escapeinside={--}{\^^M},escapebegin={\color{green}--},escapeend={},
    literate= {+}{{$+$}}1 {/}{{$/$}}1 {*}{{$*$}}1 {=}{{$=$}}1
              {>}{{$>$}}1 {<}{{$<$}}1 {\\}{{$\lambda$}}1
              {\\\\}{{\char`\\\char`\\}}1
              {->}{{$\rightarrow$}}2 {>=}{{$\geq$}}2 {<-}{{$\leftarrow$}}2
              {<=}{{$\leq$}}2 {=>}{{$\Rightarrow$}}2
              {\ .}{{$\circ$}}2 {\ .\ }{{$\circ$}}2
              {>>}{{>>}}2 {>>=}{{>>=}}2
              {|}{{$\mid$}}1
              {\_}{{\underline{\hspace{2mm}}}}2
}

\makeatletter
\newcommand\hpn[2][]{%
  \ext@arrow 9999{\hpnfill@}{#1}{#2}}
\newcommand\hpnfill@{%
  \arrowfill@\leftharpoonup\relbar\rightharpoondown}
\makeatother

\NewCommandCopy{\oldin}{\in}
\renewcommand\in{{{\hspace{1pt}{\oldin}\hspace{1pt}}}}
\NewCommandCopy{\oldnotin}{\notin}
\renewcommand\notin{{{\hspace{1pt}{\oldnotin}\hspace{1pt}}}}

\NewCommandCopy{\oldsetminus}{\setminus}
\renewcommand\setminus{{{\hspace{1pt}{\oldsetminus}\hspace{1pt}}}}

\newcommand\xor{\mathbin{\mathsf{\small xor}}}

\newcommand\atopen{T}

\newcommand\avaluation{f} %

\newcommand\leftopeninterval[1]{(#1]}
\newcommand\rightopeninterval[1]{[#1)}

\newcommand\opens{{\tf{Open}}}
\newcommand\regularOpens{\tf{Open}_{\f{reg}}}
\newcommand\topens{\tf{Topen}}
\newcommand\closed{\tf{Closed}}
\newcommand\regularClosed{\tf{Closed}_{\f{reg}}}

\newcommand\closure[1]{|#1|}

\newcommand{\dotarrow}{%
   \mathrel{\ooalign{\hss\raise.85ex\hbox{\scalebox{1.25}{\normalfont .}}%
   \kern0.35ex\hss\cr$\rightarrow$}}}

\newcommand\onlineref[2]{\url{#1} (permalink: \url{#2})}
\newcommand\footnoteref[2]{\footnote{See \onlineref{#1}{#2}.}}

\newcommand\notbetween{\mathbin{\cancel{\between}}}
\newcommand\notintertwinedwith{\mathrel{\notbetween}}

\newcommand\notintersectswith{\notbetween}

\newcommand\intertwined[1]{#1_{\between}}
\newcommand\intertwinedwith{\mathrel{\between}}

\newcommand\leqk{\leq_{\hspace{-.7pt}\intertwinedwith}}
\newcommand\geqk{\geq_{\hspace{-.7pt}\intertwinedwith}}
\newcommand\nbhd[0]{\f{nbhd}}
\newcommand\interior[0]{\f{interior}}
\newcommand\kiss[0]{\f{kiss}}
\newcommand\community[0]{\f{K}}

\makeatletter
\newcommand\@deffont[2][]{{\bfseries #2}\index{#1}}
\newcommand\deffont{\@dblarg\@deffont}
\makeatother
\newcommand\powerset{\f{pow}}

\newcommand\f[1]{\mathit{#1}}
\newcommand\tf[1]{\mathsf{#1}}
\newcommand\ns[1]{\bm{\mathsf{#1}}}

\newcommand\liff{\Longleftrightarrow}
\newcommand\limp{\Longrightarrow}

\DeclareMathSymbol{\shortminus}{\mathbin}{AMSa}{"39}
\newcommand\minus{{\shortminus}}
\newcommand\plus{{+}}
\newcommand\Forall[1]{\forall #1.}
\newcommand\Exists[1]{\exists #1.}

\newcommand\cent{\vdash}

\newcommand\ment{\vDash}

\newcommand\boundary{\f{boundary}}

\newcommand\mone{{\text{-}1}}

\makeatletter
\DeclareRobustCommand{\barcent}{\mathbin{\mathpalette\barcent@@\relax}}
\newcommand{\barcent@@}[2]{%
  \vbox{\offinterlineskip
    \sbox\z@{$\m@th#1\cent$}%
    \ialign{%
      \hfil##\hfil\cr
      $\m@th#1{}_{\minus}\kern-\scriptspace$\cr
      \noalign{\kern-.3\ht\z@}
      \box\z@\cr
    }%
  }%
}
\makeatother

\makeatletter
\def\pmb@#1#2{\setbox8\hbox{$\m@th#1{#2}$}%
  \setboxz@h{$\m@th#1\mkern-.1mu$}\pmbraise@\wdz@
  \binrel@{#2}%
  \dimen@-\wd8 %
  \binrel@@{%
    \mkern-.1mu\copy8 %
    \kern\dimen@\mkern-.2mu\copy8 %
    \kern\dimen@\mkern-.3mu\copy8 %
    \kern\dimen@\mkern-.4mu\copy8 %
    \kern\dimen@\mkern.1mu\copy8 %
    \kern\dimen@\mkern.2mu\copy8 %
    \kern\dimen@\mkern.3mu\copy8 %
    \kern\dimen@\mkern.0mu\raise\pmbraise@\copy8 %
    \kern\dimen@\mkern.4mu\box8 %
           }%
}
\makeatother

\newcommand\tor{{\pmb\vee}}

\makeatletter
\newcommand{\circlearrow}{}%
\DeclareRobustCommand{\circlearrow}{%
  \mathrel{\vphantom{\shortrightarrow}\mathpalette\circle@arrow\relax}%
}
\newcommand{\circle@arrow}[2]{%
  \m@th
  \ooalign{%
    \hidewidth$#1\circ\mkern1mu$\hidewidth\cr
    $#1\longrightarrow$\cr}%
}
\makeatother

\makeatletter
\newcommand*\bigcdot{\mathpalette\bigcdot@{.5}}
\newcommand*\bigcdot@[2]{\mathbin{\vcenter{\hbox{\scalebox{#2}{$\m@th#1\bullet$}}}}}
\makeatother

\usepackage{datetime}
\yyyymmdddate

\begin{document}
\title{Semitopology: a topological approach to decentralised collaborative action} 
\newcommand\titlerunning{\emph{Semitopology \& decentralised action}}
\newcommand\authorrunning{\emph{Murdoch J. Gabbay}}
\author{Murdoch J. Gabbay \affil{Heriot-Watt University, UK}
}

\begin{abstract}
We introduce \emph{semitopology}, a generalisation of point-set topology that removes the restriction that intersections of open sets need necessarily be open.
The intuition is that points represent participants in a decentralised system, and open sets represent collections of participants that collectively have the authority to collaborate to update their local state; we call this an \emph{actionable coalition}.

Examples of actionable coalition include: majority stakes in proof-of-stake blockchains; communicating peers in peer-to-peer networks; and even pedestrians working together to not bump into one another in the street.
Where actionable coalitions exist, they have in common that: collaborations are local (updating the states of the participants in the coalition, but not immediately those of the whole system); collaborations are voluntary (up to and including breaking rules); participants may be heterogeneous in their computing power or in their goals (not all pedestrians want to go to the same place); participants can choose with whom to collaborate; and they are not assumed subject to permission or synchronisation by a central authority.

We develop a topology-flavoured mathematics that goes some way to explaining how and why these complex decentralised systems can exhibit order, and gives us new ways to understand existing practical implementations. 

Semitopology is also interesting in and of itself, having a rich and interesting theory which quickly deviates from standard accounts on topological spaces.
It soon becomes clear that the most interesting semitopologies are rather ill-behaved from the usual viewpoint, as they are never Hausdorff. 
A notion of `transitive open sets' (topens) becomes central to the story, as topens define subsets of participants who should decide the same value in a distributed system that tries to achieve consensus, and points are called `regular' when they have a topen neighbourhood. 
The theory is then further developed by introducing intertwined points, closures, closed sets, and two interesting characterisations of regularity.

\end{abstract}
\maketitle
\thispagestyle{empty}

\tableofcontents

\jamiesection{Introduction}
\label{sect.intro}

\jamiesubsection{What is a `decentralised collaborative action', and what is a semitopology?}
\label{subsect.what.is}

A system is \emph{decentralised} when it is distributed over several machines and furthermore the system as a whole is not centrally controlled.
Most blockchain systems and peer-to-peer networks are decentralised (they are distributed over multiple participants, and no single entity controls the system). 
The internet is also (mostly) decentralised, at least in principle.\footnote{The internet was designed to be an information network that would be resilient to nuclear attack.  It did this by being `centrifugal'; emphasising node-to-node actions instead of centre-to-centre actions.  See~\cite{ryan:hisidf}, summarised by Ars Technica~\cite{ars-technica:howabg}.} 
Common practical problems from daily life can also be understood in terms of decentralised collaborative action: for example when we drive along a road, or walk around in a shop, we collaborate with the other agents (drivers, or shoppers) in a local and decentralised manner to avoid collisions.

So decentralised collaborative action is everywhere, but it has gained particular interest recently to designers of computer systems because it is an \emph{essential} feature of many modern highly-decentralised computer systems, such as blockchains.
So at a very high level, what do we have?
\begin{enumerate*}
\item
There is a notion of what I will call an \emph{actionable coalition} (or just \emph{open set}).

This is a set $O\subseteq\ns P$ of participants with the capability, though not the obligation, to act collaboratively to advance (= update / transition) the local state of the elements in $O$, possibly but not necessarily in the same way for every $p\in O$.\footnote{E.g. in a blockchain, we may want all updates to be uniform so that we implement a decentralised ledger; but in a peer-to-peer system or the internet, updates need not be uniform, e.g. if nodes are swapping or forwarding data.}
\item
$\varnothing$ is trivially an actionable coalition.
Also we assume that $\ns P$ is actionable, since if it were not then literally nothing could ever get done.
\item
A sets union of actionable coalitions, is an actionable coalition.
\end{enumerate*}
Some important notes about this:
\begin{enumerate*}
\item
State must be stored and updated locally (if state were centralised, then whoever controls the state has \emph{de facto} control of the system, which would not be decentralised). 
\item
An actionable coalition can progress locally, \emph{without} consulting the rest of the system (if they had to, then the system would not be decentralised).
\item
Being a member of an actionable coalition does not imply control.
Actionable coalitions describe legal collaborations, but do not imply any obligation.
\item 
If $O$ is an actionable coalition for $p\in O$, and $p'\in O$ is another participant in $O$, then $O$ is also an actionable coalition for $p'$.
Note that this makes actionable coalitions look a bit like open sets in a topology.
\end{enumerate*}
So we can now introduce our first mathematical abstraction: we identify participants as \emph{points}, and we let \emph{open sets} be \emph{actionable coalitions}.
An actionable coalition is a \emph{coalition of participants with the capacity to act}.
They are not obliged to act, and if they do their action need not be identical across all participants, but the potential exists for this set to collaborate to progress their states.
This leads us to the definition of a semitopology.
\begin{nttn}
\label{nttn.powerset}
Suppose $\ns P$ is a set.
Write $\powerset(\ns P)$ for the powerset of $\ns P$ (the set of subsets of $\ns P$).
\end{nttn}

\begin{defn}
\label{defn.semitopology}
A \deffont{semitopological space}, or \deffont{semitopology} for short, consists of a pair $(\ns P, \opens(\ns P))$ of 
\begin{itemize*}
\item
a (possibly empty) set $\ns P$ of \deffont{points}, and 
\item
a set $\opens(\ns P)\subseteq\powerset(\ns P)$ of \deffont{open sets}, 
\end{itemize*}
such that:
\begin{enumerate*}
\item\label{semitopology.empty.and.universe}
$\varnothing\in\opens(\ns P)$ and $\ns P\in\opens(\ns P)$.
\item\label{semitopology.unions}
If $X\subseteq\opens(\ns P)$ then $\bigcup X\in\opens(\ns P)$.\footnote{There is a little overlap between this clause and the first one: if $X=\varnothing$ then by convention $\bigcup X=\varnothing$.  Thus, $\varnothing\in\opens(\ns P)$ follows from both clause~1 and clause~2.  If desired, the reader can just remove the condition $\varnothing\in\opens(\ns P)$ from clause~1, and no harm would come of it.} 
\end{enumerate*}
We may write $\opens(\ns P)$ just as $\opens$, if $\ns P$ is irrelevant or understood, and we may write $\opens_{\neq\varnothing}$ for the set of nonempty open sets.
\end{defn}

The reader will recognise a semitopology as being like a \emph{topology} on $\ns P$, but without the condition that the intersection of two open sets necessarily be an open set.
This reflects the fact that the intersection of two actionable coalitions need not itself be an actionable coalition.

\jamiesubsection{How does this lead to new maths?}

To get a flavour of our mathematical results, consider a fundamental problem in any decentralised system: ensuring that its participants remain in agreement, for some suitable sense of `agree'.

To take a simple example from blockchain: if we reach a situation where half of the nodes say that we have paid for a service, and the other half say that we have not --- then \emph{everyone} has a problem, because the system has become incoherent and it is not clear how the system as a whole can restore coherence and progress.\footnote{coherent (adj.) 1550s, ``harmonious;'' 1570s, ``sticking together,'' also ``connected, consistent'' (of speech, thought, etc.), from French cohérent (16c.), from Latin cohaerentem (nominative cohaerens), present participle of cohaerere ``cohere,'' from assimilated form of com ``together'' (see co-) + haerere ``to adhere, stick'' (etymologyonline: \url{https://www.etymonline.com/word/coherent).}}
This phenomenon is called \emph{forking}, and blockchain designers really want to avoid it!

We will call our mathematical abstraction of agreement, \emph{antiseparation}.
In a little more detail, antiseparation properties are coherence properties that are guaranteed to hold of a decentralised system
\emph{just} by analysing the structure of its actionable coalitions.
If we recall the usual separation axioms of topology --- such as $T_0$, $T_1$, Hausdorff, and so on --- note that these separation conditions have to do with the existence of non-intersecting open sets (or similar).  
Concretely, antiseparation assumptions on semitopologies are dual to this; they give various senses in which open neighbourhoods \emph{must intersect}.

It turns out that these are interesting properties to have, because they determine participants who should decide the same value in a distributed system that tries to achieve consensus.
It turns out that we can get surprisingly detailed information about consensus behaviour in decentralised systems just from 
quite weak and abstract antiseparation assumptions on the actionable coalitions (= open sets).

We emphasise this point: sometimes we can predict important macro properties of a system's behaviour without knowing anything about its specifics, so long as we have certain good properties on its actionable coalitions.

Let us start by considering a simple situation where participants are trying to agree on a binary consensus problem: whether to announce a single value `true' or `false'.
Continuing the theme of simplicity, assume some finite nonempty set of participants $\mathbb E$ and let their actionable coalitions be just any set of participants that forms a majority (so it contains strictly more than half of the set of all participants).
Now suppose that the participants in some actionable coalition $O\subseteq\mathbb E$ have communicated and have agreed on `true'.
Because they form an actionable coalition, they are entitled to act and to announce `true', and so they do.
They have now all committed to this state update and they cannot change their minds.

So: can this system fork?
Consider some participant $p\not\oldin O$.
If $p$ wants to make progress, is must also agree on `true', because all of its actionable coalitions intersect with $O$ and so contain at least one participant that has committed to `true' and cannot change its mind.
This does not mean that $p$ has to agree on `true'; it could choose not to progress, or it could break the rules.
But, by definition if $p$ does want to progress legally, then the decision has been made and it must eventually go along with the majority.
Thus, we have proved that any progress that is made by one participant within the rules (\dots must be shared with some actionable coalition of that participant, and since all such coalitions intersect it \dots) must eventually be followed any other participant that also progresses.
Thus forking is impossible.

The reader may already be familiar with this example, but note that this antiseparation property comes simply \emph{from the structure of the actionable coalitions}.
There is no need to consider the protocol, or even how values are interpreted.

Surprisingly, it turns out that antiseparation-style behaviour is common, and arises even if we do not require actionable coalitions that are simple majorities.
For example, let participants be $\mathbb Z=\{0,1,\minus 1,2,\minus 2,\dots\}$ and let actionable coalitions be generated by sets of three consecutive numbers starting at an even number $\{2i,2i\plus 1,2i\plus 2\}$, and suppose again that we are trying to agree on `true' or `false'.
Note that in contrast to the previous example, actionable coalitions need not intersect.
Yet, the moment one triplet of participants commits to `true', the rest of the system is obliged to eventually agree, if all participants play by the rules.
Now this particular example system is not particularly safe or desirable in practice, because we can imagine that $\{0,1,2\}$ agree on `true', and $\{4,5,6\}$ acting independently but in good faith agree on `false', and then $3$ cannot legally progress, because within $\{2,3,4\}$, $2$ has announced `true' and $4$ has announced `false' and $3$ cannot agree with both.
But, we know that \emph{if} all participants do legally progress, then they announce the same value.
So this example illustrates how antiseparation can arise even when actionable coalitions are rather small.\footnote{See also Remark~\ref{rmrk.transitive.comment}.}

The two examples above are quite different.
In one, all actionable coalitions intersect, and in the other they mostly do not.
This suggests that a `general mathematics of (anti)separation' is possible, based on the study of actionable coalitions.
In a nutshell, that mathematical story is what we will develop. 

The notion of actionable coalitions is introduced in this document, but in retrospect we see them everywhere. 
For example:
Some blockchain systems make actionable coalitions explicit, e.g. in the XRP Ledger~\cite{schwartz_ripple_2014} and the Stellar network~\cite{lokhafa:fassgp} the notion of actionable coalition is represented explicitly in the engineering architecture of the system.
Social choice theorists have a similar notion called a \emph{winning coalition} \cite[Item~5, page~40]{riker:thepc}, which is used to study voting systems; and if the reader has a background in logic then they may be reminded of a whole field of \emph{generalised quantifiers} (a good survey is in~\cite{sep-generalized-quantifiers}).\footnote{But, note that voting and generalised quantifiers have a centralised flavour to them.  For instance: a vote in the typical democratic sense is a synchronous, global operation (unless the result is disputed): votes are cast, collected, and then everyone gets together --- e.g. in a vote counting hall --- to count the votes and agree on who won and so certify the outcome.}
Cross-chain systems (which operate or translate across multiple blockchain) inherently have to deal with heterogeneous actionable coalitions, since the actionable coalitions of one blockchain need not be (and usually are not) identical to those of another. 
Concrete algorithms to attain consensus often use a notion of \emph{quorum}~\cite{lamport_part-time_1998,lamport:byzgp} $Q$ for a participant $p$; 
simply put, this is a set of participants $Q$ whose unanimous adoption of a value guarantees that $p$ will eventually also adopt this value.
If quorums are majorities (more than half) or supermajorities (more than two-thirds) of all participants then quorums already \emph{are} actionable coalitions; if not, then we can obtain an actionable coalition in a natural way by considering any set $O$ such that every participant (element) $p\in O$ has some subset $Q\subseteq O$ that is a quorum for $p$.

\jamiesubsection{Who should read this paper?}

\begin{enumerate}
\item
\emph{Practitioners} looking for a mathematical framework that subsumes what they're already doing, puts it in a broader context, creates a common language to speak with one another and with mathematicians, and suggests new engineering options.
\item 
\emph{Theoreticians} looking for maths to help design the next generation of advanced decentralised computer systems. 
\item
\emph{Pure mathematicians} who might be pleased to discover a new topology-adjacent field and might see it as a fresh research opportunity.\footnote{We can also learn what things are important and interesting to look at, and what distinctions make a difference in practice; I know that I have.}
\item
\emph{Mathematicians} looking to get into practical systems.
Real systems are often messy, because they have to accommodate a messy reality.
Semitopologies provide a useful abstraction that can help us to understand what is going on at a high level.
\end{enumerate} 

\jamiesubsection{Why did I write it?}

Numerous authors have recently studied designing systems where participants have different opinions on who is part of the system or on who is trustworthy or not~\cite{Alpos2024,sheff_heterogeneous_2021,cachin_quorum_2023,li_quorum_2023,bezerra_relaxed_2022,garcia2018federated,lokhafa:fassgp,losa:stecbi,florian_sum_2022,li_open_2023}.
These systems go by names such as \emph{(permissionless) fail-prone systems} and \emph{(heterogeneous) quorum systems} (more discussion, with more references, is in Subsection~\ref{subsect.related.work}).

Most of these systems are (or to be more precise: they directly give rise to) semitopologies, and it seems to me that the literature above is, in fact, \emph{rediscovering topology through semitopology}, but they did not know it. 
Here, we make the connection to classical mathematics explicit, and build on it to obtain results that matter and say something about the (expected) behaviours of these new classes of systems.

\jamiesubsection{Map of the paper}
\label{subsect.map}

\begin{enumerate}
\item
Section~\ref{sect.intro} is the Introduction.  You Are Here.
\item
In Section~\ref{sect.semitopology} we show how continuity corresponds to local agreement (Definition~\ref{defn.semitopology} and Lemma~\ref{lemm.open.lc}).
\item
In Section~\ref{sect.transitive.sets} we introduce \emph{transitive sets}, \emph{topens}, and \emph{intertwined points}.
These are all different views on the anti-separation well-behavedness properties that will interest us. 
Most of Section~\ref{sect.transitive.sets} is concerned with showing how these different views relate and in what senses they are equivalent (e.g. Proposition~\ref{prop.cc.char}).
Transitive sets are guaranteed to be in agreement (in a sense made precise in Theorem~\ref{thrm.correlated} and Corollary~\ref{corr.correlated.intersect}), and we take a first step to understanding the fine structure of semitopologies by proving that every semitopology partitions into topen sets (Theorem~\ref{thrm.topen.partition}), plus other kinds of points which we classify in the next Section.
\item
In Section~\ref{sect.regular.points} we start to classify points in more detail, introducing notions of \emph{regular}, \emph{weakly regular}, and \emph{quasiregular} points (Definition~\ref{defn.tn}).\footnote{The other main classification is \emph{conflicted} points, in Definition~\ref{defn.conflicted}.  These properties are connected by an equation: regular = weakly regular + unconflicted; see Theorem~\ref{thrm.r=wr+uc}.}
 
Regular points are those contained in some topen set, and they display particularly good behaviour.
Regularity will be very important to us and we will characterise it in multiple ways: see Remark~\ref{rmrk.how.regularity}.
(A survey of characterisations of weak regularity requires more machinery and appears in Remark~\ref{rmrk.how.weakly.regular}.)
\item
In Section~\ref{sect.closed.sets} we study closed sets, and in particular the interaction between intertwined points, topens, and closures.
Typical results are Proposition~\ref{prop.intertwined.as.closure} and Theorem~\ref{thrm.up.down.char} which characterise sets of intertwined points as minimal closures.
The significance to consensus is discussed in Remarks~\ref{rmrk.fundamental.consensus} and~\ref{rmrk.why.top.closure}.
\item
In Section~\ref{sect.unconflicted.point} we study unconflicted and hypertransitive points, leading to two useful characterisations of regularity in Theorems~\ref{thrm.r=wr+uc} and~\ref{thrm.regular=qr+sc}.
\item
In Section~\ref{sect.conclusions} we conclude and discuss related and future work.
We discuss connections with related work in Subsection~\ref{subsect.related.work}.
\end{enumerate}

\begin{rmrk}
Algebraic topology has been applied to the solvability of distributed-computing tasks in various computational models (e.g. the impossibility of wait-free $k$-set consensus using read-write registers and the Asynchronous Computability Theorem~\cite{herlihy_asynchronous_1993,borowsky_generalized_1993,saks_wait-free_1993}; see~\cite{herlihy_distributed_2013} for a survey).
Semitopology is not topology, and this is not a paper about algebraic topology applied to the solvability of distributed-computing tasks!

This paper is about the mathematics of actionable coalitions, as made precise by point-set semitopologies; their antiseparation properties; and the implications to partially continuous functions on of them.
If we discuss distributed systems, it is by way of providing motivating examples or noting applicability.
\end{rmrk}

\jamiesection{Semitopology}
\label{sect.semitopology}

\jamiesubsection{Definitions, examples, and some discussion}

\jamiesubsubsection{Definitions}

Recall from Definition~\ref{defn.semitopology} the definition of a semitopology.

\begin{rmrk}
\label{rmrk.two.ways.to.think}
\leavevmode
\begin{enumerate*}
\item
As a sets structure, a semitopology on $\ns P$ is like a \emph{topology} on $\ns P$, but without the condition that the intersection of two open sets be an open set.
\item
As a lattice structure, a semitopology on $\ns P$ is a 
bounded complete join-subsemilattice of $\powerset(\ns P)$.\footnote{\emph{Bounded} means closed under empty intersections and unions, i.e. containing the empty and the full set of points.  \emph{Complete} means closed under arbitrary (possibly empty, possibly infinite) sets unions.

The reader may know that a complete lattice is also co-complete: if we have all joins, then we also have all meets.
However, note that there is no reason for the meets in $\opens$ to coincide with the meets in $\powerset(\ns P)$, i.e. for them to be sets intersections.  
}
\item
Every semitopology $(\ns P,\opens)$ induces two natural topological completions: the least topology that contains $\opens$, and the greatest topology contained in $\opens$.
But there is more to semitopologies than just their topological completions, because:
\begin{enumerate*}
\item
We are explicitly interested in situations where intersections of open sets need \emph{not} be open.
\item
Completing to a topology loses information.
For example: the `many', `all-but-one', and `more-than-one' semitopologies in Example~\ref{xmpl.semitopologies} express three distinct notions of quorum, yet if $\ns P$ is infinite then for all three, the least topology containing them is the discrete semitopology (Definition~\ref{defn.value.assignment}(\ref{item.discrete.semitopology})), and the greatest topology that they contain is the trivial topology $\{\varnothing,\ns P\}$ (Example~\ref{xmpl.semitopologies}(\ref{item.trivial.topology})).
See also the overview in Subsection~\ref{subsect.vs}. 
\end{enumerate*}
\end{enumerate*}
\end{rmrk}

Semitopologies are not topologies.
We take a moment to spell out one concrete difference:
\begin{lemm}
\label{lemm.two.min}
In topologies, if a point $p$ has a minimal open neighbourhood then it is least (= unique minimal).
In semitopologies, a point may have multiple distinct minimal open neighbourhoods.
\end{lemm}
\begin{proof}
To see that in a topology every minimal open neighbourhood is least, just note that if $p\in A$ and $p\in B$ then $p\in A\cap B$.
So if $A$ and $B$ are two minimal open neighbourhoods then $A\cap B$ is contained in both and by minimality is equal to both.

To see that in a semitopology a minimal open neighbourhood need not be least, it suffices to provide an example.
Consider $(\ns P,\opens)$ defined as follows, as illustrated in Figure~\ref{fig.two.min}:
\begin{itemize*}
\item
$\ns P=\{0,1,2\}$
\item
$\opens = \bigl\{ \varnothing,\ \{0,1\},\ \{1,2\},\ \{0,1,2\} \bigr\}$
\end{itemize*}
Note that $1$ has two minimal open neighbourhoods: $\{0,1\}$ and $\{1,2\}$. 
\end{proof}

\begin{figure}
\vspace{-1em}
\centering
\includegraphics[align=c,width=0.4\columnwidth,trim={50 120 50 120},clip]{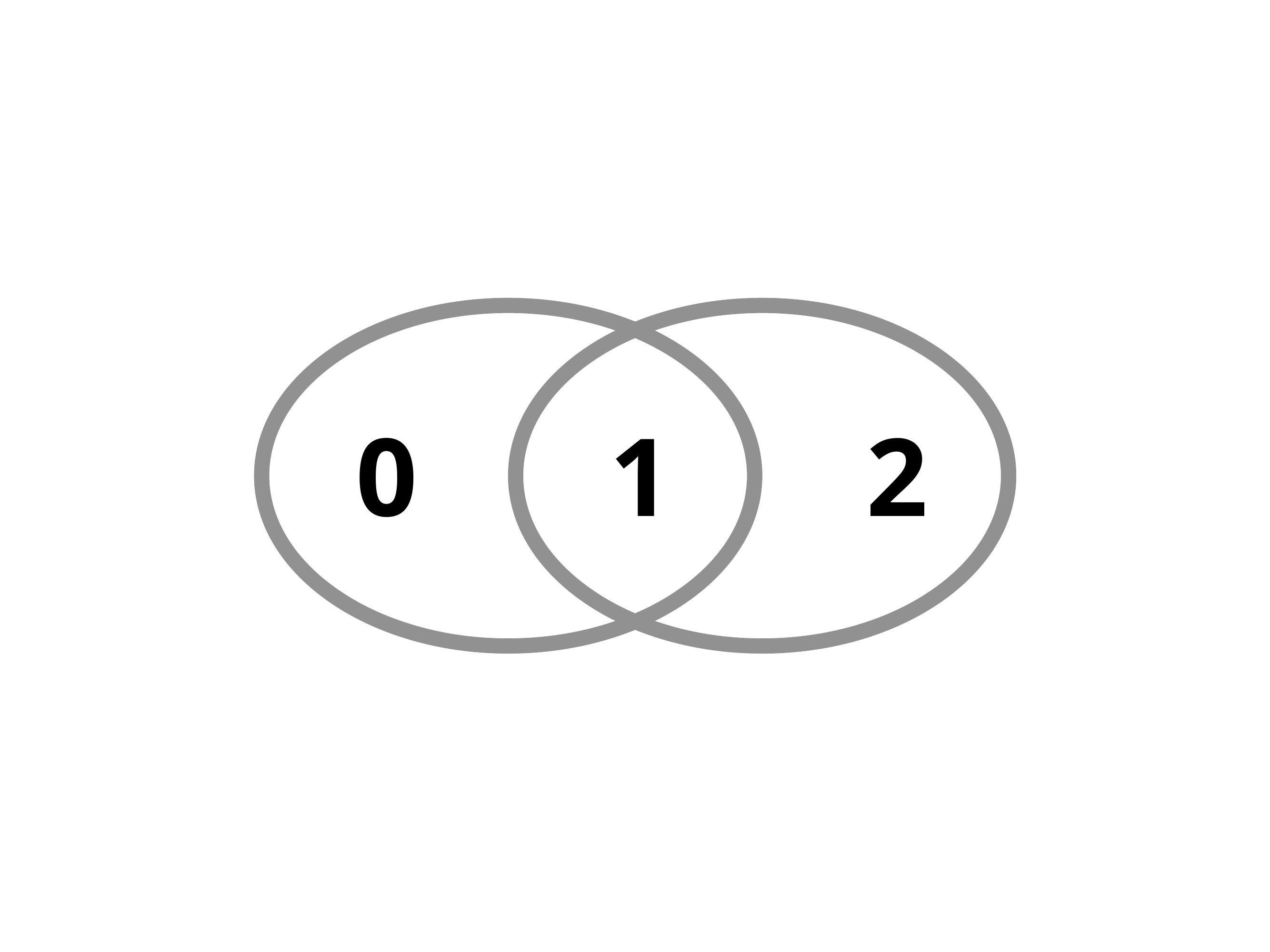}
\vspace{-1em}
\caption{An example of a point with two minimal open neighbourhoods (Lemma~\ref{lemm.two.min})}
\label{fig.two.min}
\end{figure}

\jamiesubsubsection{Examples}

As standard, we can make any set $\tf{Val}$ into a semitopology (indeed, it is also a topology) just by letting open sets be the powerset: 
\begin{defn}
\label{defn.value.assignment}
\leavevmode
\begin{enumerate*}
\item\label{item.discrete.semitopology}
Call $(\ns P,\powerset(\ns P))$ the \deffont{discrete semitopology on $\ns P$}.
 
We may call a set with the discrete semitopology a \deffont{semitopology of values}, and when we do we will usually call it $\tf{Val}$.
We may identify $\tf{Val}$-the-set and $\tf{Val}$-the-discrete-semitopology; meaning will always be clear.
\item\label{item.value.assignment}
When $(\ns P,\opens)$ is a semitopology and $\tf{Val}$ is a semitopology of values, we may call a function $f:\ns P\to\tf{Val}$ a \deffont[value assignment $f:\ns P\to\tf{Val}$]{value assignment}.

Note that a value just assigns values to points, and in particular we do not assume \emph{a priori} that it is continuous, where continuity is defined just as for topologies (see Definition~\ref{defn.continuity}).
\end{enumerate*} 
\end{defn}

\begin{xmpl}
\label{xmpl.semitopologies}
We consider further examples of semitopologies:
\begin{enumerate}
\item
Every topology is also a semitopology; intersections of open sets are allowed to be open in a semitopology, they are just not constrained to be open.
In particular, the discrete topology is also a discrete semitopology (Definition~\ref{defn.value.assignment}(\ref{item.discrete.semitopology})).
\item
The \deffont{initial semitopology} $(\varnothing,\{\varnothing\})$ and the \deffont{final semitopology} $(\{\ast\},\{\varnothing,\{\ast\}\})$ are semitopologies. 
\item\label{item.boolean.discrete}
An important discrete semitopological space is 
$$
\mathbb B=\{\bot,\top\}
\quad\text{with the discrete semitopology}\quad
\opens(\mathbb B)=\{\varnothing, \{\bot\},\{\top\},\{\bot,\top\}\}.
$$
We may silently treat $\mathbb B$ as a (discrete) semitopological space henceforth.
\item\label{item.trivial.topology}
Take $\ns P$ to be any nonempty set.
Let the \deffont[trivial semitopology]{trivial semitopology} (this is also a topology) on $\ns P$ have 
$$
\opens =\{\varnothing, \ns P\}.
$$
So (as usual) there are only two open sets: the one containing nothing, and the one containing every point.\footnote{According to Wikipedia, this space is also called \emph{indiscrete}, \emph{anti-discrete}, \emph{concrete}, and \emph{codiscrete} (\url{https://en.wikipedia.org/wiki/Trivial_topology}).}

The only nonempty open is $\ns P$ itself, reflecting a notion of actionable coalition that requires unanimous agreement. 
\item
Suppose $\ns P$ is a set and $\mathcal F\subseteq\powerset(\ns P)$ is nonempty and up-closed (so if $P\in\mathcal F$ and $P\subseteq P'\subseteq\ns P$ then $P'\in\mathcal F$, then $(\ns P,\mathcal F)$ is a semitopology.
This is not necessarily a topology, because we do not insist that $\mathcal F$ is a filter (i.e. is closed under intersections).

We give four sub-examples for different choices of $\mathcal P\subseteq\powerset(\ns P)$.
\begin{enumerate}
\item\label{item.supermajority}
Take $\ns P$ to be any finite nonempty set.
Let the \deffont{supermajority semitopology} have 
$$
\opens =\{\varnothing\}\cup\{O\subseteq\ns P \mid \f{cardinality}(O)\geq \nicefrac{2}{3}*\f{cardinality}(\ns P)\}.
$$
So $O$ is open when it contains at least two-thirds of the points.

Two-thirds is a typical threshold used for making progress in consensus algorithms.
\item
Take $\ns P$ to be any nonempty set.
Let the \deffont{many semitopology} have
$$
\opens = \{\varnothing\}\cup\{O\subseteq\ns P \mid \f{cardinality}(O)=\f{cardinality}(\ns P)\} .
$$
For example, if $\ns P=\mathbb N$ then open sets include $\f{evens}=\{2*n \mid n\in\mathbb N\}$ and $\f{odds}=\{2*n\plus 1 \mid n\in\mathbb N\}$.

Its notion of open set captures an idea that an actionable coalition is a set that may not be all of $\ns P$, but does at least biject with it.
\item\label{item.counterexample.X-x}
Take $\ns P$ to be any nonempty set.
Let the \deffont{all-but-one semitopology} have
$$
\opens = \{\varnothing,\ \ns P\}\cup\{\ns P\setminus \{p\}\mid p\in\ns P\} .
$$
This semitopology is not a topology.

The notion of actionable coalition here is that there may be at most one objector (but not two).
\item\label{item.counterexample.more-than-one}
Take $\ns P$ to be any set with cardinality at least $2$.
Let the \deffont{more-than-one semitopology} have
$$
\opens = \{\varnothing\}\cup\{O\subseteq\ns P \mid \f{cardinality}(O) \geq 2\} .
$$
This semitopology is not a topology.

This notion of actionable coalition reflects a security principle in banking and accounting (and elsewhere) of \emph{separation of duties}, that functional responsibilities be separated such that at least two people are required to complete an action --- so that errors (or worse) cannot be made without being discovered by another person.
\end{enumerate}
\item
Take $\ns P=\mathbb R$ (the set of real numbers) and let open sets be generated by intervals of the form $\rightopeninterval{0,r}$ or $\leftopeninterval{\minus r,0}$ for any strictly positive real number $r>0$.

This semitopology is not a topology, since (for example) $\leftopeninterval{1,0}$ and $\rightopeninterval{0,1}$ are open, but their intersection $\{0\}$ is not open.
\item\label{item.quorum.system}
In~\cite{naor:loacaq} a notion of \emph{quorum system} is discussed, defined as any collection of pairwise intersecting sets.
Quorum systems are a field of study in their own right, especially in the theory of concrete consensus algorithms.

Every quorum system gives rise naturally to a semitopology, just by closing under arbitrary unions.
We obtain what we will call an \emph{intertwined space} (Notation~\ref{nttn.intertwined.space}; a semitopology all of whose nonempty open sets intersect).\footnote{A topologist would call this a \emph{hyperconnected space}, but be careful! There are multiple such notions in semitopologies, so intuitions need not transfer over.  See the discussion in Subsection~\ref{subsection.topens.in.topologies}.}

Going in the other direction is interesting for a different reason, that it is slightly less canonical: of course every intertwined space is already a quorum system; but (for the finite case) we can also map to the set of all open covers of all points.

To give one specific example of a quorum system from~\cite{naor:loacaq}, consider $n\times n$ grid of cells with quorums being sets consisting of any full row and a full column; note that any two quorums must intersect in at least two points.
We obtain a semitopology just by closing under arbitrary unions.
\end{enumerate}
\end{xmpl}

\begin{rmrk}[Logical models of semitopologies]

\noindent One class of examples of semitopologies deserves its own discussion.
Consider an arbitrary logical system with predicates $\tf{Pred}$ and entailment relation $\cent$.\footnote{A validity relation $\ment$ would also work.}
Call $\Phi\subseteq\tf{Pred}$ \deffont[deductively closed (set of predicates)]{deductively closed} when $\Phi\cent\phi$ implies $\phi\in\Phi$.
Then take 
\begin{itemize*}
\item
$\ns P=\tf{Pred}$, and 
\item
let $O\in\opens$ be $\tf{Pred}$ or the complement to a deductively closed set $\Phi$, so $O=\tf{Pred}\setminus\Phi$.
\end{itemize*}
Note that an arbitrary union of open sets is open (because an arbitrary intersection of deductively closed sets is deductively closed), but an intersection of open sets need not be open (because the union of deductively closed sets need not be deductively closed).
This is a semitopology.
\end{rmrk}

\jamiesubsubsection{Why the name `semitopologies', and other discussion}

\begin{rmrk}[Why the name `semitopologies']
\label{rmrk.why.name.semitopologies}
When we give a name `semitopologies' to things that are like topologies but without intersections, this is a riff on 
\begin{itemize*}
\item
`semilattices', for things that are like lattices with joins but without meets (or vice-versa), and 
\item
`semigroups', for things that are like groups but without inverses.
\end{itemize*}
But, this terminology also reflects a real mathematical connection, because semitopologies \emph{are} semilattices \emph{are} semigroups, in standard ways which we take a moment to spell out: 
\begin{itemize*}
\item
A semitopology $(\ns P,\opens)$ is a bounded join subsemilattice of the powerset $\powerset(\ns P)$, by taking the join $\tor$ to be sets union $\cup$ and the bounds $\bot$ and $\top$ to be $\varnothing$ and $\ns P$ respectively. 
\item
A semilattice is an idempotent commutative monoid, which is an idempotent commutative semigroup with an identity, by taking the multiplication $\circ$ to be $\tor$ and the identity element to be $\bot$ ($\top$ becomes what is called a \emph{zero} or \emph{absorbing} element, such that $\top\circ x=\top$ always).
\end{itemize*} 
\end{rmrk}

\begin{figure}
\centering
\includegraphics[align=c,width=0.4\columnwidth,trim={50 0 50 0},clip]{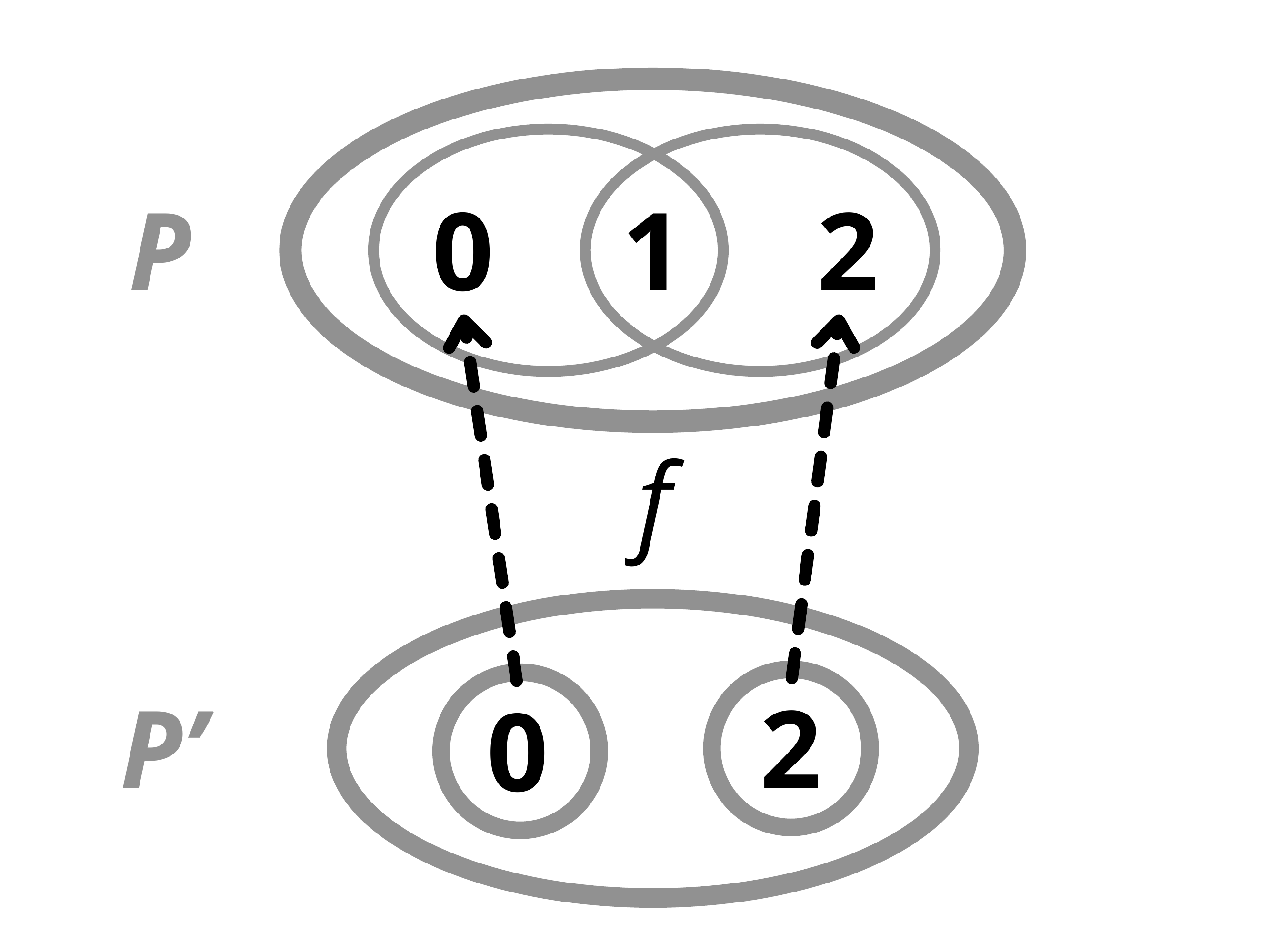}
\caption{Two nonidentical semitopologies (Remark~\ref{rmrk.PtoP})}
\label{fig.PtoP}
\end{figure}

\begin{rmrk}[Semitopologies are not \emph{just} semilattices]
\label{rmrk.PtoP}
We noted in Remark~\ref{rmrk.why.name.semitopologies} that every semitopology is a semilattice.
This is true, but the reader should not read this statement as reductive: semitopologies are not \emph{just} semilattices. 

To see why, consider the following two simple semitopologies, as illustrated in Figure~\ref{fig.PtoP}:
\begin{enumerate*}
\item
$(\ns P,\opens)$ where $\ns P=\{0,1,2\}$ and $\opens=\bigl\{\varnothing,\{0,1\},\{1,2\},\{0,1,2\}\bigr\}$.
\item
$(\ns P',\opens')$ where $\ns P=\{0,2\}$ and $\opens'=\bigl\{\varnothing,\{0\},\{2\},\{0,2\}\bigr\}$.
\end{enumerate*}
Note that the semilattices of open sets $\opens$ and $\opens'$ are isomorphic --- so, when viewed as semilattices these two semitopologies are the same (up to isomorphism).

However, $(\ns P,\opens)$ is not the same semitopology as $(\ns P',\opens')$.
There is more than one way to see this, but perhaps the simplest indication is that for every continuous $f:(\ns P,\opens)\to(\ns P',\opens')$, there is no continuous map $g:(\ns P',\opens')\to(\ns P,\opens)$ such that $g\circ f$ is the identity (we will define continuity in a moment in Definition~\ref{defn.continuity}(\ref{item.continuous.function}) but it is just as for topologies, so we take the liberty of using it here).
There are a limited number of possibilities for $f$ and $g$, and we can just enumerate them and check:
\begin{itemize*}
\item
If $f(0)=0$ and $f(2)=2$ and $g(1)=0$, then $g^\mone(\{2\})=\{2\}\not\oldin\opens$, and if $g(1)=1$ then $g^\mone(\{0\})=\{0\}\not\oldin\opens$. 
\item
If $f(0)=0$ and $f(2)=1$ and $g(1)=0$, then $g^\mone(\{2\})=\{1\}\not\oldin\opens$, and if $g(1)=2$ then $g^\mone(\{0\})=\{0\}\not\oldin\opens$. 
\item
Other possibilities are no harder.
\end{itemize*}
\llap{\phantom{$(\mathbb Q,\opens_{\mathbb Q})$ eliminate LaTeX bug in next para}} 
A similar observation holds for \emph{topologies}: for example, if we write $(\mathbb Q,\opens_{\mathbb Q})$ for the rational numbers with their usual open set topology, and $(\mathbb R,\opens_{\mathbb R})$ for the real numbers with their usual open set topology, then their topologies are isomorphic as lattices, with one direction of the isomorphism given just by $O\in \opens_{\mathbb R}$ maps to $O\cap \mathbb Q\oldin\opens_{\mathbb Q}$. 
This counterexample works for semitopologies too since every topology is also a semitopology.

However, we would still argue that the counterexample in Figure~\ref{fig.PtoP} is inherently stronger; not just because it is smaller (two and three points instead of countably and uncountably many) but also because --- while we can recover $\mathbb R$ from $\mathbb Q$ in a natural and canonical way by forming a completion --- the upper semitopology in Figure~\ref{fig.PtoP} is not \emph{a priori} canonically derived from the lower one.
The two semitopologies in Figure~\ref{fig.PtoP} seem to be distinct in some structural way, yet they still corresponding to the same semilattice, so we see that there is other structure here, which is not reflected by the pure semilattice derived from their open sets. 
\end{rmrk}

\begin{rmrk}[`Stronger' does not necessarily equal `better']
We conclude with some easy predictions about the theory of semitopologies, made just from general mathematical principles.
Fewer axioms means: 
\begin{enumerate*}
\item
\emph{more} models, 
\item
\emph{finer discrimination} between definitions, and 
\item
(because there are more models) \emph{more counterexamples}.
\end{enumerate*}
So we can expect a theory with the look-and-feel of topology, but with new models, new distinctions between definitions that in topology may be equivalent, and some new definitions, theorems, and counterexamples --- and this indeed will be the case.
 
Note that fewer axioms does not necessarily mean fewer interesting things to say and prove.
On the contrary: if we can make finer distinctions, there may also be more interesting things to prove; and furthermore, assumptions we make can become \emph{more} impactful in a weaker system, because these assumptions may exclude more models than would have been the case with more powerful axioms.

For example consider semigroup theory and group theory: every group is a semigroup, but both groups and semigroups have their own distinct character, literature, and applications. 
To take this to an extreme, consider the \emph{terminal} theory, which has just one first-order axiom: $\Exists{x}\Forall{y}x=y$.
This `subsumes' groups, lattices, graphs, and much besides, in the sense that every model of the terminal theory \emph{is} a group, a lattice, and a graph, in a natural way.  
But this theory is so strong, and its models so restricted (just the singleton model with one element) that there is not much left to say about it. 
Additional assumptions we may make on elements add literally nothing of value, because there was only one element to begin with!
\end{rmrk}

\jamiesubsection{Continuity, and its interpretation}
\label{subsect.continuity}

We can import the topological notion of continuity and it works fine in semitopologies, and the fact that there are no surprises is a feature. 
In Remark~\ref{rmrk.continuity=consensus} we explain how these notions matter to us:

\begin{defn}
\label{defn.continuity}
We import standard topological notions of inverse image and continuity:
\begin{enumerate}
\item
Suppose $\ns P$ and $\ns P'$ are any sets and $f:\ns P\to\ns P'$ is a function.
Suppose $O'\subseteq\ns P'$.
Then write $f^\mone(O')$ for the \deffont[inverse image $f^\mone(O')$]{inverse image} or \deffont[preimage $f^\mone(O')$]{preimage} of $O'$, defined by
$$
f^\mone(O')=\{p{\in}\ns P \mid f(p)\in O'\} . 
$$
\item\label{item.continuous.function}
Suppose $(\ns P,\opens)$ and $(\ns P',\opens')$ are semitopological spaces (Definition~\ref{defn.semitopology}).
Call a function $f:\ns P\to\ns P'$ \deffont[continuous function]{continuous} when the inverse image of an open set is open.
In symbols:
$$
\Forall{O'\in\opens'} f^\mone(O')\oldin\opens .
$$
\item\label{item.continuous.function.at.p}
Call a function $f:\ns P\to\ns P'$ \deffont[continuous function at a point]{continuous at $p\in\ns P$} when
$$
\Forall{O'{\in}\opens'}f(p)\in O'\limp \Exists{O_{p,O'}{\in}\opens}p\in O_{p,O'}\land O_{p,O'}\subseteq f^\mone(O') .
$$
In words: $f$ is continuous at $p$ when the inverse image of every open neighbourhood of $f(p)$ contains an open neighbourhood of $p$.
\item
Call a function $f:\ns P\to\ns P'$ \deffont[continuous function on a set]{continuous on $P\subseteq\ns P$} when $f$ is continuous at every $p\in P$.
\end{enumerate}
\end{defn}

\begin{lemm}
\label{lemm.alternative.cont}
Suppose $(\ns P,\opens)$ and $(\ns P',\opens')$ are semitopological spaces (Definition~\ref{defn.semitopology}) and suppose $f:\ns P\to\ns P'$ is a function.
Then the following are equivalent:
\begin{enumerate*}
\item
$f$ is continuous (Definition~\ref{defn.continuity}(\ref{item.continuous.function})).
\item
$f$ is continuous at every $p\in\ns P$ (Definition~\ref{defn.continuity}(\ref{item.continuous.function.at.p})).
\end{enumerate*}
\end{lemm}
\begin{proof}
The top-down implication is immediate, taking $O=f^\mone(O')$.

For the bottom-up implication, given $p$ and an open neighbourhood $O'\ni f(p)$, we write
$$
O=\bigcup\{O_{p,O'}\in\opens \mid p\in\ns P,\ f(p)\in O'\}.
$$
Above, $O_{p,O'}$ is the open neighbourhood of $p$ in the preimage of $O'$, which we know exists by Definition~\ref{defn.continuity}(\ref{item.continuous.function.at.p}).

It is routine to check that $O= f^\mone(O')$, and since this is a union of open sets, it is open. 
\end{proof}

\begin{defn}
\label{defn.locally.constant}
Suppose that:
\begin{itemize*}
\item
$(\ns P,\opens)$ is a semitopology and 
\item
$\tf{Val}$ is a semitopology of values (Definition~\ref{defn.value.assignment}(\ref{item.discrete.semitopology})) and 
\item
$f:\ns P\to \tf{Val}$ is a value assignment (Definition~\ref{defn.value.assignment}(\ref{item.value.assignment}); an assignment of a value to each element in $\ns P$).
\end{itemize*}
Then:
\begin{enumerate*}
\item
Call $f$ \deffont[locally constant at a point]{locally constant at $p\in\ns P$} when there exists $p\in O_p\in\opens$ such that 
$$
\Forall{p'{\in}O_p}f(p)=f(p').
$$
So $f$ is locally constant at $p$ when it is constant on some open neighbourhood $O_p$ of $p$.
\item
Call $f$ \deffont[locally constant on a set]{locally constant} when it is locally constant at every $p\in\ns P$.
\end{enumerate*} 
\end{defn}

\begin{lemm}
\label{lemm.open.lc}
Suppose $(\ns P,\opens)$ is a semitopology and $\tf{Val}$ is a semitopology of values and $f:\ns P\to\tf{Val}$ is a value assignment.
Then the following are equivalent:
\begin{itemize*}
\item
$f$ is locally constant / locally constant at $p\in\ns P$ (Definition~\ref{defn.locally.constant}).
\item
$f$ is continuous / continuous at $p\in\ns P$ (Definition~\ref{defn.continuity}). 
\end{itemize*}
\end{lemm}
\begin{proof}
This is just by pushing around definitions, but we spell it out:
\begin{itemize}
\item
Suppose $f$ is continuous, consider $p\in\ns P$, and write $v=f(p)$.
By our assumptions we know that $f^\mone(v)$ is open, and $p\in f^\mone(v)$.
This is an open neighbourhood $O_p$ on which $f$ is constant, so we are done.
\item
Suppose $f$ is locally constant, consider $p\in\ns P$, and write $v=f(p)$.
By assumption we can find $p\in O_p\in\opens$ on which $f$ is constant, so that $O_p\subseteq f^\mone(v)$.
\qedhere\end{itemize}
\end{proof}

\begin{rmrk}[Continuity = agreement]
\label{rmrk.continuity=consensus}
Lemma~\ref{lemm.open.lc} tells us that
we can view the problem of attaining agreement across an actionable coalition (as discussed in Subsection~\ref{subsect.what.is}) as being the same thing as computing a value assignment that is continuous on that coalition (and possibly elsewhere).

To see why, consider a semitopology $(\ns P, \opens)$ and following the intuitions discussed in Subsection~\ref{subsect.what.is} view points $p\in \ns P$ as \emph{participants}; and view open neighbourhoods $p\in O\in\opens$ as \deffont{actionable coalitions} that include $p$.
Then to say ``$f$ is a value assignment that is continuous at $p$'' is to say that:
\begin{itemize*}
\item
$f$ assigns a value or belief to $p\in\ns P$, and
\item
$p$ is part of a (by Lemma~\ref{lemm.open.lc} continuity) set of peers that agrees with $p$ and (being open) can progress to act on this agreement.
\end{itemize*}
Conceptually and mathematically this reduces the general question 
\begin{quote}
\emph{How can we model collaborative action?} 
\end{quote}
(which, to be fair, has more than one possible answer!) to a more specific research question
\begin{quote}
\emph{Understand continuous value assignments on semitopologies}.
\end{quote}
We then devote ourselves to elaborating (some of) a body of mathematics that we can pull out of this idea.
\end{rmrk}

\jamiesubsection{Neighbourhoods of a point}

Definition~\ref{defn.open.neighbourhood} is a standard notion from topology, and Lemma~\ref{lemm.open.is.open} is a (standard) characterisation of openness, which will be useful later: 

\begin{defn}
\label{defn.open.neighbourhood}
Suppose $(\ns P,\opens)$ is a semitopology and $p\in\ns P$ and $O\in\opens$.
Then call $O$ an \deffont{open neighbourhood} of $p$ when $p\in O$.

In other words: an open set is (by definition) an \emph{open neighbourhood} precisely for the points that it contains.
\end{defn}

\begin{lemm}
\label{lemm.open.is.open}
Suppose $(\ns P,\opens)$ is a semitopology and suppose $P\subseteq\ns P$ is any set of points.
Then the following are equivalent:
\begin{itemize*}
\item
$P\in\opens$.
\item
Every point $p$ in $P$ has an open neighbourhood in $P$. 
\end{itemize*}
In symbols we can write:
$$
\Forall{p{\in}P}\Exists{O{\in}\opens}(p\in O\land O\subseteq P)
\quad\text{if and only if}\quad
P\in\opens
$$
\end{lemm}
\begin{proof}
If $P$ is open then $P$ itself is an open neighbourhood for every point that it contains. 

Conversely, if every $p\in P$ contains some open neighbourhood $p\in O_p \subseteq P$ then $P=\bigcup\{O_p\mid p\in P\}$ and this is open by condition~\ref{semitopology.unions} of Definition~\ref{defn.semitopology}.
\end{proof}

\begin{rmrk}
An initial inspiration for modelling collaborative action using semitopologies, came from noting that the standard topological property described above in Lemma~\ref{lemm.open.is.open}, corresponds to the \emph{quorum sharing} property in \cite[Property~1]{losa:stecbi}; the connection to topological ideas had not been noticed in~\cite{losa:stecbi}.
\end{rmrk}

\jamiesection{Transitive sets \& topens}
\label{sect.transitive.sets}

\jamiesubsection{Some background on sets intersection}

Some notation will be convenient:
\begin{nttn}
\label{nttn.between}
Suppose $X$, $Y$, and $Z$ are sets.
\begin{enumerate*}
\item\label{item.between}
Write 
$$
X\between Y
\quad\text{when}\quad 
X\cap Y\neq\varnothing.
$$
When $X\between Y$ holds then we say (as standard) that $X$ and $Y$ \deffont[intersecting sets $X\between Y$]{intersect}.\index{$X\between Y$ (intersection of sets)}
\item
We may chain the $\between$ notation, writing for example 
$$
X\between Y\between Z
\quad\text{for}\quad
X\between Y\ \land \  Y\between Z
$$
\item
We may write $X\notbetween Y$ for $\neg(X\between Y)$, thus $X\notbetween Y$ when $X\cap Y=\varnothing$.
\end{enumerate*}
\end{nttn}

\begin{rmrk}
\emph{Note on design in Notation~\ref{nttn.between}:}
It is uncontroversial that if $X\neq\varnothing$ and $Y\neq\varnothing$ then $X\between Y$ should hold precisely when $X\cap Y\neq\varnothing$ --- but there is an edge case! 
What truth-value should $X\between Y$ return when $X$ or $Y$ is empty?
\begin{enumerate*}
\item
It might be nice if $X\subseteq Y$ would imply $X\between Y$.
This argues for setting 
$$
(X=\varnothing\lor Y=\varnothing)\limp X\between Y .
$$
\item
It might be nice if $X\between Y$ were monotone on both arguments (i.e. if $X\between Y$ and $X\subseteq X'$ then $X'\between Y$).
This argues for setting 
$$
(X=\varnothing\lor Y=\varnothing)\limp X\notbetween Y .
$$
\item
It might be nice if $X\between X$ always --- after all, should a set \emph{not} intersect itself? --- and this argues for setting 
$$
\varnothing\between\varnothing ,
$$ 
even if we also set $\varnothing\notbetween Y$ for nonempty $Y$. 
\end{enumerate*}
All three choices are defensible, and they are consistent with the following nice property:
$$
X\between Y \limp (X\between X \lor Y\between Y) . 
$$
We choose the second --- if $X$ or $Y$ is empty then $X\notbetween Y$ --- because it gives the simplest definition that $X\between Y$ precisely when $X\cap Y\neq\varnothing$.
\end{rmrk}

We list some elementary properties of $\between$ from Notation~\ref{nttn.between}(\ref{item.between}):
\begin{lemm}
\label{lemm.between.elementary}
\leavevmode
\begin{enumerate*}
\item\label{item.between.nonempty}
$X\between X$ if and only if $X\neq\varnothing$.
\item\label{item.between.symmetric}
$X\between Y$ if and only if $Y\between X$.
\item\label{between.elementary.either.or}
$X\between (Y\cup Z)$ if and only if $(X\between Y) \lor (X\between Z)$.
\item\label{between.subset}
If $X\subseteq X'$ and $X\neq\varnothing$ then $X\between X'$.
\item\label{between.monotone}
Suppose $X\between Y$.
Then $X\subseteq X'$ implies $X'\between Y$, and $Y\subseteq Y'$ implies $X\between Y'$. 
\item\label{between.nonempty}
If $X\between Y$ then $X\neq\varnothing$ and $Y\neq\varnothing$.
\end{enumerate*}
\end{lemm}
\begin{proof}
By facts of sets intersection.
\end{proof}

\jamiesubsection{Transitive open sets and value assignments}

\begin{defn}
\label{defn.transitive}
Suppose $(\ns P,\opens)$ is a semitopology.
Suppose $\atopen\subseteq\ns P$ is any set of points.
\begin{enumerate*}
\item\label{transitive.transitive}
Call $\atopen$ \deffont[transitive set]{transitive} when 
$$
\Forall{O,O'{\in}\opens} O\between \atopen \between O' \limp O\between O'. 
$$
\item\label{transitive.cc}
Call $\atopen$ \deffont[topen set]{topen} when $\atopen$ is nonempty transitive and open.\footnote{%
The empty set is trivially transitive and open, so it would make sense to admit it as a (degenerate) topen.  However, it turns out that we mostly need the notion of `topen' to refer to certain kinds of neighbourhoods of points (we will call them \emph{communities}; see Definition~\ref{defn.tn}).  It is therefore convenient to exclude the empty set from being topen, because while it is the neighbourhood of every point that it contains, it is not a neighbourhood of any point.} 

We may write 
$$
\topens=\{ \atopen\in\opens_{\neq\varnothing} \mid \atopen\text{ is transitive}\} .
$$
\item\label{transitive.max.cc}
Call $S$ a \deffont[maximal topen set]{maximal topen} when $S$ is a topen that is not a subset of any strictly larger topen.\footnote{`Transitive open' $\to$ `topen', like `closed and open' $\to$ `clopen'.

For convenient reference, note that related notions of \emph{strong} transitivity and topen are in Definition~\ref{defn.strongly.transitive}.}
\end{enumerate*}
\end{defn}

Theorem~\ref{thrm.correlated} clarifies why transitivity is interesting: continuous value assignments are constant --- if we think of points as participants, `constant function' here means `in agreement' --- across transitive sets.
\begin{thrm}
\label{thrm.correlated}
Suppose that:
\begin{itemize*}
\item
$(\ns P,\opens)$ is a semitopology.
\item
$\tf{Val}$ is a semitopology of values (a nonempty set with the discrete semitopology; see Definition~\ref{defn.value.assignment}(\ref{item.discrete.semitopology})). 
\item
$f:\ns P\to\tf{Val}$ is a value assignment (Definition~\ref{defn.value.assignment}(\ref{item.value.assignment})). 
\item
$T\subseteq\ns P$ is a transitive set (Definition~\ref{defn.transitive}) --- in particular this will hold if $\atopen$ is topen --- and $p,p'\in T$.
\end{itemize*} 
Then:
\begin{enumerate*}
\item\label{item.correlated.1}
If $f$ is continuous at $p$ and $p'$ then $f(p)=f(p')$.
\item\label{item.correlated.2}
As a corollary, if $f$ is continuous on $\atopen$, then $f$ is constant on $\atopen$.
\end{enumerate*}
In words we can say: 
\begin{quote}
Continuous value assignments are constant across transitive sets.
\end{quote}
\end{thrm}
\begin{proof}
Part~\ref{item.correlated.2} follows from part~\ref{item.correlated.1} since if $f(p)=f(p')$ for \emph{any} $p,p'\in T$, then by definition $f$ is constant on $\atopen$.
So we now just need to prove part~\ref{item.correlated.1} of this result.

Consider $p,p'\in T$.
By continuity on $\atopen$, there exist open neighbourhoods $p\in O\subseteq f^\mone(f(p))$ and $p'\in O'\subseteq f^\mone(f(p'))$.
By construction $O\between \atopen \between O'$ (because $p\in O\cap T$ and $p'\in T\cap O'$).
By transitivity of $\atopen$ it follows that $O\between O'$. 
Thus, there exists $p''\in O\cap O'$, and by construction $f(p) = f(p'') = f(p')$.
\end{proof}

Corollary~\ref{corr.correlated.intersect} is an easy and useful consequence of Theorem~\ref{thrm.correlated}:
\begin{corr}
\label{corr.correlated.intersect}
Suppose that:
\begin{itemize*}
\item
$(\ns P,\opens)$ is a semitopology. 
\item
$f:\ns P\to \tf{Val}$ is a value assignment to some set of values $\tf{Val}$ (Definition~\ref{defn.value.assignment}). 
\item
$f$ is continuous on topen sets $\atopen, \atopen'\in\topens$.
\end{itemize*}
Then 
$$
\atopen\between \atopen'
\quad\text{implies}\quad 
\Forall{p\in\atopen,p'\in\atopen'} f(p)=f(p').
$$
\end{corr}
\begin{proof}
By Theorem~\ref{thrm.correlated} $f$ is constant on $\atopen$ and $\atopen'$.
We assumed that $\atopen$ and $\atopen'$ intersect, and the result follows.
\end{proof}

A converse to Theorem~\ref{thrm.correlated} also holds:
\begin{prop}
\label{prop.correlated.converse}
Suppose that:
\begin{itemize*}
\item
$(\ns P,\opens)$ is a semitopology.
\item
$\tf{Val}$ is a semitopology of values with at least two elements (to exclude a degenerate case that no functions exist, or they exist but there is only one because there is only one value to map to).
\item
$T\subseteq\ns P$ is any set. 
\end{itemize*} 
Then 
\begin{itemize*}
\item
\emph{if} for every $p,p'\in T$ and every value assignment $f:\ns P\to\tf{Val}$, $f$ continuous at $p$ and $p'$ implies $f(p)=f(p')$, 
\item
\emph{then} $\atopen$ is transitive.
\end{itemize*}
\end{prop}
\begin{proof}
We prove the contrapositive. 
Suppose $\atopen$ is not transitive, so there exist $O,O'\in\opens$ such that $O\between \atopen\between O'$ and yet $O\cap O'=\varnothing$.
We choose two distinct values $v\neq v'\in\tf{Val}$ and define $f$ to map any point in $O$ to $v$ and any point in $\ns P\setminus O$ to $v'$.

Choose some $p\in O$ and $p'\in O'$.
It does not matter which, and some such $p$ and $p'$ exist, because $O$ and $O'$ are nonempty by Lemma~\ref{lemm.between.elementary}(\ref{between.nonempty}), since $O\between\atopen$ and $O'\between\atopen$).

We note that $f(p)=v$ and $f(p')=v'$ and $f$ is continuous at $p\in O$ and $p'\in O'\subseteq\ns P\setminus O$, yet $f(p)\neq f(p')$.
\end{proof}

We can sum up what Theorem~\ref{thrm.correlated} and Proposition~\ref{prop.correlated.converse} mean, as follows:
\begin{rmrk}
\label{rmrk.transitive.correlated}
Suppose $(\ns P,\opens)$ is a semitopology and $\tf{Val}$ is a semitopology of values with at least two elements.
Say that a value assignment $f:\ns P\to\tf{Val}$ \deffont[splits (value assignment splits a set)]{splits} a set $T\subseteq\ns P$ when there exist $p,p'\in T$ such that $f$ is continuous at $p$ and $p'$ and $f(p)\neq f(p')$. 
Then Theorem~\ref{thrm.correlated} and Proposition~\ref{prop.correlated.converse} together say in words that: 
\begin{quote}
$T\subseteq\ns P$ is transitive if and only if it cannot be split by a value assignment that is continuous on $T$. 
\end{quote}
Intuitively, transitive sets characterise areas of guaranteed agreement.

This reminds us of a basic result in topology about \emph{connected spaces}~\cite[Chapter~8, section~26]{willard:gent}.
Call a topological space $(\ns T,\opens)$ \deffont[disconnected (semi)topology]{disconnected} when there exist open sets $O,O'\in\opens$ such that $O\cap O'=\varnothing$ (in our notation: $O\notbetween O'$) and $O\cup O'=\ns T$; otherwise call $(\ns T,\opens)$ \deffont[connected (semi)topology]{connected}.
Then $(\ns T,\opens)$ is disconnected if and only if (in our terminology above) it can be split by a value assignment. 
Theorem~\ref{thrm.correlated} and Proposition~\ref{prop.correlated.converse} are not identical to that result, but they are in the same spirit. 
\end{rmrk}

\begin{rmrk}
\label{rmrk.transitive.comment}
The notion of transitive set gives us enough to comment on the two examples in Subsection~\ref{subsect.what.is}.
Recall that we considered:
\begin{enumerate*}
\item
A nonempty finite set $\mathbb E$ with open sets $\opens(\mathbb E)$ (`actionable coalitions') being majority subsets $O\subseteq\mathbb E$.
\item
Integers $\mathbb Z$ with open sets $\opens(\mathbb Z)$ generated by triplets $\{2i,2i\plus 1,2i\plus 2\}$.
\end{enumerate*}
The reader can check that in $(\mathbb E,\opens(\mathbb E))$ \emph{every} set is transitive, because every pair of nonempty open sets intersect; thus, no $T\subseteq\mathbb E$ can be split by a value assignment that is continuous on $T$. 
In contrast, the reader can check that in $(\mathbb Z,\opens(\mathbb Z))$, most sets are not transitive, including (for example) $\{0,4\}$. 
This lack of transitivity reflects an intuitive observation we made in Subsection~\ref{subsect.what.is} that our second example was `not necessarily particularly safe or desirable in practice'; in our more technical language, we can now note that there exists a value assignment that splits $\{0,4\}$, yet is continuous at $0$ and $4$.
What $(\mathbb Z,\opens(\mathbb Z))$ does satisfy is the weaker (but still useful!) safety property that any continuous value assignment that is continuous everywhere, is constant (corresponding to our informal observation that ``\emph{if} all participants do legally progress, then they announce the same value'').\footnote{We can be more precise if we like: e.g. $T$ cannot be split by a value assignment that is continuous on a contiguous segment of $\mathbb Z$ that includes $T$.  Continuity on all of $\mathbb Z$ is one sufficient condition for this, which corresponds (in the language of consensus) to assuming that all participants are correct.  But we digress.}
This reflects a useful intuition, that the topological notion of `continuity at a point', corresponds to an intuition of $p$ as a participant `following the rules'.
\end{rmrk}

\jamiesubsection{Examples and discussion of transitive sets and topens}

We may routinely order sets by subset inclusion; including open sets, topens, closed sets, and so on, and we may talk about maximal, minimal, greatest, and least elements.
We include the (standard) definition for reference: 
\begin{nttn}
\label{nttn.min.max}
Suppose $(\ns P,\leq)$ is a poset.
Then:
\begin{enumerate*}
\item
Call $p\in\ns P$ \deffont[maximal element (in poset)]{maximal} when $\Forall{p'}p{\leq}p'\limp p'=p$ and \deffont[minimal element (in poset)]{minimal} when $\Forall{p'}p'{\leq}p\limp p'=p$.
\item
Call $p\in\ns P$ \deffont[greatest element (in poset)]{greatest} when $\Forall{p}p'\leq p$ and \deffont[least element (in poset)]{least} when $\Forall{p'}p\leq p'$.
\end{enumerate*}
\end{nttn}

\begin{xmpl}[Examples of transitive sets]
\label{xmpl.singleton.transitive}
\leavevmode
\begin{enumerate*}
\item\label{item.singleton.transitive}
$\{p\}$ is transitive, for any single point $p\in\ns P$. 
\item
The empty set $\varnothing$ is (trivially) transitive.
It is not topen because we insist in Definition~\ref{defn.transitive}(\ref{transitive.cc}) that topens are nonempty.
\item
Call a set $P\subseteq\ns P$ \emph{topologically indistinguishable} when (using Notation~\ref{nttn.between}) for every open set $O$, 
$$
P\between O\liff P\subseteq O .
$$ 
It is easy to check that if $P$ is topologically indistinguishable, then it is transitive.
\end{enumerate*} 
\end{xmpl}

\begin{xmpl}[Examples of topens]
\label{xmpl.cc}
\leavevmode
\begin{enumerate*}
\item\label{item.cc.two.regular}
Take $\ns P=\{0, 1, 2\}$, with open sets $\varnothing$, $\ns P$, $\{0\}$, and $\{2\}$. 
This has two maximal topens $\{0\}$ and $\{2\}$  as illustrated in Figure~\ref{fig.012} (top-left diagram). 
\item\label{item.cc.two.regular.b}
Take $\ns P=\{0, 1, 2\}$, with open sets $\varnothing$, $\ns P$, $\{0\}$, $\{0, 1\}$, $\{2\}$, $\{1,2\}$, and $\{0,2\}$. 
This has two maximal topens $\{0\}$ and $\{2\}$, as illustrated in Figure~\ref{fig.012} (top-right diagram). 
\item\label{item.xmpl.cc.3}
Take $\ns P=\{0,1,2,3,4\}$, with open sets generated by $\{0, 1\}$, $\{1\}$, $\{3\}$, and $\{3,4\}$.
This has two maximal topens $\{0,1\}$ and $\{3,4\}$, as illustrated in Figure~\ref{fig.012} (lower-left diagram). 
\item\label{item.xmpl.cc.4}
Take $\ns P=\{0,1,2,\ast\}$, with open sets generated by $\{0\}$, $\{1\}$, $\{2\}$, $\{0, 1,\ast\}$, and $\{1,2,\ast\}$.
This has three maximal topens $\{0\}$, $\{1\}$, and $\{2\}$, as illustrated in Figure~\ref{fig.012} (lower-right diagram). 
\item
Take the all-but-one semitopology from Example~\ref{xmpl.semitopologies}(\ref{item.counterexample.X-x}) on $\mathbb N$: so $\ns P=\mathbb N$ with opens $\varnothing$, $\mathbb N$, and $\mathbb N\setminus \{x\}$ for every $x\in\mathbb N$.
This has a single maximal topen $\mathbb N$.
\item
The semitopology in Figure~\ref{fig.square.diagram} has no topen sets at all ($\varnothing$ is transitive and open, but by definition in Definition~\ref{defn.transitive}(\ref{transitive.cc}) topens have to be nonempty).
\end{enumerate*}
\end{xmpl}

\begin{figure}
\centering
\includegraphics[align=c,width=0.4\columnwidth,trim={50 60 50 120},clip]{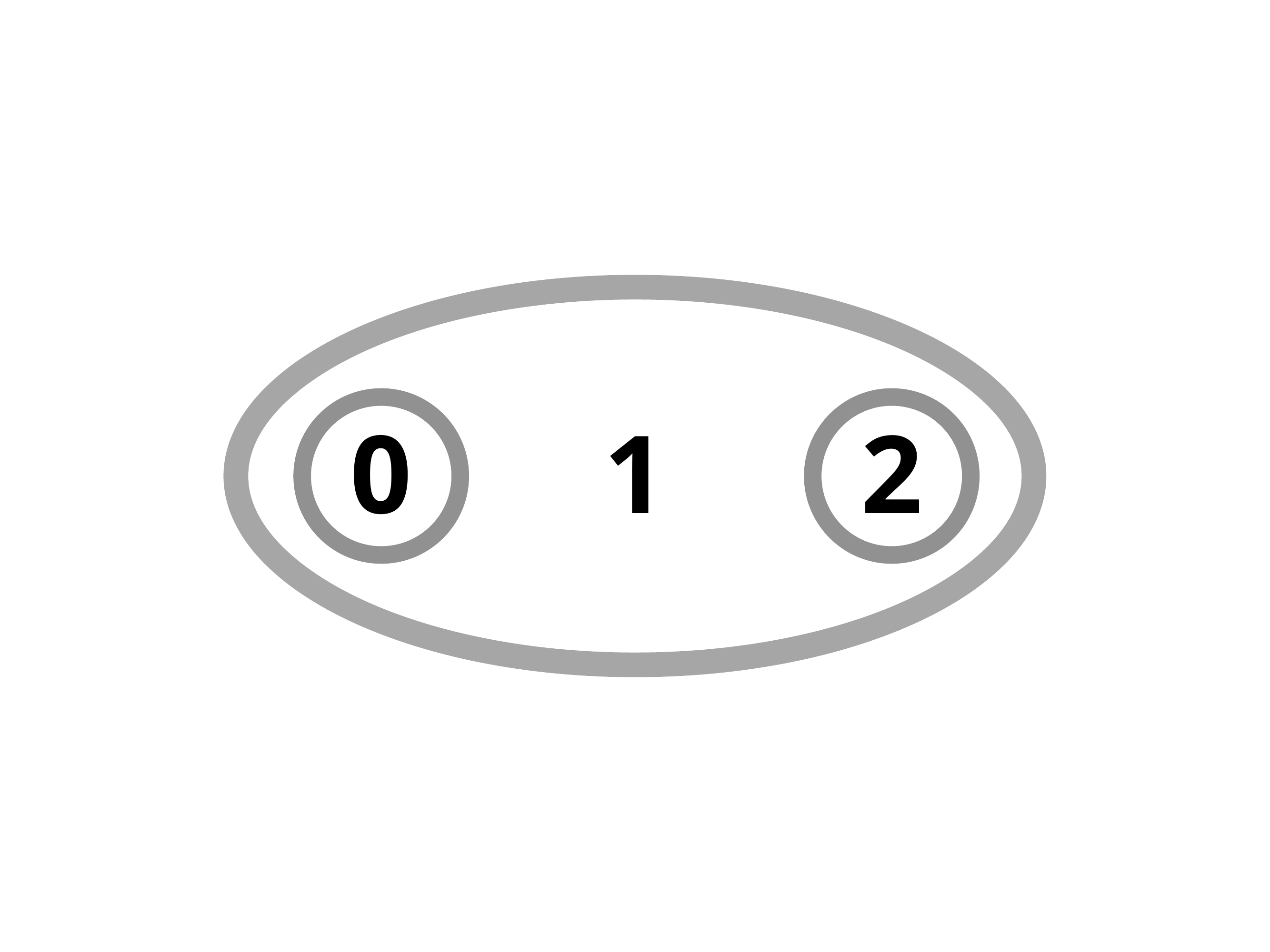}
\includegraphics[align=c,width=0.4\columnwidth,trim={50 60 50 220},clip]{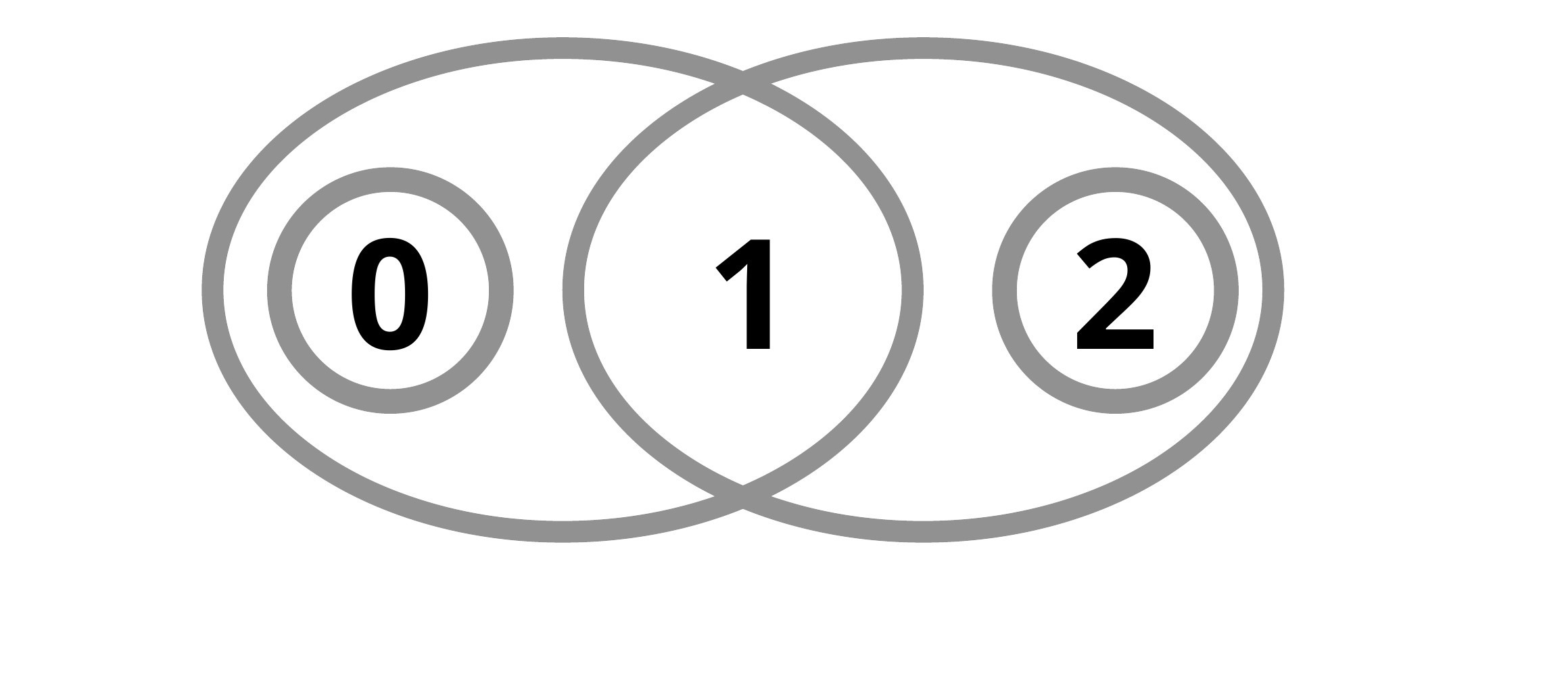}
\\
\includegraphics[align=c,width=0.35\columnwidth,trim={20 20 20 20},clip]{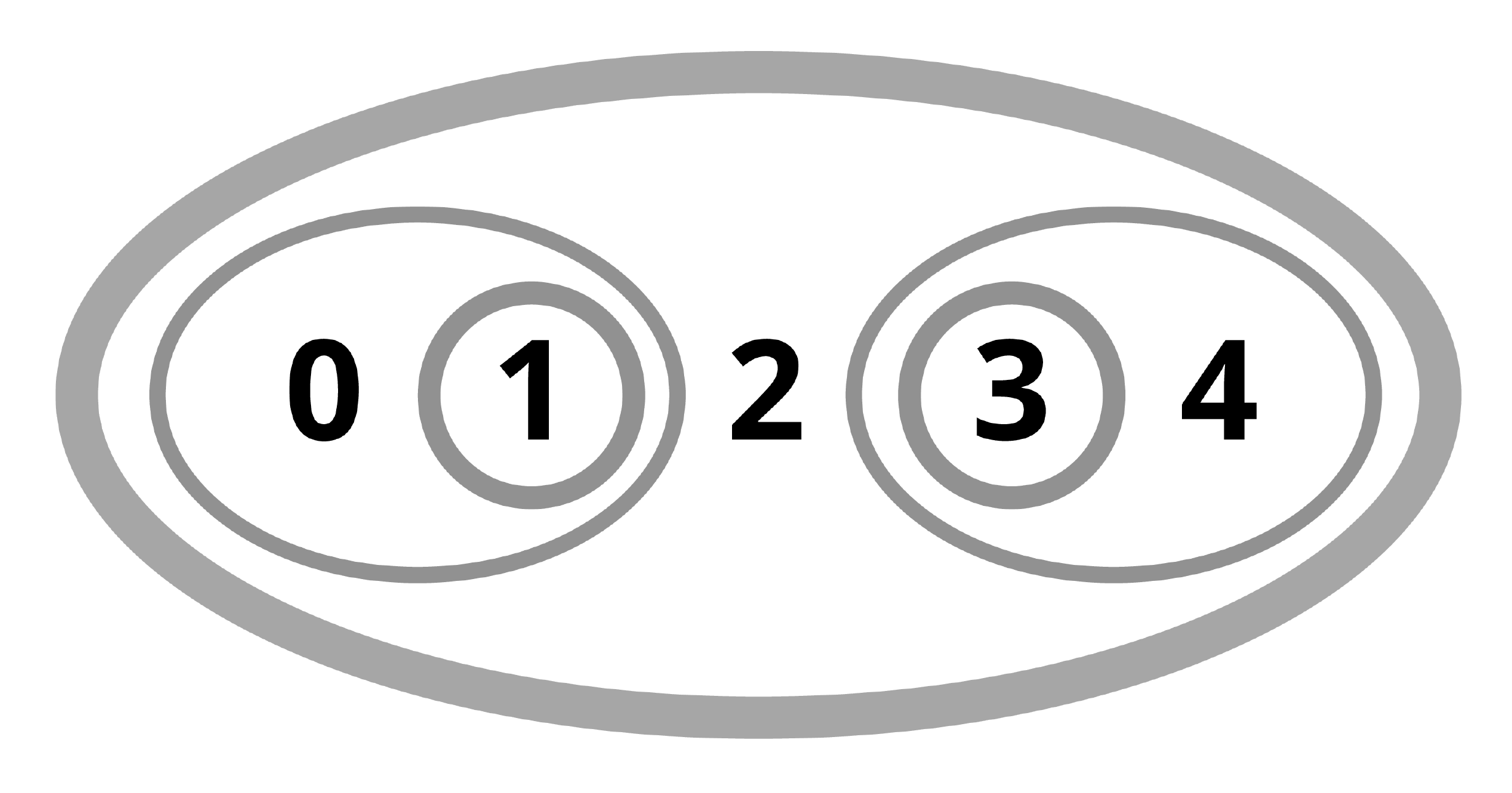}
\quad\  
\includegraphics[align=c,width=0.35\columnwidth,trim={50 20 50 20},clip]{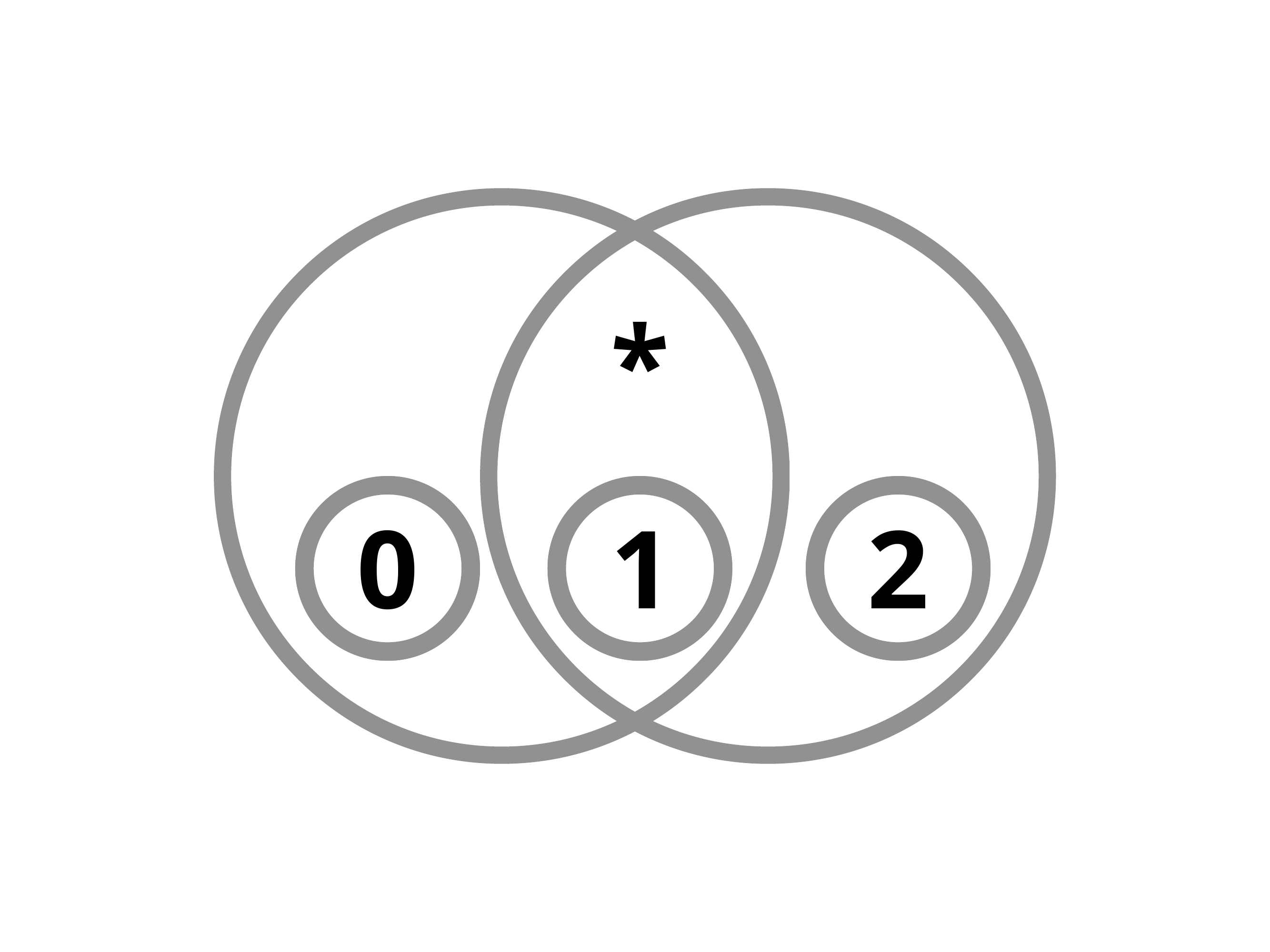}

\begin{flushleft}
\noindent\emph{Here and elsewhere, we might omit open sets that are unions of open sets that are illustrated.  
For example, we explicitly draw the universal open set in the left-hand diagrams above, but not in the right-hand diagrams above.
Meaning is clear and we get cleaner diagrams.
}
\end{flushleft}
\caption{Examples of topens (Example~\ref{xmpl.cc})}
\label{fig.012}
\end{figure}

\begin{rmrk}[Discussion]
We take a moment for a high-level discussion of where we are going.

The semiopologies in Example~\ref{xmpl.cc} invite us to ask what makes these examples different (especially parts~\ref{item.cc.two.regular} and~\ref{item.cc.two.regular.b}).
Clearly they are not equal, but that is a superficial answer in the sense that it is valid just in the world of sets, and it ignores semitopological structure.

For comparison: if we ask what makes $0$ and $1$ different in $\mathbb N$, we could just to say that $0\neq 1$, but this ignores what makes them different \emph{as numbers}.
For more insight, we could note that $0$ is the additive unit whereas $1$ is the multiplicative unit of $\mathbb N$ as a semiring; or that $0$ is a least element and $1$ is the unique atom of $\mathbb N$ as a well-founded poset; or that $1$ is the successor of $0$ of $\mathbb N$ as a well-founded inductive structure. 
Each of these answers gives us more understanding, not only into $0$ and $1$ but also into the structures that can be given to $\mathbb N$ itself. 

So we can ask:
\begin{quote}
\emph{What semitopological property or properties on points can identify the essential nature of the differences between the semitopologies in Example~\ref{xmpl.cc}?}
\end{quote}
There would be some truth to saying that the rest of our investigation is devoted to developing and understanding answers to this question!
In particular, we will shortly define the set of \emph{intertwined points} $\intertwined{p}$ in Definition~\ref{defn.intertwined.points}.
Example~\ref{xmpl.how.different?} will note that $\intertwined{1}=\{0,1,2\}$ in Example~\ref{xmpl.cc}(\ref{item.cc.two.regular}), whereas $\intertwined{1}=\{1\}$ in Example~\ref{xmpl.cc}(\ref{item.cc.two.regular.b}), and $\intertwined{x}=\mathbb N$ for every $x$ in Example~\ref{xmpl.cc}(\ref{item.xmpl.cc.3}).
\end{rmrk}

\jamiesubsection{Closure properties of transitive sets}
\label{subsect.closure.properties.of.tt}

\begin{rmrk}
Transitive sets have some nice closure properties which we treat in this Subsection --- here we mean `closure' in the sense of ``the set of transitive sets is closed under various operations'', and not in the topological sense of `closed sets'.

Topens --- nonempty transitive \emph{open} sets --- will have even better closure properties, which emanate from the requirement in Lemma~\ref{lemm.transitive.transitive} that at least one of the transitive sets $\atopen$ or $\atopen'$ is open. 
See Subsection~\ref{subsect.closure.properties.of.cc}.
\end{rmrk}

\begin{lemm}
\label{lemm.transitive.subset}
Suppose $(\ns P,\opens)$ is a semitopology and $\atopen\subseteq \ns P$. 
Then:
\begin{enumerate*}
\item\label{item.transitive.subset.1}
If $\atopen$ is transitive and $\atopen'\subseteq \atopen$, then $\atopen'$ is transitive.
\item\label{item.transitive.subset.2}
If $\atopen$ is topen and $\varnothing\neq \atopen'\subseteq \atopen$ is nonempty and open, then $\atopen'$ is topen.
\end{enumerate*}
\end{lemm}
\begin{proof}
\leavevmode
\begin{enumerate}
\item
By Definition~\ref{defn.transitive} it suffices to consider open sets $O$ and $O'$ such that $O\between \atopen'\between O'$, and prove that $O\between O'$.
But this is simple: by Lemma~\ref{lemm.between.elementary}(\ref{between.monotone}) $O\between \atopen\between O'$, so $O\between O'$ follows by transitivity of $\atopen$. 
\item
Direct from part~\ref{item.transitive.subset.1} of this result and Definition~\ref{defn.transitive}(\ref{transitive.cc}).
\qedhere\end{enumerate}
\end{proof}

\begin{lemm}
\label{lemm.transitive.transitive}
Suppose that:
\begin{itemize*}
\item
$(\ns P,\opens)$ is a semitopology.
\item
$\atopen,\atopen'\subseteq\ns P$ are transitive.
\item
At least one of $\atopen$ and $\atopen'$ is open.
\end{itemize*}
Then:
\begin{enumerate*}
\item\label{item.transitive.transitive.1} 
$\Forall{O,O'\in\opens}O\between \atopen \between \atopen'\between O' \limp O\between O'$. 
\item\label{item.transitive.transitive.2} 
If $\atopen\between \atopen'$ then $\atopen\cup \atopen'$ is transitive.
\end{enumerate*}
\end{lemm}
\begin{proof}
\leavevmode
\begin{enumerate}
\item
We simplify using Definition~\ref{defn.transitive} and our assumption that one of $\atopen$ and $\atopen'$ is open.
We consider the case that $\atopen'$ is open: 
$$
\begin{array}{r@{\ }l@{\qquad}l}
O\between \atopen\between \atopen'\between O'
\limp&
O\between \atopen' \between O'
&\text{$\atopen$ transitive, $\atopen'$ open}
\\
\limp&
O\between O'
&\text{$\atopen'$ transitive}.
\end{array}
$$
The argument for when $\atopen$ is open, is precisely similar.
\item
Suppose $O\between \atopen\cup \atopen'\between O'$.
By Lemma~\ref{lemm.between.elementary}(\ref{between.elementary.either.or}) (at least) one of the following four possibilities must hold:
$$
O\between \atopen\land \atopen\between O',
\quad
O\between \atopen'\land \atopen\between O',
\quad
O\between \atopen\land \atopen'\between O',
\quad\text{or}\quad
O\between \atopen'\land \atopen'\between O' .
$$
If $O\between \atopen\ \land\ \atopen'\between O'$ then by part~\ref{item.transitive.transitive.1} of this result we have $O\between O'$ as required. 
The other possibilities are no harder.
\qedhere\end{enumerate}
\end{proof}

\begin{defn}[Ascending/descending chain]\leavevmode
\label{defn.ascending.chains}
A \deffont[chain of sets]{chain} of sets $\mathcal X$ is a collection of sets that is totally ordered by subset inclusion $\subseteq$.\footnote{A total order is reflexive, transitive, antisymmetric, and total.}

We may call a chain \deffont[ascending chain of sets]{ascending} or \deffont[descending chain of sets]{descending} if we want to emphasise that we are thinking of the sets as `going up' or `going down'.
\end{defn}

\begin{lemm}
\label{lemm.cac.transitive}
Suppose $(\ns P,\opens)$ is a semitopology and suppose $\mathcal \atopen$ is a chain of transitive sets (Definition~\ref{defn.ascending.chains}).
Then $\bigcup\mathcal \atopen$ is a transitive set.
\end{lemm}
\begin{proof}
Suppose $O\between \bigcup\mathcal \atopen\between O'$.
Then there exist $\atopen,\atopen'\in\mathcal\atopen$ such that $O\between \atopen$ and $\atopen'\between O'$.
But $\mathcal\atopen$ is totally ordered, so either $\atopen\subseteq\atopen'$ or $\atopen\supseteq\atopen'$.
In the former case it follows that $O\between \atopen'\between O'$ so that $O\between O'$ by transitivity of $\atopen'$; the latter case is precisely similar. 
\end{proof}

\jamiesubsection{Closure properties of topens}
\label{subsect.closure.properties.of.cc}

Definition~\ref{defn.connected.set} will be useful in Lemma~\ref{lemm.cc.unions}(\ref{item.clique.of.topens}): 
\begin{defn}
\label{defn.connected.set}
Suppose $(\ns P,\opens)$ is a semitopology.
Call a set of nonempty open sets $\mathcal O\subseteq\opens_{\neq\varnothing}$ a \deffont[clique of sets]{clique} when its elements pairwise intersect.\footnote{%
We call this a \emph{clique}, because if we form the \emph{intersection graph} with nodes elements of $\mathcal O$ and with an (undirected) edge between $O$ and $O'$ when $O\between O'$, then $\mathcal O$ is a clique precisely when its intersection graph is indeed a clique.
See also Definition~\ref{defn.tangled}.
}
In symbols: 
$$
\mathcal O\subseteq\opens\ \text{is a clique}
\quad\text{when}\quad
\Forall{O,O'\in\mathcal O}O\between O'.
$$
Note that if $\mathcal O$ is a clique then every $O\in\mathcal O$ is nonempty, since if $O=\varnothing$ then by $O\notbetween O$ by Lemma~\ref{lemm.between.elementary}(\ref{item.between.nonempty}).
\end{defn}

\begin{lemm}
\label{lemm.cc.unions}
Suppose $(\ns P,\opens)$ is a semitopology.
Then:
\begin{enumerate*}
\item\label{item.intersecting.pair.of.topens}
If $\atopen$ and $\atopen'$ are an intersecting pair of topens (i.e. $\atopen\between \atopen'$), then $\atopen\cup \atopen'$ is topen. 
\item\label{item.clique.of.topens}
If $\mathcal \atopen$ is a clique of topens (Definition~\ref{defn.connected.set}), then $\bigcup\mathcal \atopen$ is topen. 
\item\label{item.chain.of.topens}
If $\mathcal \atopen$ is a nonempty ascending chain of topens then $\bigcup\mathcal \atopen$ is topen.
\end{enumerate*}
\end{lemm}
\begin{proof}
\leavevmode
\begin{enumerate}
\item
$\atopen\cup \atopen'$ is open because by Definition~\ref{defn.semitopology}(\ref{semitopology.unions}) open sets are closed under arbitrary unions, and by Lemma~\ref{lemm.transitive.transitive}(\ref{item.transitive.transitive.2}) $\atopen\cup \atopen'$ is transitive.
\item
$\bigcup\mathcal \atopen$ is open by Definition~\ref{defn.semitopology}(\ref{semitopology.unions}).
Also, if $O\between\bigcup\mathcal \atopen\between O'$ then there exist $\atopen,\atopen'\in\mathcal \atopen$ such that $O\between \atopen$ and $\atopen'\between O'$.
We assumed $\atopen\between \atopen'$, so by Lemma~\ref{lemm.transitive.transitive}(\ref{item.transitive.transitive.1}) (since $\atopen$ and $\atopen'$ are open) we have $O\between O'$ as required. 
\item
Any chain is pairwise intersecting.  We use part~\ref{item.clique.of.topens} of this result.\footnote{We could also use Lemma~\ref{lemm.cac.transitive}.  The chain needs to be nonempty because $\bigcup\varnothing=\varnothing$ and this is open but not topen (= nonempty, transitive, and open).  The reader might ask why Lemma~\ref{lemm.cac.transitive} was not derived directly from Lemma~\ref{lemm.transitive.transitive}(\ref{item.transitive.transitive.2}); this is because (interestingly) Lemma~\ref{lemm.cac.transitive} does not require openness.}
\qedhere
\end{enumerate}
\end{proof}

\begin{corr}
\label{corr.max.cc}
Suppose $(\ns P,\opens)$ is a semitopology.
Then every topen $\atopen$ is contained in a unique maximal topen.
\end{corr}
\begin{proof}
Consider $\mathcal \atopen$ defined by
$$
\mathcal \atopen = \{\atopen\cup \atopen' \mid \atopen'\text{ topen}\land \atopen\between \atopen'\} .
$$
By Lemma~\ref{lemm.cc.unions}(\ref{item.intersecting.pair.of.topens}) this is a set of topens.
By construction they all contain $\atopen$, and by our assumption that $\atopen\neq\varnothing$ they pairwise intersect (since they all contain $\atopen$, at least), so by Lemma~\ref{lemm.cc.unions}(\ref{item.clique.of.topens}) $\bigcup\mathcal \atopen$ is topen.
It is easy to check that this is the unique maximal transitive open set that contains $\atopen$. 
\end{proof}

\begin{thrm}
\label{thrm.topen.partition}
Suppose $(\ns P,\opens)$ is a semitopology.
Then any $P\subseteq \ns P$, and in particular $\ns P$ itself, can be partitioned into:
\begin{itemize*}
\item
Some disjoint collection of maximal topens.
\item
A set of other points, which are not contained in any topen.
\end{itemize*}
\end{thrm}
\begin{proof}
Routine from Corollary~\ref{corr.max.cc}.
\end{proof}

\begin{rmrk}
\label{rmrk.forward}
\label{rmrk.partition}
It may be useful to put Theorem~\ref{thrm.topen.partition} in the context of the terminology, results, and examples that will follow below. 
We will have Definition~\ref{defn.tn}(\ref{item.regular.point}\&\ref{item.irregular.point}) and Theorem~\ref{thrm.max.cc.char}.
These will allow us to call a point $p$ contained in some maximal topen $\atopen$ \emph{regular}, to call the maximal topen $\atopen$ of a regular point its \emph{community}, and a point that is not contained in any topen \emph{irregular}.
Then Theorem~\ref{thrm.topen.partition} says that a semitopology $\ns P$ can be partitioned into:
\begin{itemize*}
\item
Disjoint maximal communities of regular points which, in a sense made formal in Theorem~\ref{thrm.correlated}, are a coalition acting together --- and
\item
a set of irregular points, which are in no community and so are not members of any coalition.
\end{itemize*} 
We give examples in Example~\ref{xmpl.cc} and Figure~\ref{fig.012}, and we will see more elaborate examples below (see in particular the collection in Example~\ref{xmpl.two.topen.examples}). 

In the special case that the entire space consists of a single topen community, there are no irregular points and all participants are guaranteed to agree, where algorithms succeed.
For the application of a single blockchain trying to arrive at consensus, this discussion tells us that we want the underlying semitopology to consist of a single topen, because this means that all participants are guaranteed to agree, where algorithms succeed.
A semitopology that consists of a single topen set is precisely one all of whose open sets intersect, and the reader familiar with literature on quorum systems (for example~\cite{losa:stecbi}) will recognise this as corresponding to the \emph{quorum intersection property}. 
\end{rmrk}

\jamiesubsection{Intertwined points} 
\label{subsect.intertwined.points}

\jamiesubsubsection{The basic definition, and some lemmas}

\begin{defn}
\label{defn.intertwined.points}
Suppose $(\ns P,\opens)$ is a semitopology and $p,p'\in\ns P$.
\begin{enumerate*}
\item\label{item.p.intertwinedwith.p'}
Call $p$ and $p'$ \deffont[intertwined (two points $p\intertwinedwith p'$)]{intertwined} when $\{p,p'\}$ is transitive.\index{$p\intertwinedwith p'$ (two intertwined points)}
Unpacking Definition~\ref{defn.transitive} this means:
$$
\Forall{O,O'{\in}\opens} (p\in O\land p'\in O') \limp O\between O' .
$$ 
By a mild abuse of notation, write 
$$
p\intertwinedwith p' \quad \text{when}\quad \text{$p$ and $p'$ are intertwined}.
$$
\item\label{intertwined.defn}
Define $\intertwined{p}$\index{intertwined of $p$ ($\intertwined{p}$)}\index{$\intertwined{p}$ (points intertwined with a point $p$)} (read `intertwined of $p$') to be the set of points intertwined with $p$.
In symbols: 
$$
\intertwined{p}=\{p'\in\ns P \mid p\intertwinedwith p'\} .
$$
\end{enumerate*}
\end{defn}

\begin{xmpl}
\label{xmpl.how.different?}
We return to the examples in Example~\ref{xmpl.cc}.  
There we note that:
\begin{enumerate*}
\item
$\intertwined{1}=\{0,1,2\}$ and $\intertwined{0}=\{0,1\}$ and $\intertwined{2}=\{1,2\}$.
\item
$\intertwined{1}=\{1\}$ and $\intertwined{0}=\{0\}$ and $\intertwined{2}=\{2\}$.
\item
$\intertwined{0}=\intertwined{1}=\{0,1,2\}$ and $\intertwined{3}=\intertwined{4}=\{2,3,4\}$ and $\intertwined{2}=\ns P$.
\item
$\intertwined{0}=\{0\}$ and $\intertwined{1}=\intertwined{\ast}=\{1,\ast\}$ and $\intertwined{2}=\{2\}$. 
\item
$\intertwined{x}=\ns P$ for every $x$. 
\item
$\intertwined{x}=\{x\}$ for every $x$. 
\end{enumerate*}
\end{xmpl}

Here is one reason to care about intertwined points; a value assignment is constant on a pair of intertwined points, where it is continuous:
\begin{lemm}
\label{lemm.intertwined.correlated}
Suppose $\tf{Val}$ is a semitopology of values and $f:\ns P\to\tf{Val}$ is a value assignment (Definition~\ref{defn.value.assignment})
and $p,p'\in\ns P$ and $p\between p'$.
Then if $f$ is continuous at $p$ and $p'$, then $f(p)=f(p')$.
\end{lemm}
\begin{proof}
$\{p,p'\}$ is transitive by Definition~\ref{defn.intertwined.points}(\ref{item.p.intertwinedwith.p'}).
we use Theorem~\ref{thrm.correlated}.
\end{proof}

We might suppose that being intertwined is transitive.
Lemma~\ref{lemm.intertwined.not.transitive} shows that this is not necessarily the case (the case when $\between$ \emph{is} transitive at $p$ is an important well-behavedness property, which we will call being \emph{unconflicted}; see Subsection~\ref{subsect.reg.tra.int} and Definition~\ref{defn.conflicted}):
\begin{lemm}
\label{lemm.intertwined.not.transitive}
Suppose $(\ns P,\opens)$ is a semitopology.
Then:
\begin{enumerate*}
\item
The `is intertwined' relation $\between$ is reflexive and symmetric. 
\item
$\between$ is not necessarily transitive.
That is: $p'\intertwinedwith p\intertwinedwith p''$ does not necessarily imply $p'\intertwinedwith p''$.
\end{enumerate*}
\end{lemm}
\begin{proof}
Reflexivity and symmetry are clear from Definition~\ref{defn.intertwined.points}(\ref{item.p.intertwinedwith.p'}) and Lemma~\ref{lemm.between.elementary}(\ref{between.elementary.either.or}).

To show that transitivity need not hold, it suffices to provide a counterexample.
The semitopology from Example~\ref{xmpl.cc}(\ref{item.cc.two.regular}) (illustrated in Figure~\ref{fig.012}, top-left diagram) will do.
Take 
$$
\ns P=\{0,1,2\}
\quad\text{and}\quad
\opens=\{\varnothing,\ns P,\{0\},\{2\}\}.
$$
Then 
$$
0\between 1
\ \ \text{and}\ \ 1\between 2,
\quad\text{but}\quad
\neg(0\between 2).
$$
\end{proof}

We conclude with an easy observation:
\begin{nttn}
\label{nttn.intertwined.space}
Suppose $(\ns P,\opens)$ is a semitopology.
Call $\ns P$ \deffont[intertwined (a set $\ns P$)]{intertwined} when 
$$
\Forall{p,p'\in\ns P}p\intertwinedwith p'.
$$
In words: $\ns P$ is intertwined when all of its points are pairwise intertwined.
\end{nttn}

Lemma~\ref{lemm.intertwined.space} will be useful later, notably for Lemma~\ref{lemm.intertwined.space.regular}:
\begin{lemm}
\label{lemm.intertwined.space}
Suppose $(\ns P,\opens)$ is a semitopology.
Then the following conditions are equivalent:
\begin{enumerate*}
\item\label{item.intertwined.space.P}
$\ns P$ is an intertwined space.
\item\label{item.intertwined.space.P.transitive}
$\ns P$ is a transitive set in the sense of Definition~\ref{defn.transitive}(\ref{transitive.transitive}).
\item
All nonempty open sets intersect.
\item
Every nonempty open set is topen.
\end{enumerate*}
\end{lemm}
\begin{proof}
Routine by unpacking the definitions.
\end{proof}

\begin{rmrk}
A topologist would call an intertwined space \emph{hyperconnected} (see Definition~\ref{defn.tangled} and the following discussion).
This is also --- modulo closing under arbitrary unions --- what an expert in the classical theory of consensus might call a \emph{quorum system}~\cite{naor:loacaq}.
\end{rmrk}

\jamiesubsubsection{Pointwise characterisation of transitive sets}

\begin{lemm}
\label{lemm.three.transitive}
Suppose $(\ns P,\opens)$ is a semitopology and $\atopen\subseteq\ns P$.
Then the following are equivalent:
\begin{enumerate*}
\item\label{item.three.transitive.1}
$\atopen$ is transitive.
\item\label{item.three.transitive.2}
$p\intertwinedwith p'$ (meaning by Definition~\ref{defn.intertwined.points} that $\{p,p'\}$ is transitive) 
for every $p,p'\in \atopen$.
\end{enumerate*}
\end{lemm}
\begin{proof}
Suppose $\atopen$ is transitive.
Then by Lemma~\ref{lemm.transitive.subset}(\ref{item.transitive.subset.1}), $\{p,p'\}$ is transitive for every $p,p'\in \atopen$.

Suppose $\{p,p'\}$ is transitive for every $p,p'\in \atopen$.
Consider open sets $O$ and $O'$ such that $O\between \atopen\between O'$. 
Choose $p\in O\cap \atopen$ and $p'\in O\cap \atopen'$.
By construction $\{p,p'\}\subseteq \atopen$ so this is transitive.
It follows that $O\between O'$ as required.
\end{proof}

The special case of Lemma~\ref{lemm.three.transitive} where $\atopen$ is an open set will be particularly useful:
\begin{prop}
\label{prop.cc.char}
Suppose $(\ns P,\opens)$ is a semitopology and $\atopen\subseteq\ns P$.
Then the following are equivalent:
\begin{enumerate*}
\item
$\atopen$ is topen.
\item
$\atopen\in\opens_{\neq\varnothing}$ and $\Forall{p,p'{\in}\atopen}p\intertwinedwith p'$.
\end{enumerate*}
In words we can say:
\begin{quote}
A topen is a nonempty open set of intertwined points.
\end{quote}
\end{prop}
\begin{proof}
By Definition~\ref{defn.transitive}(\ref{transitive.cc}), $\atopen$ is topen when it is nonempty, open, and transitive. 
By Lemma~\ref{lemm.three.transitive} this last condition is equivalent to $p\intertwinedwith p'$ for every $p,p'\in \atopen$. 
\end{proof}

\begin{rmrk}[Intertwined as `non-Hausdorff']
\label{rmrk.not.hausdorff}
\leavevmode
\\
\noindent Recall that we call a topological space $(\ns P,\opens)$ \deffont[Hausdorff space]{Hausdorff} (or \deffont[$T_2$ space (Hausdorff condition)]{$T_2$}) when any two points can be separated by pairwise disjoint open sets.
Using the $\between$ symbol from Notation~\ref{nttn.between}, we rephrase the Hausdorff condition as
$$
\Forall{p,p'}p\neq p'\limp \Exists{O,O'}(p\in O\land p'\in O'\land \neg (O\between O')) , 
$$
we can simplify to 
$$
\Forall{p,p'}p\neq p'\limp p\notintertwinedwith p' ,
$$
and thus we simplify the Hausdorff condition just to
\begin{equation}
\label{eq.hausdorff}
\Forall{p}\intertwined{p}=\{p\}.
\end{equation}
Note how distinct $p$ and $p'$ being intertwined is the \emph{opposite} of being Hausdorff: $p\intertwinedwith p'$ when $p'\in\intertwined{p}$, and they \emph{cannot} be separated by pairwise disjoint open sets.
Thus the assertion $p\intertwinedwith p'$ in Proposition~\ref{prop.cc.char} is a negation to the Hausdorff property:
$$
\Exists{p}\intertwined{p}\neq\{p\} .
$$
This is useful because for semitopologies as applied to consensus, 
\begin{itemize*}
\item
being Hausdorff means that the space is separated (which is probably a bad thing, if we are looking for a system with lots of points in consensus), whereas 
\item
being full of intertwined points means 
by Theorem~\ref{thrm.correlated} that the system will (where algorithms succeed) be full of points whose value assignment agrees (which is a good thing).
\end{itemize*}
In the blockchain literature, we say that a blockchain \emph{forks} when it partitions into two sets of participants with incompatible beliefs about the state of the system.
In this light, we can view Theorem~\ref{thrm.correlated} as a result making precise sufficient conditions to ensure that this does not happen. 
\end{rmrk}

\jamiesubsection{Strong topens: topens that are also subspaces}

\jamiesubsubsection{Definition and main result}

Let us take stock and recall that:
\begin{itemize*}
\item
$\atopen$ is \emph{topen} when it is a nonempty open transitive set (Definition~\ref{defn.transitive}).
\item
$\atopen$ is \emph{transitive} when $O\between \atopen \between O'$ implies $O\between O'$ for all $O,O'\in\tf{Opens}$ (Definition~\ref{defn.transitive}). 
\item
$O\between O'$ means that $O\cap O'\neq\varnothing$ (Notation~\ref{nttn.between}). 
\end{itemize*}
But, note above that if $\atopen$ is topen and $O\between \atopen\between O'$ then $O\cap O'$ need not intersect \emph{inside $\atopen$}.
It could be that $O$ and $O'$ intersect outside of $\atopen$ (an example is in the proof Lemma~\ref{lemm.cc.subspaces} below).

Definition~\ref{defn.subspace} spells out a standard topological construction in the language of semitopologies:
\begin{defn}[Subspaces]
\label{defn.subspace}
Suppose $(\ns P,\opens)$ is a semitopology and suppose $\atopen\subseteq\ns P$ is a set of points.
Write $(\atopen,\opens\cap \atopen)$ for the semitopology such that:
\begin{itemize*}
\item
The points are $\atopen$.
\item
The open sets have the form $O\cap \atopen$ for $O\in\opens$.
\end{itemize*}
We say that $(\atopen, \opens\cap \atopen)$ is $\atopen$ with the \deffont{semitopology induced by $(\ns P,\opens)$}.

We may call $(\atopen,\opens\cap \atopen)$ a \deffont{subspace} of $(\ns P,\opens)$, and if the open sets are understood then we may omit mention of them and just write:
\begin{quote}
A subset $\atopen\subseteq\ns P$ is naturally a \deffont{(semitopological) subspace} of $\ns P$.
\end{quote}
\end{defn}

\begin{figure}
\vspace{-1em}
\centering
\subcaptionbox{A topen that is not strong (Lemma~\ref{lemm.cc.subspaces})}{\includegraphics[width=0.4\columnwidth,trim={50 0 50 20},clip]{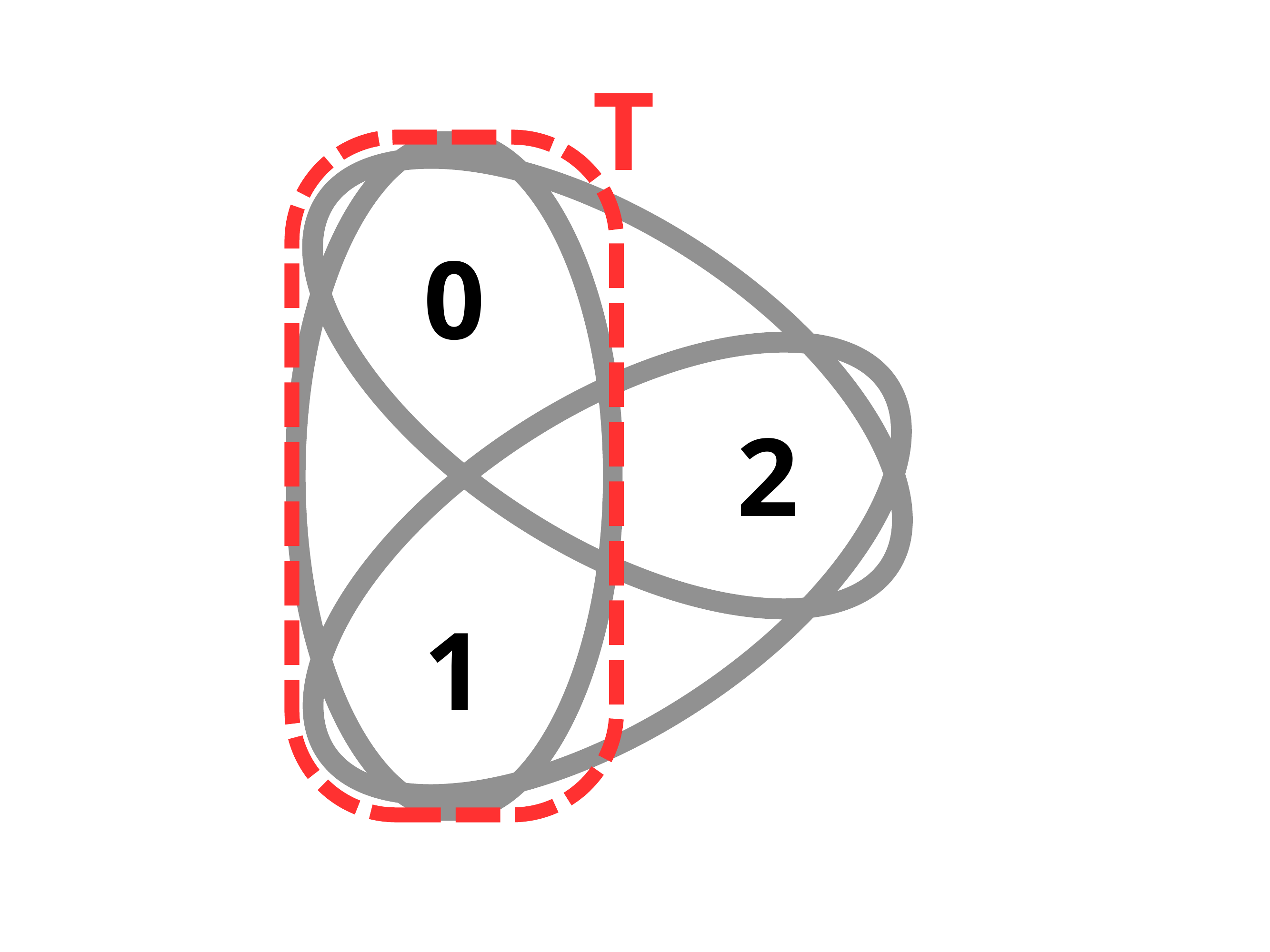}}
\qquad
\subcaptionbox{A transitive set that is not strongly transitive (Lemma~\ref{lemm.strong.is.stronger}(\ref{item.strong.is.stronger.2}))}{\includegraphics[width=0.5\columnwidth,trim={50 30 50 30},clip]{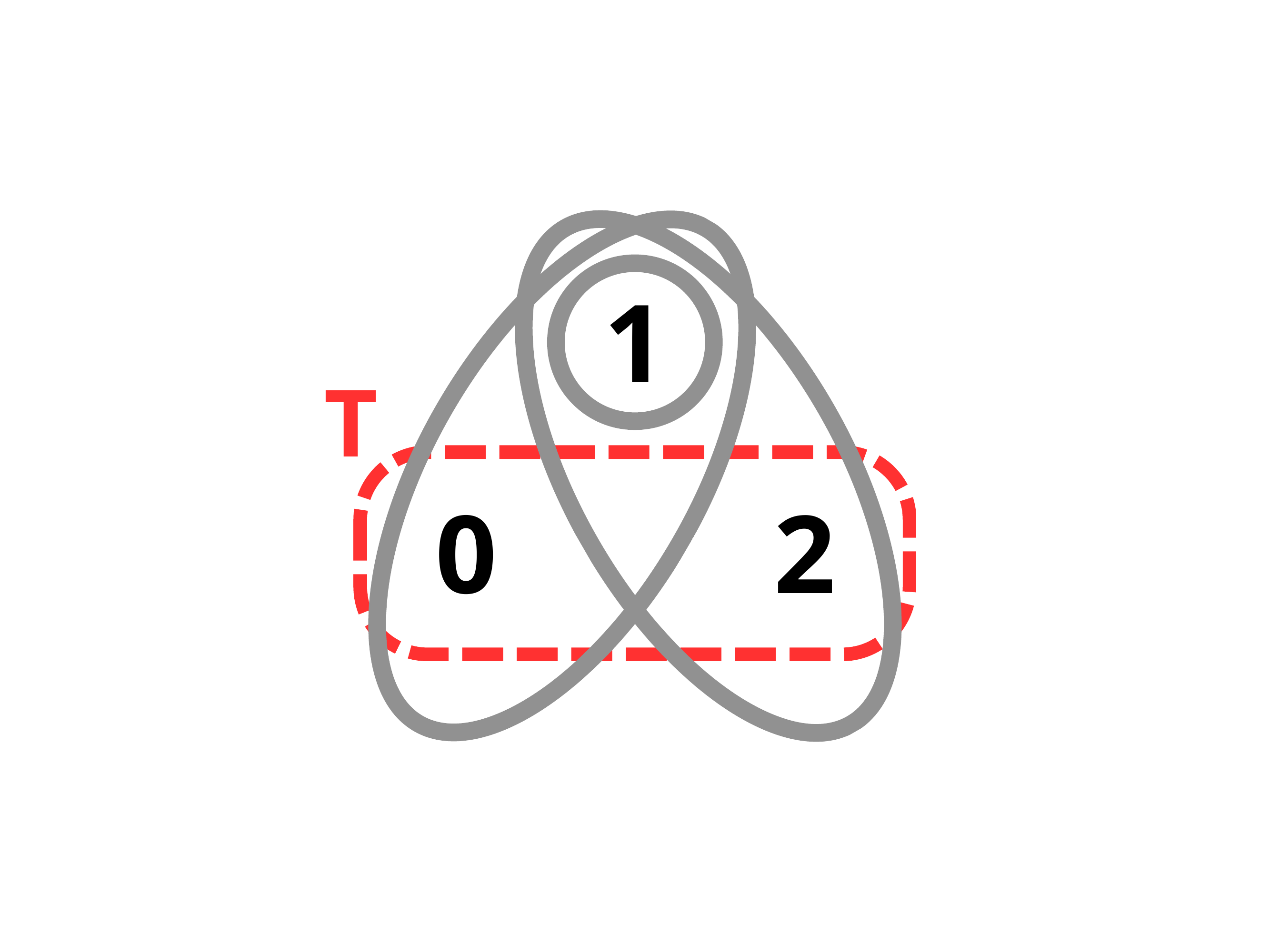}}
\caption{Two counterexamples for (strong) transitivity}
\label{fig.not-strong-topen}
\end{figure}

\begin{lemm}
\label{lemm.cc.subspaces}
The property of being a (maximal) topen is not necessarily closed under taking subspaces.
\end{lemm}
\begin{proof}
It suffices to exhibit a semitopology $(\ns P,\opens)$ and a subset $\atopen\subseteq\ns P$ such that $\atopen$ is topen in $(\ns P,\opens)$ but $\atopen$ is not topen in $(\atopen,\opens\cap \atopen)$.
We set:
$$
\ns P=\{0, 1, 2\}
\qquad
\opens=\{\varnothing,\ \{0, 2\},\ \{1, 2\},\ \{0,1\},\ \ns P\}
\qquad
\atopen=\{0,1\}
$$
as illustrated in Figure~\ref{fig.not-strong-topen} (left-hand diagram).
Now:
\begin{itemize*}
\item
$\atopen$ is topen in $(\ns P,\opens)$, because every open neighbourhood of $0$ --- that is $\{0,2\}$, $\{0,1\}$, and $\ns P$ --- intersects with every open neighbourhood of $1$ --- that is $\{1,2\}$, $\{0,1\}$, and $\ns P$.
\item
$\atopen$ is not topen in $(\atopen,\opens\cap \atopen)$, because $\{0\}$ is an open neighbourhood of $0$ and $\{1\}$ is an open neighbourhood of $1$ and these do not intersect.
\qedhere\end{itemize*}
\end{proof}

Lemma~\ref{lemm.cc.subspaces} motivates the following definitions:

\begin{defn}
\label{defn.betweenY}
Suppose $X$, $Y$, and $Z$ are sets.
Write $X\between_Y Z$, and say that $X$ and $Z$ \deffont[meet in $Y$ ($X\between_Y Z$)]{meet}\index{$X\between_Y Z$ ($X$ and $Z$ intersect in $Y$)} or \deffont[intersect in $Y$ ($X\between_Y Z$)]{intersect in $Y$}, when $(X\cap Y)\between (Z\cap Y)$.
\end{defn}

\begin{lemm}
\label{lemm.betweenY.basic.sets}
Suppose $X$, $Y$, and $Z$ are sets.
Then:
\begin{enumerate*}
\item\label{item.betweenY.basic.sets.1}
The following are equivalent:
$$
X\cap Y\cap Z\neq\varnothing 
\quad\liff\quad
X\between_Y Z
\quad\liff\quad
Y\between_X Z
\quad\liff\quad
X\between_Z Y .
$$
\item\label{item.betweenY.basic.sets.2}
$X\between_Y Y$ if and only if $X \between Y$.
\item\label{item.betweenY.basic.sets.3}
If $X\between_Y Z$ then $X\between Z$.
\end{enumerate*}
\end{lemm}
\begin{proof}
From Definition~\ref{defn.betweenY}, by elementary sets calculations.
\end{proof}

\begin{defn}
\label{defn.strongly.transitive}
Suppose $(\ns P, \opens)$ is a semitopology and recall from Definition~\ref{defn.transitive} the notions of \emph{transitive set} and \emph{topen}.
\begin{enumerate*}
\item\label{item.strongly.transitive}
Call $\atopen\subseteq\ns P$ \deffont[strongly transitive set]{strongly transitive} when
$$
\Forall{O,O'{\in}\opens} O\between \atopen \between O' \limp O\between_\atopen O' . 
$$
\item\label{strong.transitive.cc}
Call $\atopen$ a \deffont{strong topen}\index{strongly topen set} when $\atopen$ is nonempty open and strongly transitive, 
\end{enumerate*}
\end{defn}

\begin{lemm}
\label{lemm.strong.is.stronger}
Suppose $(\ns P, \opens)$ is a semitopology and $\atopen\subseteq\ns P$.
Then:
\begin{enumerate*}
\item\label{item.strong.is.stronger.1}
If $\atopen$ is strongly transitive then it is transitive.
\item\label{item.strong.is.stronger.2}
The reverse implication need not hold (even if $(\ns P,\opens)$ is a topology): it is possible for $\atopen$ to be transitive but not strongly transitive.
\end{enumerate*} 
\end{lemm}
\begin{proof}
We consider each part in turn:
\begin{enumerate}
\item
Suppose $\atopen$ is strongly transitive and suppose $O\between\atopen\between O'$.
By Lemma~\ref{lemm.betweenY.basic.sets}(\ref{item.betweenY.basic.sets.2}) $O\between_\atopen \atopen \between_\atopen O'$.
By strong transitivity $O\between_\atopen O'$.
By Lemma~\ref{lemm.betweenY.basic.sets}(\ref{item.betweenY.basic.sets.3}) $O\between O'$.
Thus $\atopen$ is transitive.
\item
It suffices to provide a counterexample.
This is illustrated in Figure~\ref{fig.not-strong-topen} (right-hand diagram).
We set:
\begin{itemize*}
\item
$\ns P = \{0,1,2\}$, and
\item
$\opens= \{\varnothing,\ \{1\},\ \{0,1\},\ \{1,2\},\ \{0,1,2\}\}$.
\item
We set $\atopen=\{0,2\}$.
\end{itemize*}
We note that $(\ns P,\opens)$ is a topology, and it is easy to check that $\atopen$ is transitive --- we just note that $\{0,1\}\between\atopen\between\{1,2\}$ and $\{0,1\}\between\{1,2\}$.
However, $\atopen$ is not strongly transitive, because $\{0,1\}\cap\{1,2\}=\{1\}\not\subseteq\atopen$.
\qedhere\end{enumerate}
\end{proof}

\begin{prop}
Suppose $(\ns P,\opens)$ is a semitopology and suppose $\atopen\in\opens$.
Then the following are equivalent:
\begin{enumerate*}
\item
$\atopen$ is a strong topen.
\item
$\atopen$ is a topen in $(\atopen,\opens\cap \atopen)$ (Definition~\ref{defn.subspace}).
\end{enumerate*} 
\end{prop}
\begin{proof}
Suppose $\atopen$ is a strong topen; thus $\atopen$ is nonempty, open, and strongly transitive in $(\ns P,\opens)$.
Then by construction $\atopen$ is open in $(\atopen,\opens\cap \atopen)$, and the strong transitivity property of Definition~\ref{defn.strongly.transitive} asserts precisely that $\atopen$ is transitive as a subset of $(\atopen,\opens\cap \atopen)$.

Now suppose $\atopen$ is a topen in $(\atopen,\opens\cap \atopen)$; thus $\atopen$ is nonempty, open, and transitive in $(\atopen,\opens\cap \atopen)$.
Then $\atopen$ is nonempty and by assumption above $\atopen\in\opens$.\footnote{It does not follow from $\atopen$ being open in $(\atopen,\opens\cap \atopen)$ that $\atopen$ is open in $(\ns P,\opens)$, which is why we included an assumption that this holds in the statement of the result.}
Now suppose $O,O'\in\opens$ and $O\between \atopen\between O'$.
Then by Lemma~\ref{lemm.betweenY.basic.sets}(\ref{item.betweenY.basic.sets.2}) $O \between_\atopen \atopen\between_\atopen O'$, so by transitivity of $\atopen$ in $(\atopen,\opens\cap \atopen)$ also $O\between_\atopen O'$, and thus by Lemma~\ref{lemm.betweenY.basic.sets}(\ref{item.betweenY.basic.sets.3}) also $O\between O'$. 
\end{proof}

\jamiesubsubsection{Connection to lattice theory}

There is a notion from order-theory of a \emph{join-irreducible} element (see for example in \cite[Definition~2.42]{priestley:intlo}), and a dual notion of \emph{meet-irreducible} element:
\begin{defn}
Call an element $s$ in a lattice $\mathcal L$ 
\begin{itemize*}
\item
\deffont[join-irreducible element]{join-irreducible} when $s$ is not a bottom element, and $s$ is not a join of two strictly smaller elements: if $x\vee y=s$ then $x=s$, or $y=s$, and
\item
\deffont[meet-irreducible element]{meet-irreducible} when $s$ is not a top element, and $s$ is not a meet of two strictly greater elements: if $x\wedge y=s$ then $x=s$ or $y=s$. 
\end{itemize*}
This definition is typically given for lattices, but it makes just as much sense for semilattices as well.
\end{defn}

\begin{xmpl}
\label{xmpl.meet-irreducible}
\leavevmode
\begin{enumerate*}
\item
Consider the lattice of finite (possibly empty) subsets of $\mathbb N$, with $\mathbb N$ adjoined as a top element.
Then $\mathbb N$ is join-irreducible; $\mathbb N\subseteq\mathbb N$ is not a bottom element, and if $x\cup y=\mathbb N$ then either $x=\mathbb N$ or $y=\mathbb N$.
\item\label{item.final.N}
Consider $\mathbb N$ with the \deffont{final segment semitopology} such that opens are either $\varnothing$ or sets $n_\geq = \{n'\in\mathbb N \mid n'\geq n\}$.

Then $\varnothing$ is meet-irreducible; $\varnothing$ is not a top element, and if $x\cap y=\varnothing$ then either $x=\varnothing$ or $y=\varnothing$.
\item
Consider the integers with the lattice structure in which meet is minimum and join is maximum.
Then every element is join- and meet-irreducible; if $x\vee y=z$ then $x=z$ or $y=z$, and similarly for $x\wedge y$. 
\end{enumerate*}
\end{xmpl}

We spell out how this is related to our notions of transitivity from Definitions~\ref{defn.transitive} and~\ref{defn.strongly.transitive}:
\begin{lemm}
\label{lemm.meet-irreducible}
Suppose $(\ns P,\opens)$ is a semitopology and $\atopen\subseteq\ns P$.
Then: 
\begin{enumerate*}
\item\label{item.meet-irreducible.1}
$\atopen$ is strongly transitive if and only if $\varnothing$ is meet-irreducible in $(\atopen,\opens\cap \atopen)$ (Definition~\ref{defn.subspace}). 
\item
$\atopen$ is transitive if $\varnothing$ is meet-irreducible in $(\atopen,\opens\cap \atopen)$.
\item
If $\atopen$ is transitive it does not necessarily follow that $\varnothing$ is meet-irreducible in $(\atopen,\opens\cap \atopen)$.
\end{enumerate*}
\end{lemm}
\begin{proof}
We reason as follows: 
\begin{enumerate}
\item
$\varnothing$ is meet-irreducible in $(\atopen,\opens\cap \atopen)$ means that $(O\cap \atopen)\cap (O'\cap \atopen)=\varnothing$ implies $O\cap \atopen=\varnothing$ or $O\cap \atopen'=\varnothing$.

$\atopen$ is strongly transitive when (taking the contrapositive in Definition~\ref{defn.strongly.transitive}(\ref{item.strongly.transitive})) $(O\cap \atopen)\cap (\atopen\cap O')=\varnothing$ implies $O\cap \atopen=\varnothing$ or $\atopen\cap O'=\varnothing$.

That these conditions are equivalent follows by straightforward sets manipulations. 
\item
We can use part~\ref{item.meet-irreducible.1} of this result and Lemma~\ref{lemm.strong.is.stronger}(\ref{item.strong.is.stronger.1}), or give a direct argument by sets calculations: if $O\cap O'=\varnothing$ then $(O\cap \atopen)\cap (\atopen\cap O')=\varnothing$ and by meet-irreducibility $O\cap \atopen=\varnothing$ or $\atopen\cap O'=\varnothing$ as required.
\item
Figure~\ref{fig.not-strong-topen} (left-hand diagram) provides a counterexample, taking $\atopen=\{0,1\}$ and $O=\{0,2\}$ and $O'=\{1,2\}$.
Then $(O\cap \atopen)\cap (\atopen\cap O')=\varnothing$ but it is not the case that $O\cap \atopen=\varnothing$ or $O'\cap \atopen=\varnothing$.
\qedhere\end{enumerate}
\end{proof}

\begin{rmrk}
\label{rmrk.imperfect}
The proof of Lemma~\ref{lemm.meet-irreducible} not hard, but the result is interesting for what it says, and also for what it does not say:
\begin{enumerate}
\item
The notion of being a strong topen maps naturally to something in order theory; namely that $\varnothing$ is meet-irreducible in the induced poset $\{O\cap \atopen\mid O\in\opens\}$ which is the set of open sets of the subspace $(\atopen,\opens\cap \atopen)$ of $(\ns P,\opens)$.
\item
However, this mapping is imperfect: the poset is not a lattice, and it is also not a sub-poset of $\opens$ --- even if $\atopen$ is topen.
If $\opens$ were a topology and closed under intersections then we would have a lattice --- but it is precisely the point of difference between semitopologies vs. topologies that open sets need not be closed under intersections. 
\item
Being transitive does not correspond to meet-irreducibility; there is an implication in one direction, but certainly not in the other. 
\end{enumerate}
So, Lemma~\ref{lemm.meet-irreducible} says that (strong) transitivity has a flavour of meet-irreducibility, but in a way that also illustrates --- as did Proposition~\ref{prop.max.topen.min.closed}(\ref{item.max.topen.min.closed.2}) --- how semitopologies are different, because they are not closed under intersections, and have their own behaviour.
\end{rmrk}

\jamiesubsubsection{Topens in topologies}
\label{subsection.topens.in.topologies}

We conclude by briefly looking at what `being topen' means if our semitopology is actually a topology.
We recall a standard definition from topology:
\begin{defn}
\label{defn.tangled}
Suppose $(\ns P,\opens)$ is a semitopology.
Call $\atopen\subseteq\ns P$ \deffont[hyperconnected set]{hyperconnected} when all nonempty open subsets of $\atopen$ intersect.\footnote{Calling this \emph{hyperconnected} is a slight but natural generalisation of the usual definition: in topology, `hyperconnected' is typically used to refer to an entire space rather than a subset of it.  In the case that $\atopen=\ns P$, our definition specialises to the usual one.}
In symbols: 
$$
\Forall{O,O'\in\opens_{\neq\varnothing}} O,O'\subseteq\atopen \limp O\between O' .
$$
\end{defn}

\begin{lemm}
\label{lemm.tran.neosi}
Suppose $(\ns P,\opens)$ is a semitopology.
Then if $\atopen\subseteq\ns P$ is transitive then it is hyperconnected.
\end{lemm}
\begin{proof}
Suppose $\varnothing\neq O,O'\subseteq\atopen$.
Then $O\between\atopen\between O'$ and by transitivity $O\between O'$ as required.
\end{proof}

What is arguably particularly interesting about Lemma~\ref{lemm.tran.neosi} is that its reverse implication does \emph{not} hold, and in quite a strong sense: 
\begin{lemm}
\label{lemm.tran.no.neosi}
Suppose $(\ns P,\opens)$ is a semitopology and $\atopen\subseteq\ns P$. 
Then:
\begin{enumerate*}
\item
$\atopen$ can be hyperconnected but not transitive, even if $(\ns P,\opens)$ is a topology (not just a semitopology).
\item
$\atopen$ can be hyperconnected but not transitive, even if $\atopen$ is an open set.
\end{enumerate*}
\end{lemm}
\begin{proof}
It suffices to provide counterexamples:
\begin{enumerate}
\item
Consider the semitopology illustrated in the lower-left diagram in Figure~\ref{fig.012} (which is a topology), and set $\atopen=\{0,4\}$.
This has no nonempty open subsets so it is trivially hyperconnected.
However, $\atopen$ is not transitive because $\{0,1\}\between \atopen \between \{3,4\}$ yet $\{0,1\}\notbetween\{3,4\}$.
\item
Consider the semitopology illustrated in the top-right diagram in Figure~\ref{fig.012}, and set $\atopen=\{0,1\}$.
This has two nonempty open subsets, $\{0\}$ and $\{0,1\}$, so it is hyperconnected.
However, $\atopen$ is not transitive, because $\{0\}\between \atopen \between \{1,2\}$ yet $\{0\}\notbetween\{1,2\}$.
\qedhere\end{enumerate}
\end{proof}

We know from Lemma~\ref{lemm.strong.is.stronger}(\ref{item.strong.is.stronger.2}) that `transitive' does not imply `strongly transitive' for an arbitrary subset $\atopen\subseteq\ns P$, even in a topology.
When read together with Lemmas~\ref{lemm.tran.neosi} and~\ref{lemm.tran.no.neosi}, this invites the question of what happens when 
\begin{itemize*}
\item
$(\ns P,\opens)$ is a topology, and \emph{also} 
\item
$\atopen$ is an open set.
\end{itemize*}
In this natural special case, strong transitivity, transitivity, and being hyperconnected, all become equivalent: 
\begin{lemm}
\label{lemm.transitive.topology}
Suppose $(\ns P,\opens)$ is a topology and suppose $\atopen\in\opens$ is an open set.
Then the following are equivalent:
\begin{itemize*}
\item
$\atopen$ is a strong topen (Definition~\ref{defn.strongly.transitive}(\ref{strong.transitive.cc})).
\item
$\atopen$ is a topen.
\item
$\atopen$ is hyperconnected.
\end{itemize*}
\end{lemm}
\begin{proof}
We assumed $\atopen$ is open, so the equivalence above can also be thought of as 
\begin{quote}
strongly transitive $\liff$ transitive $\liff$ all nonempty open subsets intersect.
\end{quote}
We prove a chain of implications:
\begin{itemize}
\item
If $\atopen$ is a strong topen then it is a topen by Lemma~\ref{lemm.strong.is.stronger}(\ref{item.strong.is.stronger.1}).
\item
If $\atopen$ is a topen then we use Lemma~\ref{lemm.tran.neosi}.
\item
Suppose $\atopen$ is hyperconnected, so every pair of nonempty open subsets of $\atopen$ intersect; and 
suppose $O,O'\in\opens_{\neq\varnothing}$ and $O\between\atopen\between O'$.
Then also $(O\cap\atopen) \between \atopen \between (O'\cap\atopen)$.
Now $O\cap\atopen$ and $O'\cap\atopen$ are open: because $\atopen$ is open; and $\ns P$ is a topology (not just a semitopology), so intersections of open sets are open.
By transitivity of $\atopen$ we have $O\cap\atopen\between O'\cap\atopen$.
Since $O$ and $O'$ were arbitrary, $\atopen$ is strongly transitive.
\qedhere\end{itemize} 
\end{proof}

\jamiesection{Interiors, communities \& regular points}
\label{sect.regular.points}

\jamiesubsection{Community of a (regular) point}

Definition~\ref{defn.interior} is standard:
\begin{defn}[Open interior]
\label{defn.interior}
Suppose $(\ns P,\opens)$ is a semitopology and $P\subseteq\ns P$.
Define $\interior(P)$ the \deffont{(open) interior of $P$}\index{$\interior(P)$ (open interior)} by
$$
\interior(P)=\bigcup\{ O\in\opens \mid O\subseteq P\} .
$$
\end{defn}

\begin{lemm}
\label{lemm.interior.open}
Suppose $(\ns P,\opens)$ is a semitopology and $P\subseteq\ns P$.
Then $\interior(P)$ from Definition~\ref{defn.interior} is the greatest open subset of $P$.
\end{lemm}
\begin{proof}
Routine by the construction in Definition~\ref{defn.interior} and closure of open sets under unions (Definition~\ref{defn.semitopology}(\ref{semitopology.unions})).
\end{proof}

\begin{corr}
\label{corr.interior.monotone}
Suppose $(\ns P,\opens)$ is a semitopology and $P,P'\subseteq\ns P$.
Then if $P\subseteq P'$ then $\interior(P)\subseteq\interior(P')$.
\end{corr}
\begin{proof}
Routine using Lemma~\ref{lemm.interior.open}.
\end{proof}

\begin{defn}[Community of a point, and regularity]
\label{defn.tn}
Suppose $(\ns P,\opens)$ is a semitopology and $p\in\ns P$.
Then:
\begin{enumerate*}
\item\label{item.tn}
Define $\community(p)$ the \deffont[community of $p$ ($\community(p)$)]{community of $p$}\index{$\community(p)$ (community of a point)} by 
$$
\community(p)=\interior(\intertwined{p}) .
$$
\item\label{item.community.P}
Extend $\community$ to subsets $P\subseteq\ns P$ by taking a sets union:
$$
\community(P) = \bigcup\{\community(p) \mid p\in P\} .
$$
\item\label{item.regular.point}
Call $p$ a \deffont{regular point} when its community is a topen neighbourhood of $p$.
In symbols:
$$
p\text{ is regular}\quad\text{when}\quad p\in\community(p)\in\topens .
$$
\item\label{item.weakly.regular.point}
Call $p$ a \deffont{weakly regular point} when its community is an open (but not necessarily topen) neighbourhood of $p$.
In symbols:
$$
p\text{ is weakly regular}\quad\text{when}\quad p\in\community(p)\in\opens .
$$
\item\label{item.quasiregular.point}
Call $p$ a \deffont{quasiregular point} when its community is nonempty.
In symbols:
$$
p\text{ is quasiregular}\quad\text{when}\quad \varnothing\neq\community(p)\in\opens .
$$
\item\label{item.irregular.point}
If $p$ is not regular then we may call it an \deffont{irregular point}, or just say that it is not regular.
\item\label{item.regular.S}
If $P\subseteq\ns P$ and every $p\in P$ is regular/weakly regular/quasiregular/irregular then we may call $P$ a \deffont{regular/weakly regular/quasiregular/irregular set} respectively (see also Definition~\ref{defn.conflicted}(\ref{item.unconflicted})).
\qedhere\end{enumerate*}
\end{defn}

\begin{rmrk}
\label{rmrk.r.wr.qr}
Lemmas~\ref{lemm.wr.r} and~\ref{lemm.wr.r.no} give an overview of the relationships between the properties in Definition~\ref{defn.tn}.
\end{rmrk}

\begin{lemm}
\label{lemm.wr.r}
Suppose $(\ns P,\opens)$ is a semitopology and $p\in\ns P$.
Then:
\begin{enumerate*}
\item\label{item.r.implies.wr}
If $p$ is regular, then $p$ is weakly regular.
\item\label{item.wr.implies.qr}
If $p$ is weakly regular, then $p$ is quasiregular.
\end{enumerate*}
\end{lemm}
\begin{proof}
We consider each part in turn:
\begin{enumerate}
\item
If $p$ is regular then by Definition~\ref{defn.tn}(\ref{item.regular.point}) $p\in\community(p)\in\topens$, so certainly $p\in\community(p)$ and by Definition~\ref{defn.tn}(\ref{item.weakly.regular.point}) $p$ is weakly regular.
\item
If $p$ is weakly regular then by Definition~\ref{defn.tn}(\ref{item.weakly.regular.point}) $p\in\community(p)\in\opens$, so certainly $\community(p)\neq\varnothing$ and by Definition~\ref{defn.tn}(\ref{item.quasiregular.point}) $p$ is quasiregular.
\qedhere
\end{enumerate}
\end{proof}

\begin{xmpl}
\label{xmpl.wr}
\leavevmode
\begin{enumerate*}
\item
In Figure~\ref{fig.not-strong-topen} (left-hand diagram),\ $0$, $1$, and $2$ are three intertwined points and the entire space $\{0,1,2\}$ consists of a single topen set.
It follows that $0$, $1$, and $2$ are all regular and their community is $\{0,1,2\}$.
\item\label{item.wr.2}
In Figure~\ref{fig.012} (top-left diagram),\ $0$ and $2$ are regular and $1$ is weakly regular but not regular ($1\in\community(1)=\{0,1,2\}$ but $\{0,1,2\}$ is not topen). 
\item\label{item.qr.2}
In Figure~\ref{fig.012} (lower-right diagram),\ $0$, $1$, and $2$ are regular and $\ast$ is quasiregular ($\community(\ast)=\{1\}$).
\item
In Figure~\ref{fig.012} (top-right diagram),\ $0$ and $2$ are regular and $1$ is neither regular, weakly regular, nor quasiregular ($\community(1)=\varnothing$).
\item
In a semitopology of values $(\tf{Val},\powerset(\tf{Val}))$ (Definition~\ref{defn.value.assignment}) every value $v\in\tf{Val}$ is regular, weakly regular, and unconflicted.
\item\label{item.wr.6}
In $\mathbb R$ with its usual topology (which is also a semitopology), every point is unconflicted because the topology is Hausdorff and by Equation~\ref{eq.hausdorff} in Remark~\ref{rmrk.not.hausdorff} this means precisely that $\intertwined{p}=\{p\}$ so $p$ is intertwined just with itself.
Furthermore $p$ is not (quasi/weakly)regular, because $\community(p)=\interior(\intertwined{p})=\varnothing$.
\end{enumerate*} 
\end{xmpl}

\begin{lemm}
\label{lemm.wr.r.no}
Suppose $(\ns P,\opens)$ is a semitopology and $p\in\ns P$.
Then:
\begin{enumerate*}
\item\label{item.wr.r.not.quasiregular}
$p$ might not be quasiregular (i.e. $\community(p)=\varnothing$); thus by Lemma~\ref{lemm.wr.r} it is also not weakly regular and not regular.
\item\label{item.wr.r.no.converse.1}
$p$ might be quasiregular but not weakly regular (i.e. $\community(p)\neq\varnothing$ but $p\notin\community(p)$); and 
\item\label{item.wr.r.no.converse.2}
$p$ might be weakly regular but not regular (i.e. $p\in\community(p)\notin\topens$). 
\end{enumerate*}
\end{lemm}
\begin{proof}
We consider each part in turn:
\begin{enumerate}
\item
Point $0\in\mathbb R$ in Example~\ref{xmpl.wr}(\ref{item.wr.6}) is not quasiregular.
\item
Point $1$ in Example~\ref{xmpl.wr}(\ref{item.wr.2}) (illustrated in Figure~\ref{fig.012}, top-left diagram) is weakly regular ($\community(1)=\{0,1,2\}$) but not regular ($\community(1)$ is open but not topen).
\item
Point $\ast$ in Example~\ref{xmpl.wr}(\ref{item.qr.2}) (illustrated in Figure~\ref{fig.012}, lower-right diagram) is quasiregular ($\community(\ast)=\{1\}$ is nonempty but does not contain $\ast$).
\qedhere
\end{enumerate}
\end{proof}

\begin{lemm}
\label{lemm.intertwined.space.regular}
Suppose $(\ns P,\opens)$ is a semitopology.
Then:
\begin{enumerate*}
\item\label{item.intertwined.space.regular.1}
If all nonempty open sets intersect then $(\ns P,\opens)$ is regular (meaning that every $p\in\ns P$ is regular).
\item\label{item.intertwined.space.regular.2}
The reverse implication need not hold: it is possible for $(\ns P,\opens)$ to be regular but not all open sets intersect (cf. Corollary~\ref{corr.topen.partition.char}).
\end{enumerate*}
\end{lemm}
\begin{proof}
We consider each part in turn:
\begin{enumerate}
\item
By Lemma~\ref{lemm.intertwined.space}(\ref{item.intertwined.space.P.transitive}) $\ns P\in\topens$ (since it is transitive and open).
By Lemma~\ref{lemm.intertwined.space}(\ref{item.intertwined.space.P}) $\intertwined{p}=\ns P$ for every $p\in\ns P$, thus $\community(p)=\interior(\intertwined{p})=\ns P$.
Thus $p\in\community(p)\in\topens$ for every $p\in\ns P$, so $\ns P$ is regular.
\item
It suffices to provide a counterexample.
We take any discrete semitopology with at least two elements; e.g. $(\{0,1\},\powerset(\{0,1\}))$.
Then $\{0\}\notintersectswith\{1\}$, but by Corollary~\ref{corr.when.singleton.topen} $0$ and $1$ are both regular.
\qedhere
\end{enumerate}
\end{proof}

\begin{xmpl}
When we started looking at semitopologies we gave some examples in Example~\ref{xmpl.semitopologies}.
These may seem quite elementary now, but we run through them commenting on which spaces are regular, weakly regular, or quasiregular:
\begin{itemize*}
\item
Any discrete semitopology is regular; topen neighbourhoods are just the singleton sets.
\item
The initial semitopology is regular: it has no topen neighbourhoods, but also no points.
The final semitopology is regular: it has one topen neighbourhood, containing one point.
The trivial topology is regular; it has a single topen neighbourhood that is $\ns P$ itself. 
\item
The supermajority semitopology is regular.
It has one topen neighbourhood containing all of $\ns P$.
\item
The many semitopology is regular if $\ns P$ is finite (because it is equal to the trivial semitopology), and not even quasiregular if $\ns P$ is infinite, because (for infinite $\ns P$) $\intertwined{p}=\varnothing$ for every point.
For example, if $\ns P=\mathbb N$ and $p$ is even and $p'$ is odd, then $\f{evens}=\{2*n \mid n\in\mathbb N\}$ and $\f{odds}=\{2*n\plus 1 \mid n\in\mathbb N\}$ are disjoint open neighbourhoods of $p$ and $p'$ respectively.
\item
The all-but-one semitopology is regular for $\ns P$ having cardinality of $3$ or more, since all points are intertwined so there is a single topen neighbourhood which is the whole space.
If $\ns P$ has cardinality $2$ or $1$ then we have a discrete semitopology (on two points or one point) and these too are regular, with two or one topen neighbourhoods. 
\item
The more-than-one semitopology is not even quasiregular for $\ns P$ having cardinality of $4$ or more.
If $\ns P$ has cardinality $3$ then we get the left-hand topology in Figure~\ref{fig.not-strong-topen}, which is regular.
If $\ns P$ has cardinality $2$ then we get the trivial semitopology, which is regular. 
\item
Take $\ns P=\mathbb R$ (the set of real numbers) and let open sets be generated by intervals of the form $\rightopeninterval{0,r}$ or $\leftopeninterval{\minus r,0}$ for any strictly positive real number $r>0$.
The reader can check that this semitopology is regular.
\item
Any quorum system induces an intertwined semitopology, as outlined in Example~\ref{xmpl.semitopologies}(\ref{item.quorum.system}).
By Lemmas~\ref{lemm.intertwined.space.regular}(\ref{item.intertwined.space.regular.1}) and~\ref{lemm.intertwined.space} this is a regular semitopology, and every nonempty open set is a topen neighbourhood.
\end{itemize*}
\end{xmpl}

\begin{rmrk}
We pause to recap:
\leavevmode
\begin{enumerate}
\item
$\community(p)$ always exists and always is open.
It may or may not be empty, may or may not be topen, and may or may not contain $p$.
\item
When $p\in\community(p)\in\topens$ we call $p$ `regular', which suggests that non-regular behaviour --- $p\notin\community(p)$ and/or $\community(p)\notin\topens$, or even $\community(p)=\varnothing$ --- is `bad behaviour', and being regular `good behaviour'.

But what is this good behaviour that regularity implies? 
Theorem~\ref{thrm.correlated} (continuous value assignments are constant on topens) tells us that a regular $p$ is surrounded by a topen neighbourhood of points $\community(p)=\interior(\intertwined{p})$ that must agree with it under continuous value assignments.
Using our terminology \emph{community} and \emph{regular}, we can say that \emph{the community of a regular $p$ shares its values}.
\item
We can sum up the above intuitively as follows: 
\begin{enumerate*}
\item
We care about transitivity because it implies agreement.
\item
We care about being open, because it implies actionability. 
\item
Thus, a regular point is interesting because it is a participant in a maximal topen neighbourhood and therefore can \emph{i)} come to agreement and \emph{ii)} take action on that agreement. 
\end{enumerate*}
\item
The question then arises how the community of $p$ can be (semi)topologically characterised.
We will explore, notably in Theorem~\ref{thrm.max.cc.char}, Proposition~\ref{prop.views.of.regularity}, and Theorem~\ref{thrm.up.down.char}; see also Remark~\ref{rmrk.arc}.
\end{enumerate}
\end{rmrk} 

\jamiesubsection{Further exploration of (quasi-/weak) regularity and topen sets}

\begin{rmrk}
\label{rmrk.T0-T2}
Recall three common separation axioms from topology:
\begin{enumerate*}
\item
$T_0$: if $p_1\neq p_2$ then there exists some $O\in\opens$ such that $(p_1\in O)\xor (p_2\in O)$, where $\xor$ denotes \emph{exclusive or}.
\item
$T_1$: if $p_1\neq p_2$ then there exist $O_1,O_2\in\opens$ such that $p_i\in O_j \liff i=j$ for $i,j\in\{1,2\}$.
\item
$T_2$, or the \emph{Hausdorff condition}: if $p_1\neq p_2$ then there exist $O_1,O_2\in\opens$ such that $p_i\in O_j \liff i=j$ for $i,j\in\{1,2\}$, and $O_1\cap O_2=\varnothing$.
Cf. the discussion in Remark~\ref{rmrk.not.hausdorff}.
\end{enumerate*}
Even the weakest of the well-behavedness property for semitopologies that we consider in Definition~\ref{defn.tn} --- quasiregularity --- is in some sense strongly opposed to the space being Hausdorff/$T_2$ (though not to being $T_1$), as Lemma~\ref{lemm.quasiregular.hausdorff} makes formal.
\end{rmrk}

\begin{lemm}
\label{lemm.quasiregular.hausdorff}
\leavevmode
\begin{enumerate*}
\item
Every quasiregular Hausdorff semitopology is discrete.

In more detail: if $(\ns P,\opens)$ is a semitopology that is quasiregular (Definition~\ref{defn.tn}(\ref{item.quasiregular.point})) and Hausdorff (equation~\ref{eq.hausdorff} in Remark~\ref{rmrk.not.hausdorff}), then $\opens=\powerset(\ns P)$. 
\item
There exists a (quasi)regular $T_1$ semitopology that is not discrete.
\end{enumerate*} 
\end{lemm}
\begin{proof}
We consider each part in turn:
\begin{enumerate}
\item
By the Hausdorff property, $\intertwined{p}=\{p\}$.
By the quasiregularity property, $\community(p)\neq\varnothing$.
It follows that $\community(p)=\{p\}$.
But by construction in Definition~\ref{defn.tn}(\ref{item.tn}), $\community(p)$ is an open interior.
Thus $\{p\}\in\opens$.
The result follows.
\item
It suffices to provide an example.
We use the left-hand semitopology in Figure~\ref{fig.not-strong-topen}.
Thus $\ns P=\{0,1,2\}$ and $\opens$ is generated by $\{0,1\}$, $\{1,2\}$, and $\{2,0\}$.
All nonempty open sets intersect, so by Lemma~\ref{lemm.intertwined.space.regular}(\ref{item.intertwined.space.regular.1}) $\ns P$ is regular.
It is also $T_1$ (Remark~\ref{rmrk.T0-T2}).
\qedhere\end{enumerate}
\end{proof}
 
Lemma~\ref{lemm.two.intertwined} confirms in a different way that regularity (Definition~\ref{defn.tn}(\ref{item.regular.point})) is non-trivially distinct from weak regularity and quasiregularity:
\begin{lemm}
\label{lemm.two.intertwined}
Suppose $(\ns P,\opens)$ is a semitopology and $p\in\ns P$.
Then:
\begin{enumerate*}
\item\label{item.two.intertwined.1}
$\community(p)\in\opens$.
\item\label{item.two.intertwined.2}
$\community(p)$ is not necessarily topen; equivalently $\community(p)$ is not necessarily transitive.
(More on this later in Subsection~\ref{subsect.irregular}.)
\end{enumerate*}
\end{lemm}
\begin{proof}
$\community(p)$ is open by construction in Definition~\ref{defn.tn}(\ref{item.tn}), since it is an open interior.

For part~\ref{item.two.intertwined.2}, it suffices to provide a counterexample.
We consider the semitopology from Example~\ref{xmpl.cc}(\ref{item.cc.two.regular}) (illustrated in Figure~\ref{fig.012}, top-left diagram). 
We calculate that $\community(1)=\{0,1,2\}$ so that $\community(1)$ is an open neighbourhood of $1$ --- but it is not transitive, and thus not topen, since $\{0\}\cap\{2\}=\varnothing$.

Further checking reveals that $\{0\}$ and $\{2\}$ are two maximal topens within $\community(1)$. 
\end{proof}

So what is $\community(p)$?
We start by characterising $\community(p)$ as the \emph{greatest} topen neighbourhood of $p$, if this exists:
\begin{lemm}
\label{lemm.intertwined.is.the.greatest}
\label{lemm.max.cc.intertwined}
Suppose $(\ns P,\opens)$ is a semitopology and recall from Definition~\ref{defn.tn}(\ref{item.regular.point}) that $p$ is regular when $\community(p)$ is a topen neighbourhood of $p$.
\begin{enumerate*}
\item\label{item.intertwined.is.the.greatest.1}
If $\community(p)$ is a topen neighbourhood of $p$ (i.e. if $p$ is regular) then $\community(p)$ is a maximal topen.
\item\label{item.intertwined.is.the.greatest.2}
If $p\in \atopen\in\topens$ is a maximal topen neighbourhood of $p$ then $\atopen=\community(p)$.
\end{enumerate*}
\end{lemm}
\begin{proof}
\leavevmode
\begin{enumerate}
\item
Since $p$ is regular, by definition, $\community(p)$ is topen and is a neighbourhood of $p$.
It remains to show that $\community(p)$ is a maximal topen.

Suppose $\atopen$ is a topen neighbourhood of $p$; we wish to prove $\atopen\subseteq \community(p)=\interior(\intertwined{p})$.
Since $\atopen$ is open it would suffice to show that $\atopen\subseteq\intertwined{p}$.
By Proposition~\ref{prop.cc.char} $p\intertwinedwith p'$ for every $p'\in \atopen$, and it follows immediately that $\atopen\subseteq\intertwined{p}$.
\item
Suppose $\atopen$ is a maximal topen neighbourhood of $p$.

First, note that $\atopen$ is open, and by Proposition~\ref{prop.cc.char} $\atopen\subseteq\intertwined{p}$, so $\atopen\subseteq\community(p)$.

By assumption $p\in\atopen\cap\community(p)$ and both are topen so by Lemma~\ref{lemm.cc.unions}(\ref{item.intersecting.pair.of.topens}) $\atopen\cup\community(p)$ is topen, and by maximality $\community(p)\subseteq\atopen$.
\qedhere\end{enumerate}
\end{proof}

\begin{rmrk}
\label{rmrk.how.regularity}
We can use Lemma~\ref{lemm.max.cc.intertwined} to characterise regularity in five equivalent ways: see Theorem~\ref{thrm.max.cc.char} and Corollary~\ref{corr.regular.is.regular}.
Other characterisations will follow but will require additional machinery to state (the notion of \emph{closed neighbourhood}; see Definition~\ref{defn.cn}).
See Corollary~\ref{corr.corr.pKp} and Theorem~\ref{thrm.up.down.char}.
\end{rmrk}

\begin{thrm}
\label{thrm.max.cc.char}
Suppose $(\ns P,\opens)$ is a semitopology and $p\in \ns P$.
Then the following are equivalent:
\begin{enumerate*}
\item\label{char.p.regular}
$p$ is regular, or in full: $p\in\community(p)\in\tf{Topen}$.
\item\label{char.Kp.greatest.topen}
$\community(p)$ is the greatest topen neighbourhood of $p$.
\item\label{char.Kp.max.topen}
$\community(p)$ is a maximal topen neighbourhood of $p$.
\item\label{char.max.topen}
$p$ has a maximal topen neighbourhood. 
\item\label{char.some.topen}
$p$ has some topen neighbourhood.
\end{enumerate*}
\end{thrm}
\begin{proof}
We prove a cycle of implications:
\begin{enumerate}
\item
If $\community(p)$ is a topen neighbourhood of $p$ then it is maximal by Lemma~\ref{lemm.intertwined.is.the.greatest}(\ref{item.intertwined.is.the.greatest.1}).
Furthermore this maximal topen neighbourhood of $p$ is necessarily greatest, since if we have two maximal topen neighbourhoods of $p$ then their union is a larger topen neighbourhood of $p$ by Lemma~\ref{lemm.cc.unions}(\ref{item.intersecting.pair.of.topens}) (union of intersecting topens is topen).
\item
If $\intertwined{p}$ is the greatest topen neighbourhood of $p$, then certainly it is a maximal topen neighbourhood of $p$.
\item
If $\intertwined{p}$ is a maximal topen neighbourhood of $p$, then certainly $p$ has a maximal topen neighbourhood.
\item
If $p$ has a maximal topen neighbourhood then certainly $p$ has a topen neighbourhood.
\item
Suppose $p$ has a topen neighbourhood $\atopen$.
By Corollary~\ref{corr.max.cc} we may assume without loss of generality that $\atopen$ is a maximal topen.
We use Lemma~\ref{lemm.max.cc.intertwined}(\ref{item.intertwined.is.the.greatest.2}).
\qedhere\end{enumerate}
\end{proof}

Theorem~\ref{thrm.max.cc.char} has numerous corollaries:
\begin{corr}
\label{corr.when.singleton.topen}
Suppose $(\ns P,\opens)$ is a semitopology and $p\in\ns P$ and $\{p\}\in\opens$.
Then $p$ is regular. 
\end{corr}
\begin{proof}
We noted in Example~\ref{xmpl.singleton.transitive}(\ref{item.singleton.transitive}) that a singleton $\{p\}$ is always transitive, so if $\{p\}$ is also open, then it is topen, so that $p$ has a topen neighbourhood and by Theorem~\ref{thrm.max.cc.char}(\ref{char.some.topen}) $p$ is topen.\footnote{%
It does not follow from $p\in\{p\}\in\topens$ that $\community(p)=\{p\}$: consider $\ns P=\{0,1\}$ and $\opens=\{\varnothing,\{0\},\{0,1\}\}$ and $p=0$; then $\{p\}\in\topens$ yet $\community(p)=\{0,1\}$.}
\end{proof}

\begin{corr}
\label{corr.regular.is.regular}
Suppose $(\ns P,\opens)$ is a semitopology and $p\in\ns P$.
Then the following are equivalent:
\begin{enumerate*}
\item
$p$ is regular.
\item
$p$ is weakly regular and $\community(p)=\community(p')$ for every $p'\in\community(p)$.
\end{enumerate*} 
\end{corr}
\begin{proof}
We prove two implications, using Theorem~\ref{thrm.max.cc.char}:
\begin{itemize}
\item
Suppose $p$ is regular.
By Lemma~\ref{lemm.wr.r}(\ref{item.r.implies.wr}) $p$ is weakly regular.
Now consider $p'\in\community(p)$.
By Theorem~\ref{thrm.max.cc.char} $\community(p)$ is topen, so it is a topen neighbourhood of $p'$. 
By Theorem~\ref{thrm.max.cc.char} $\community(p')$ is a greatest topen neighbourhood of $p'$. 
But by Theorem~\ref{thrm.max.cc.char} $\community(p)$ is also a greatest topen neighbourhood of $p$, and $\community(p)\between\community(p')$ since they both contain $p'$.
By Lemma~\ref{lemm.cc.unions}(\ref{item.intersecting.pair.of.topens}) and maximality, they are equal.
\item
Suppose $p$ is weakly regular and suppose $\community(p)=\community(p')$ for every $p'\in\community(p)$, and consider $p',p''\in\community(p)$.
Then $p'\intertwinedwith p''$ holds, since $p''\in\community(p')=\community(p)$.
By Proposition~\ref{prop.cc.char} $\community(p)$ is topen, and by weak regularity $p\in\community(p)$, so by Theorem~\ref{thrm.max.cc.char} $p$ is regular as required. 
\qedhere\end{itemize}
\end{proof}

\begin{rmrk}
With regards to Corollary~\ref{corr.regular.is.regular}, it might be useful to look at Example~\ref{xmpl.cc}(\ref{item.cc.two.regular.b}) and Figure~\ref{fig.012} (top-right diagram).
In that example the point $1$ is \emph{not} regular, and its community $\{0,1,2\}$ is not a community for $0$ or $2$.
\end{rmrk}

\begin{corr}
\label{corr.p.p'.regular.community}
Suppose $(\ns P,\opens)$ is a semitopology and $p,p'\in\ns P$.
Then if $p$ is regular and $p'\in\community(p)$ then $p'$ is regular and has the same community.
\end{corr}
\begin{proof}
Suppose $p$ is regular --- so by Definition~\ref{defn.tn}(\ref{item.regular.point}) $p\in\community(p)\in\topens$ --- and suppose $p'\in\community(p)$.
Then by Corollary~\ref{corr.regular.is.regular} $\community(p)=\community(p')$, so $p'\in\community(p')\in\topens$ and by Theorem~\ref{thrm.max.cc.char} $p'$ is regular. 
\end{proof}

\begin{corr}
\label{corr.max.topen.char}
Suppose $(\ns P,\opens)$ is a semitopology. 
Then the following are equivalent for $\atopen\subseteq\ns P$:
\begin{itemize*}
\item
$\atopen$ is a maximal topen.
\item
$\atopen\neq\varnothing$ and $\atopen=\community(p)$ for every $p\in \atopen$.
\end{itemize*}
\end{corr}
\begin{proof}
If $\atopen$ is a maximal topen and $p\in\atopen$ then $\atopen$ is a maximal topen neighbourhood of $p$.
By Theorem~\ref{thrm.max.cc.char}(\ref{char.Kp.greatest.topen}\&\ref{char.some.topen}) $\atopen=\community(p)$.

If $\atopen\neq\varnothing$ and $\atopen=\community(p)$ for every $p\in\atopen$,
then $\community(p)=\community(p')$ for every $p'\in\community(p)$ and by Corollary~\ref{corr.regular.is.regular} $p$ is regular, so that by
Definition~\ref{defn.tn}(\ref{item.regular.point}) $\atopen=\community(p)\in\topens$ as required. 
\end{proof}

\jamiesubsection{Intersection and partition properties of regular spaces}
\label{subsect.topen.partitions}

Proposition~\ref{prop.topen.intersect.subset} is useful for consensus in practice.
Suppose we are a regular point $q$ and we have reached consensus with some topen neighbourhood $O\ni q$.
Suppose further that our topen neighbourhood $O$ intersects with the maximal topen neighbourhood $\community(p)$ of some other regular point $p$.
Then Proposition~\ref{prop.topen.intersect.subset} tells us that we were inside $\community(p)$ all along.
See also Remark~\ref{rmrk.gradecast}.
\begin{prop}
\label{prop.topen.intersect.subset}
Suppose $(\ns P,\opens)$ is a semitopology and $p\in\ns P$ is regular and $O\in\topens$ is topen.
Then 
$$
O\between\community(p)
\quad\text{if and only if}\quad
O\subseteq\community(p).
$$
\end{prop}
\begin{proof} 
The right-to-left implication is immediate from Notation~\ref{nttn.between}(\ref{item.between}), given that 
topens are nonempty by Definition~\ref{defn.transitive}(\ref{transitive.cc}).

For the left-to-right implication, suppose $O\between\community(p)$.
By Theorem~\ref{thrm.max.cc.char} $\community(p)$ is a maximal topen, and by Lemma~\ref{lemm.cc.unions}(\ref{item.intersecting.pair.of.topens}) $O\cup\community(p)$ is topen.
Then $O\subseteq\community(p)$ follows by maximality.
\end{proof}

\begin{prop}
\label{prop.community.partition}
Suppose $(\ns P,\opens)$ is a semitopology and suppose $p,p'\in\ns P$ are regular.
Then
$$
\community(p)\between\community(p')
\quad\liff\quad
\community(p)=\community(p')
$$
(See also Corollary~\ref{corr.community.intersects.community}, which considers similar properties for $p$ and $p'$ that are not necessarily regular.)
\end{prop}
\begin{proof}
We prove two implications.
\begin{itemize}
\item
Suppose there exists $p''\in\community(p)\cap\community(p')$.
By Corollary~\ref{corr.p.p'.regular.community} ($p''$ is regular and) $\community(p)=\community(p'')=\community(p')$.
\item
Suppose $\community(p)=\community(p')$.
By assumption $p\in\community(p)$, so $p\in\community(p')$.
Thus $p\in\community(p)\cap\community(p')$.
\qedhere\end{itemize}
\end{proof}

Corollary~\ref{corr.topen.partition.char} is a simple characterisation of regular semitopological spaces (it is also a kind of continuation to Lemma~\ref{lemm.intertwined.space.regular}(\ref{item.intertwined.space.regular.2})):
\begin{corr}
\label{corr.topen.partition.char}
Suppose $(\ns P,\opens)$ is a semitopology.
Then the following are equivalent:
\begin{enumerate*}
\item\label{item.topen.partition.char.1}
$(\ns P,\opens)$ is regular.
\item\label{item.topen.partition.char.2}
$\ns P$ partitions into topen sets: there exists some set of topen sets $\mathcal T$ such that $\atopen\notbetween\atopen'$ for every $\atopen,\atopen'\in\mathcal T$ and $\ns P=\bigcup\mathcal T$.
\item\label{item.topen.partition.char.3}
Every $X\subseteq\ns P$ has a cover of topen sets: there exists some set of topen sets $\mathcal T$ such that $X\subseteq\bigcup\mathcal T$.
\end{enumerate*}
\end{corr}
\begin{proof}
The proof is routine from the machinery that we already have.
We prove equivalence of parts~\ref{item.topen.partition.char.1} and~\ref{item.topen.partition.char.2}:
\begin{enumerate}
\item
Suppose $(\ns P,\opens)$ is regular, meaning by Definition~\ref{defn.tn}(\ref{item.regular.S}\&\ref{item.regular.point}) that $p\in\community(p)\in\topens$ for every $p\in\ns P$.
We set $\mathcal T=\{\community(p) \mid p\in\ns P\}$.
By assumption this covers $\ns P$ in topens, and by Proposition~\ref{prop.community.partition} the cover is a partition. 
\item
Suppose $\mathcal T$ is a topen partition of $\ns P$.
By definition for every point $p$ there exists $T\in\mathcal T$ such that $p\in T$ and so $p$ has a topen neighbourhood.
By Theorem~\ref{thrm.max.cc.char}(\ref{char.some.topen}\&\ref{char.p.regular}) $p$ is regular.
\end{enumerate}
We prove equivalence of parts~\ref{item.topen.partition.char.2} and~\ref{item.topen.partition.char.3}:
\begin{enumerate}
\item
Suppose $\mathcal T$ is a topen partition of $\ns P$, and suppose $X\subseteq\mathcal P$.
Then trivially $X\subseteq\bigcup\mathcal T$.
\item
Suppose every $X\subseteq\ns P$ has a cover of topen sets.
Then $\ns P$ has a cover of topen sets; write it $\mathcal T$.
By Corollary~\ref{corr.max.cc} we may assume without loss of generality that $\mathcal T$ is a partition, and we are done.
\qedhere\end{enumerate} 
\end{proof}

\begin{rmrk}
\label{rmrk.the.moral}
The moral we take from the results and examples above (and those to follow) is that the world we are entering has rather different well-behavedness criteria than those familiar from the study of typical Hausdorff topologies like $\mathbb R$.
Put crudely: 
\begin{enumerate*}
\item
`Bad' spaces are spaces that are not regular.

$\mathbb R$ with its usual topology (which is also a semitopology) is an example of a `bad' semitopology; it is not even quasiregular.
\item
`Good' spaces are spaces that are regular.

The supermajority and all-but-one semitopologies from Example~\ref{xmpl.semitopologies}(\ref{item.supermajority}\&\ref{item.counterexample.X-x}) are typical examples of `good' semitopologies; both are intertwined spaces (Notation~\ref{nttn.intertwined.space}).
\item
Corollary~\ref{corr.topen.partition.char} shows that the `good' spaces are just the (disjoint, possibly infinite) unions of intertwined spaces.
\end{enumerate*}
\end{rmrk}

\jamiesubsection{Examples of communities and (ir)regular points}
\label{subsect.irregular}

By Definition~\ref{defn.tn} a point $p$ is regular when its community is a topen neighbourhood.
Then a point is \emph{not} regular when its community is \emph{not} a topen neighbourhood of $p$. 
We saw one example of this in Lemma~\ref{lemm.two.intertwined}.
In this subsection we take a moment to investigate the possible behaviour in more detail.

\begin{xmpl}
\label{xmpl.p.not.regular}
\leavevmode
\begin{enumerate}
\item\label{item.p.not.regular.R}
We noted in Example~\ref{xmpl.p.not.regular}(\ref{item.wr.6}) and Lemma~\ref{lemm.wr.r.no}(\ref{item.wr.r.not.quasiregular}) that for $\mathbb R$ the real numbers with its usual topology, every $p\in\mathbb R$ is not regular. 
Then
$\intertwined{x}=\{x\}$ and $\community(x)=\varnothing$ for every $x\in\mathbb R$.
\item\label{item.p.not.regular.012}
We continue the semitopology from Example~\ref{xmpl.cc}(\ref{item.cc.two.regular}) (illustrated in Figure~\ref{fig.012}, top-left diagram), as used in Lemma~\ref{lemm.two.intertwined}:
\begin{itemize*}
\item
$\ns P=\{0,1,2\}$.
\item
$\opens$ is generated by $\{0\}$ and $\{2\}$. 
\end{itemize*}
Then:
\begin{itemize*}
\item
$\intertwined{0}=\{0,1\}$ and $\community(0)=\interior(\intertwined{0})=\{0\}$. 
\item
$\intertwined{2}=\{1,2\}$ and $\community(2)=\interior(\intertwined{2})=\{2\}$. 
\item
$\intertwined{1}=\{0,1,2\}$ and $\community(1)=\{0,1,2\}$. 
\end{itemize*}
\item\label{item.point.not.regular.but.community.is.topen}\label{item.p.not.regular.01234}
We take, as illustrated in Figure~\ref{fig.irregular} (left-hand diagram):
\begin{itemize*}
\item
$\ns P=\{0,1,2,3,4\}$.
\item
$\opens$ is generated by $\{1,2\}$, $\{0,1,3\}$, $\{0,2,4\}$, $\{3\}$, and $\{4\}$.
\end{itemize*}
Then:
\begin{itemize*}
\item
$\intertwined{x}=\{0,1,2\}$ and $\community(x)=\interior(\intertwined{x})=\{1,2\}$ for $x\in\{0,1,2\}$.
\item
$\intertwined{x}=\{x\}=\community(x)$ for $x\in\{3,4\}$.
\end{itemize*}
\item\label{item.p.not.regular.01234b}
We take, as illustrated in Figure~\ref{fig.irregular} (right-hand diagram):
\begin{itemize*}
\item
$\ns P=\{0,1,2,3,4\}$.
\item
$\opens$ is generated by $\{1\}$, $\{2\}$, $\{3\}$, $\{4\}$, $\{0, 1, 2, 3\}$, and $\{0, 1, 2, 4\}$. 
\end{itemize*}
Then:
\begin{itemize*}
\item
$\intertwined{0}=\{0,1,2\}$ and $\community(0)=\{1,2\}$.
\item
$\community(0)$ is not transitive and consists of two distinct topens $\{1\}$ and $\{2\}$.
\item
$0\notin\community(0)$. 
\end{itemize*}
See Remark~\ref{rmrk.indeed.two.closed.neighbourhoods} for further discussion of this example.
\item
The reader can also look ahead to Example~\ref{xmpl.two.topen.examples}.
In Example~\ref{xmpl.two.topen.examples}(\ref{item.two.topen.examples.1}), every point $p$ is regular and $\community(p)=\mathbb Q^2$.
In Example~\ref{xmpl.two.topen.examples}(\ref{item.two.topen.examples.2}), no point $p$ is regular and $\community(p)=\varnothing\subseteq\mathbb Q^2$.
\end{enumerate}
\end{xmpl}

\begin{figure}
\vspace{-1em}
\centering
\includegraphics[width=0.35\columnwidth]{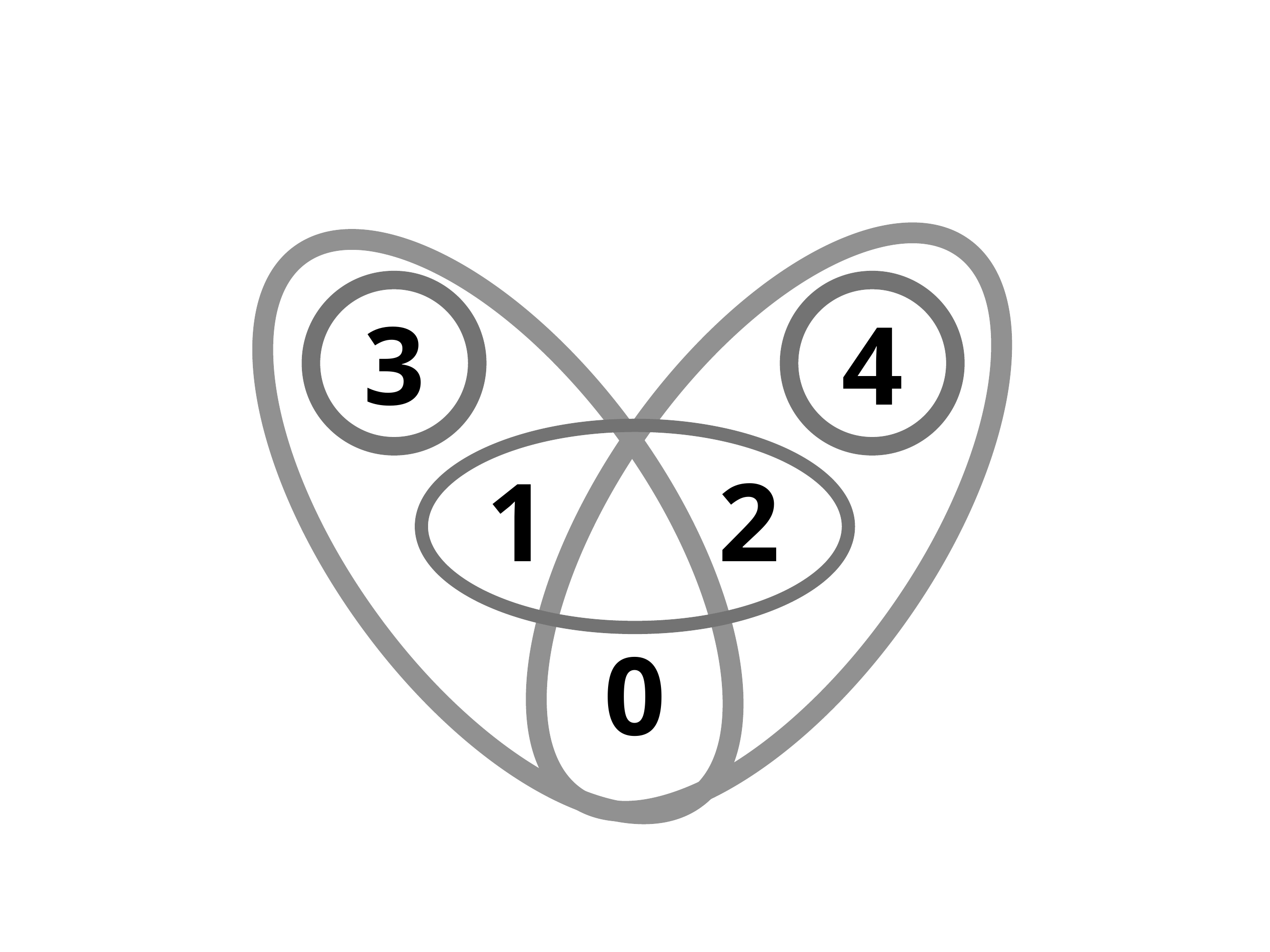}
\includegraphics[width=0.31\columnwidth]{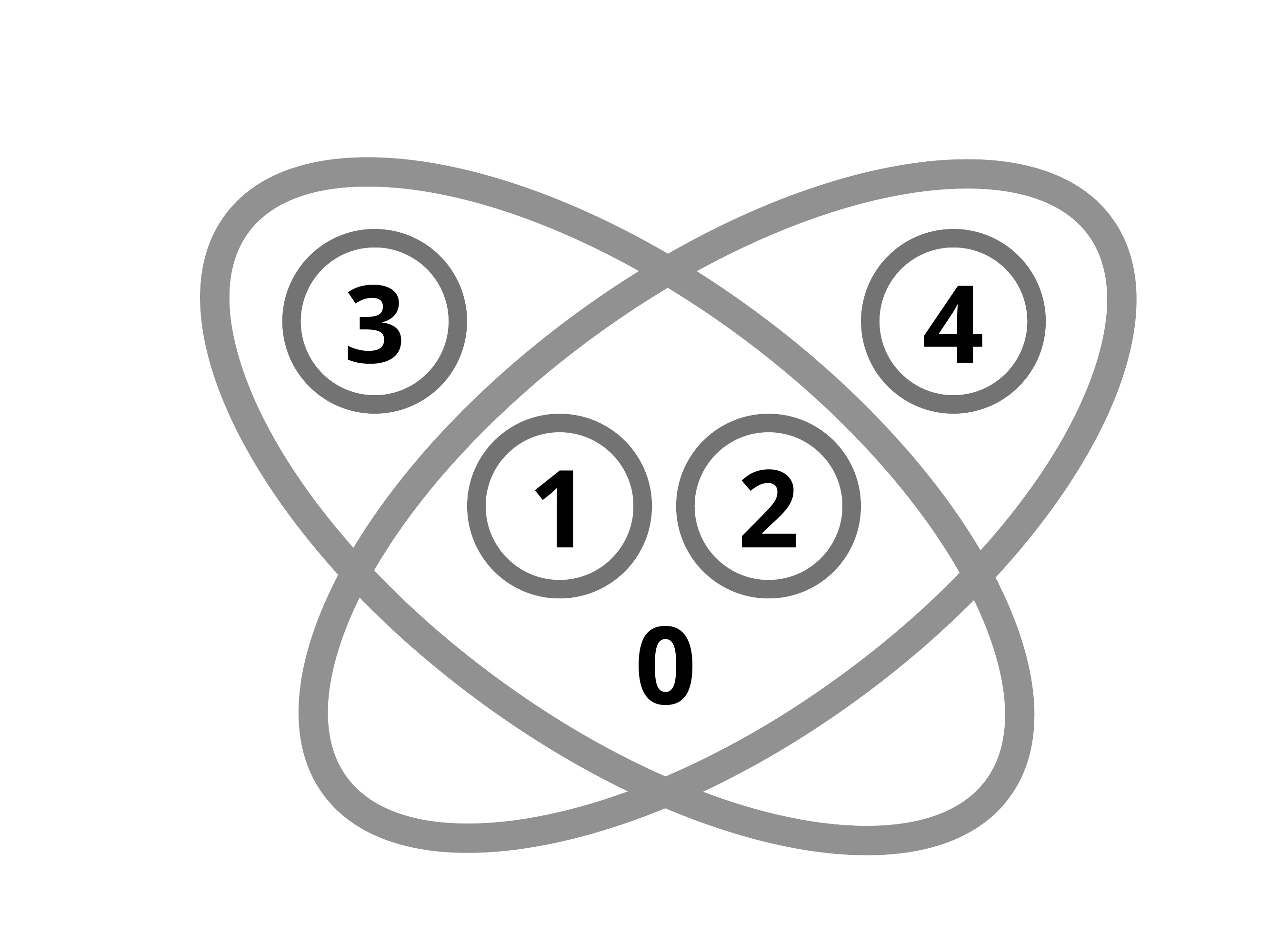}
\vspace{-0em}
\caption{Illustration of Example~\ref{xmpl.p.not.regular}(\ref{item.p.not.regular.01234}\&\ref{item.p.not.regular.01234b})}
\label{fig.irregular}
\end{figure}

\begin{lemm}
\label{lemm.p.not.regular}
Suppose $(\ns P,\opens)$ is a semitopology and $p\in\ns P$.
Then precisely one of the following possibilities must hold, and each one is possible: 
\begin{enumerate*}
\item
$p$ is regular: $p\in\community(p)$ and $\community(p)$ is topen (nonempty, open, and transitive). 
\item
$\community(p)$ is topen, but $p\notin\community(p)$. 
\item
$\community(p)=\varnothing$.
\item
$\community(p)$ is open but not transitive.
(Both $p\in\community(p)$ and $p\notin\community(p)$ are possible.)
\end{enumerate*}
\end{lemm}
\begin{proof} 
\leavevmode\begin{enumerate}
\item
To see that $p$ can be regular, consider $\ns P=\{0\}$ with the discrete topology.
Then $p\in\community(p)=\{0\}$.
\item
To see that it is possible for $\community(p)$ to be topen but $p$ is not in it, consider Example~\ref{xmpl.p.not.regular}(\ref{item.p.not.regular.01234}).
There, $\ns P=\{0,1,2,3,4\}$ and $\intertwined{0}=\{0,1,2\}$ and $\community(0)=\{1,2\}$.
Then $\community(0)$ is topen, but $0\notin\community(0)$.

(Another, slightly more compact but more distant, example is $p=\ast$ in the lower-right semitopology in Figure~\ref{fig.012}.)
\item
To see that $\community(p)=\varnothing$ is possible, consider Example~\ref{xmpl.p.not.regular}(\ref{item.p.not.regular.R}) (the real numbers $\mathbb R$ with its usual topology).
Then by Remark~\ref{rmrk.not.hausdorff} $\intertwined{r}=\{r\}$ and so $\community(x)=\interior(\{r\})=\varnothing$.
(See also Example~\ref{xmpl.two.topen.examples}(\ref{item.two.topen.examples.2}) for a more elaborate example.) 
\item
To see that it is possible for $\community(p)$ to be an open neighbourhood of $p$ but not transitive, see Example~\ref{xmpl.p.not.regular}(\ref{item.p.not.regular.012}).
There, $\ns P=\{0,1,2\}$ and $1\in \intertwined{1}=\{0,1,2\}=\community(1)$, but $\{0,1,2\}$ is not transitive (it contains two disjoint topens: $\{0\}$ and $\{2\}$).

To see that it is possible for $\community(p)$ to be open and nonempty yet not contain $p$ and not be transitive, see Example~\ref{xmpl.p.not.regular}(\ref{item.p.not.regular.01234b}) for $p=0$, and see also Remark~\ref{rmrk.indeed.two.closed.neighbourhoods} for a discussion of the connection with minimal closed neighbourhoods.
\end{enumerate}
The possibilities above are clearly mutually exclusive and exhaustive.
\end{proof}

\jamiesection{Closed sets}
\label{sect.closed.sets}

\jamiesubsection{Closed sets}
\label{subsect.closed.sets.basics}

\begin{rmrk}
\label{rmrk.computing.closures}
In Subsection~\ref{subsect.closed.sets.basics} we check that some familiar properties of closures carry over from topologies to semitopologies.
There are no technical surprises, but this is in itself a mathematical result that needs to be checked. 
From Subsection~\ref{subsect.trans.clos} and the following Subsections we will study the close relation between closures and sets of intertwined points. 
\end{rmrk}

\begin{defn}
\label{defn.closure}
Suppose $(\ns P,\opens)$ is a semitopology and suppose $p\in\ns P$ and $P\subseteq\ns P$.
Then:
\begin{enumerate*}
\item\label{item.closure}
Define $\closure{P}\subseteq\ns P$ the \deffont{closure of $P$} to be the set of points $p$ such that every open neighbourhood of $p$ intersects $P$.
In symbols using Notation~\ref{nttn.between}: 
$$
\closure{P} = \{ p'\in\ns P \mid \Forall{O{\in}\opens} p'\in O \limp P\between O\} .
$$
\item\label{item.closure.p}
As is standard, we may write $\closure{p}$ for $\closure{\{p\}}$.
Unpacking definitions for reference:
$$
\closure{p} = \{ p'\in\ns P \mid \Forall{O{\in}\opens} p'\in O \limp p\in O\} .
$$
\end{enumerate*}
\end{defn}

\begin{lemm}
\label{lemm.closure.monotone}
Suppose $(\ns P,\opens)$ is a semitopology and suppose $P,P'\subseteq\ns P$.
Then taking the closure of a set is: 
\begin{enumerate*}
\item\label{closure.monotone}
\emph{Monotone:}\quad If $P\subseteq P'$ then $\closure{P}\subseteq\closure{P'}$.
\item\label{closure.increasing}
\emph{Increasing:}\quad $P\subseteq\closure{P}$.
\item\label{closure.idempotent}
\emph{Idempotent:}\quad $\closure{P}=\closure{\closure{P}}$.
\end{enumerate*}
\end{lemm}
\begin{proof}
By routine calculations from Definition~\ref{defn.closure}.
\end{proof}

\begin{lemm}
\label{lemm.closure.open.char}
Suppose $(\ns P,\opens)$ is a semitopology and $P\subseteq\ns P$ and $O\in\opens$.
Then 
$$
P\between O
\quad\text{if and only if}\quad 
\closure{P}\between O.
$$
\end{lemm}
\begin{proof}
Suppose $P\between O$.
Then $\closure{P}\between O$ using Lemma~\ref{lemm.closure.monotone}(\ref{closure.increasing}).

Suppose $\closure{P}\between O$.
Pick $p\in \closure{P}\cap O$.
By construction of $\closure{P}$ in Definition~\ref{defn.closure} $p\in O\limp P\between O$.
It follows that $P\between O$ as required.
\end{proof}

\begin{defn}
\label{defn.closed}
Suppose $(\ns P,\opens)$ is a semitopology and suppose $C\subseteq\ns P$.
\begin{enumerate*}
\item\label{item.closed.set}
Call $C$ a \deffont{closed set} when $C=\closure{C}$.
\item
Call $C$ a \deffont{clopen set} when $C$ is closed and open.
\item
Write $\closed$ for the set of \deffont[closed sets $\closed$]{closed sets} (as we wrote $\opens$ for the open sets; the ambient semitopology will always be clear or understood).
\end{enumerate*}
\end{defn}

\begin{lemm}
\label{lemm.closure.closed}
Suppose $(\ns P,\opens)$ is a semitopology and suppose $P\subseteq\ns P$.
Then $\closure{P}$ is closed and contains $P$.
In symbols:
$$
P\subseteq \closure{P}\in\closed .
$$ 
\end{lemm}
\begin{proof}
From Definition~\ref{defn.closed}(\ref{item.closed.set}) and Lemma~\ref{lemm.closure.monotone}(\ref{closure.increasing} \& \ref{closure.idempotent}).
\end{proof}

\begin{xmpl}\leavevmode
\begin{enumerate}
\item
Take $\ns P=\{0,1\}$ and $\opens=\{\varnothing, \{0\}, \{0,1\}\}$.
Then the reader can verify that:
\begin{itemize*}
\item
$\{0\}$ is open.
\item
The closure of $\{1\}$ is $\{1\}$ and $\{1\}$ is closed.
\item
The closure of $\{0\}$ is $\{0,1\}$.
\item
$\varnothing$ and $\{0,1\}$ are the only clopen sets.
\end{itemize*}
\item
Now take $\ns P=\{0,1\}$ and $\opens=\{\varnothing, \{0\}, \{1\}, \{0,1\}\}$.\footnote{Following Definition~\ref{defn.value.assignment} and Example~\ref{xmpl.semitopologies}(\ref{item.boolean.discrete}), this is just $\{0,1\}$ with the \emph{discrete semitopology}.}
Then the reader can verify that:
\begin{itemize*}
\item
Every set is clopen.
\item
The closure of every set is itself.
\end{itemize*}
\end{enumerate}
\end{xmpl}

\begin{rmrk}
There are two standard definitions for when a set is closed: when it is equal to its closure (as per Definition~\ref{defn.closed}(\ref{item.closed.set})), and when it is the complement of an open set.
In topology these are equivalent.
We do need to check that the same holds in semitopology, but as it turns out the proof is routine:
\end{rmrk}

\begin{lemm}
\label{lemm.closed.complement.open}
Suppose $(\ns P,\opens)$ is a semitopology.
Then:
\begin{enumerate*}
\item\label{item.closed.complement.open.1}
Suppose $C\in\closed$ is closed (by Definition~\ref{defn.closed}: $C=\closure{C}$).
Then $\ns P\setminus C$ is open.
\item\label{item.closed.complement.open.2}
Suppose $O\in\opens$ is open.
Then $\ns P\setminus O$ is closed (by Definition~\ref{defn.closed}: $\closure{\ns P\setminus O}=\ns P\setminus O$).
\end{enumerate*}
\end{lemm}
\begin{proof}
\leavevmode
\begin{enumerate}
\item
Suppose $p\in \ns P\setminus C$.
Since $C=\closure{C}$, we have $p\in\ns P\setminus\closure{C}$.
Unpacking Definition~\ref{defn.closure}, this means precisely that there exists $O_p\in\opens$ with $p\in O_p \notbetween C$.
We use Lemma~\ref{lemm.open.is.open}. 
\item
Suppose $O\in\opens$.
Combining Lemma~\ref{lemm.open.is.open} with Definition~\ref{defn.closure} 
it follows that $O\notbetween \closure{\ns P\setminus O}$ so that $\closure{\ns P\setminus O}\subseteq\ns P\setminus O$.
Furthermore, by Lemma~\ref{lemm.closure.monotone}(\ref{closure.increasing}) $\ns P\setminus O\subseteq\closure{\ns P\setminus O}$.
\qedhere\end{enumerate}
\end{proof}

\begin{corr}
\label{corr.closed.complement.union}
If $C\in\closed$ then $\ns P\setminus C=\bigcup_{O\in\opens} O\notbetween C$.
\end{corr}
\begin{proof}
By Lemma~\ref{lemm.closed.complement.open}(\ref{item.closed.complement.open.1}) $\ns P\setminus C\subseteq\bigcup_{O\in\opens} O\notbetween C$.
Conversely, if $O\notbetween C$ then $O\subseteq\ns P\setminus C$ by Definition~\ref{defn.closure}(\ref{item.closure}). 
\end{proof}

\begin{corr}
\label{corr.closure.closure}
Suppose $(\ns P,\opens)$ is a semitopology and $P\subseteq\ns P$ and $\mathcal C\subseteq\powerset(\ns P)$.
Then:
\begin{enumerate*}
\item
$\varnothing$ and $\ns P$ are closed.
\item\label{closure.closure.cap}
If every $C\in\mathcal C$ is closed, then $\bigcap\mathcal C$ is closed.
Or succinctly in symbols:
$$
\mathcal C\subseteq\closed \limp \bigcap\mathcal C\in\closed .
$$
\item\label{item.closure.as.intersection}
$\closure{P}$ is equal to the intersection of all the closed sets that contain it.
In symbols:
$$
\closure{P}=\bigcap\{C\in\closed \mid P\subseteq C\}. 
$$
\end{enumerate*}
\end{corr}
\begin{proof}
\leavevmode
\begin{enumerate}
\item
Immediate from Lemma~\ref{lemm.closed.complement.open}(\ref{item.closed.complement.open.2}).
\item
From Lemma~\ref{lemm.closed.complement.open} and Definition~\ref{defn.semitopology}(\ref{semitopology.empty.and.universe}\&\ref{semitopology.unions}).
\item
By Lemma~\ref{lemm.closure.closed} $\bigcap\{C\in\closed \mid P\subseteq C\}\subseteq\closure{P}$.
By construction $P\subseteq\bigcap\{C\in\closed \mid P\subseteq C\}$, and using Lemma~\ref{lemm.closure.monotone}(\ref{closure.monotone}) and part~\ref{item.closure.as.intersection} of this result we have
$$
\closure{P} 
\stackrel{L\ref{lemm.closure.monotone}(\ref{closure.monotone})}\subseteq 
\closure{\bigcap\{C\in\closed \mid P\subseteq C\}} 
\stackrel{pt.2}= 
\bigcap\{C\in\closed \mid P\subseteq C\} .
$$ 
\qedhere\end{enumerate}
\end{proof}

The usual characterisation of continuity in terms of inverse images of closed sets being closed, remains valid:
\begin{corr}
\label{corr.alternative.cont.closed}
Suppose $(\ns P,\opens)$ and $(\ns P',\opens')$ are semitopological spaces (Definition~\ref{defn.semitopology}) and suppose $\avaluation:\ns P\to\ns P'$ is a function.
Then the following are equivalent:
\begin{enumerate*}
\item
$\avaluation$ is continuous, meaning by Definition~\ref{defn.continuity}(\ref{item.continuous.function}) that $\avaluation^\mone(O')\in\opens$ for every $O'\in\opens'$.
\item
$\avaluation^\mone(C')\in\closed$ for every $C'\in\closed'$.
\end{enumerate*}
\end{corr}
\begin{proof}
By routine calculations as for topologies, using Lemma~\ref{lemm.closed.complement.open} and the fact that the inverse image of a complement is the complement of the inverse image; see~\cite[Theorem~7.2, page~44]{willard:gent} or~\cite[Proposition~1.4.1(iv), page~28]{engelking:gent}.
\end{proof}

\jamiesubsection{Duality between closure and interior}

The usual dualities between closures and interiors remain valid in semitopologies.
There are no surprises but this still needs to be checked, so we spell out the details:
\begin{lemm}
\label{lemm.closure.interior}
Suppose $(\ns P,\opens)$ is a semitopology and $O\in\opens$ and $C\in\closed$.
Then:
\begin{enumerate*}
\item\label{item.closure.interior.open}
$O\subseteq\interior(\closure{O})$.  The inclusion may be strict.
\item\label{item.closure.interior.closed}
$\closure{\interior(C)}\subseteq C$.  The inclusion may be strict.
\item\label{item.closure.interior.complement.closure}
$\interior(\ns P\setminus O)=\ns P\setminus\closure{O}$.
\item\label{item.closure.interior.complement.interior}
$\closure{\ns P\setminus C}=\ns P\setminus\interior(C)$. 
\end{enumerate*}
\end{lemm}
\begin{proof}
The reasoning is just as for topologies, but we spell out the details:
\begin{enumerate}
\item
By Lemma~\ref{lemm.closure.monotone}(\ref{closure.increasing}) $O\subseteq\closure{O}$.
By Corollary~\ref{corr.interior.monotone} $\interior(O)\subseteq\interior(\closure{O})$.
By Lemma~\ref{lemm.interior.open} $O=\interior(O)$, so we are done.

For an example of the strict inclusion, consider $\mathbb R$ with the usual topology (which is also a semitopology) and take $O=(0,1)\cup(1,2)$.
Then $O\subsetneq\interior(\closure{O})=(0,2)$.
\item
By Lemma~\ref{lemm.interior.open} $\interior(C)\subseteq C$.
By Lemma~\ref{lemm.closure.monotone}(\ref{closure.monotone}) $\closure{\interior(C)}\subseteq\closure{C}$.
By Definition~\ref{defn.closed}(\ref{item.closed.set}) (since we assumed $C\in\closed$) $\closure{C}=C$, so we are done.

For an example of the strict inclusion, consider $\mathbb R$ with the usual topology and take $C=\{0\}$.
Then $\closure{\interior(C)}=\varnothing\subsetneq C$.
\item
Consider some $p'\in\ns P$.
By Definition~\ref{defn.interior} $p'\in \interior(\ns P\setminus O)$ when there exists some $O'\in\opens$ such that $p'\in O'\notbetween O$.
By definition in Definition~\ref{defn.closure}(\ref{item.closure}) this happens precisely when $p'\notin\closure{O}$. 
\item
By Definition~\ref{defn.closure}(\ref{item.closure}), $p'\notin \closure{\ns P\setminus C}$ precisely when there exists some $O'\in\opens$ such that $p'\in O'\notbetween \ns P\setminus C$.
By facts of sets this means precisely that $p'\in O'\subseteq C$.
By Definition~\ref{defn.interior} this means precisely that $p'\in\interior(C)$.
\qedhere\end{enumerate}
\end{proof}

\begin{corr}
\label{corr.ic.ci}
Suppose $(\ns P,\opens)$ is a semitopology and 
$O\in\opens$ and $C\in\closed$.
Then:
\begin{enumerate*}
\item
$\closure{O} = \closure{\interior(\closure{O})}$. 
\item
$\interior(C)=\interior(\closure{\interior(C)})$.
\end{enumerate*}
\end{corr}
\begin{proof}
We use Lemma~\ref{lemm.closure.interior}(\ref{item.closure.interior.open}\&\ref{item.closure.interior.complement.closure}) along with Lemma~\ref{lemm.closure.monotone}(\ref{closure.monotone}) and Corollary~\ref{corr.interior.monotone}: 
$$
\begin{array}{r@{\ }c@{\ }c@{\ }c@{\ }ll}
\closure{O}
&\stackrel{L\ref{lemm.closure.interior}(\ref{item.closure.interior.open})\&L\ref{lemm.closure.monotone}(\ref{closure.monotone})}\subseteq&
\closure{\interior(\closure{O})}
&\stackrel{L\ref{lemm.closure.interior}(\ref{item.closure.interior.closed})}\subseteq&
\interior(\closure{O})
\\
\interior(C)
&\stackrel{L\ref{lemm.closure.interior}(\ref{item.closure.interior.open})}\subseteq&
\interior(\closure{\interior(C)})
&\stackrel{L\ref{lemm.closure.interior}(\ref{item.closure.interior.closed})\&C\ref{corr.interior.monotone}}\subseteq&
\interior(C)
\end{array}
$$
\end{proof}

\jamiesubsection{Transitivity and closure}
\label{subsect.trans.clos}

We explore how the topological closure operation interacts with taking transitive sets.
\begin{lemm}
\label{lemm.open.consensus}
Suppose $(\ns P,\opens)$ is a semitopology and $T\subseteq\ns P$ is transitive and $O\in\opens$.
Then 
$$
\atopen\between O
\quad\text{implies}\quad
\closure{T}\subseteq\closure{O}.
$$
\end{lemm}
\begin{proof}
Unpacking Definition~\ref{defn.closure}
we have:
$$
\begin{array}{r@{\ }l}
p'\in\closure{T}\liff&\Forall{O'{\in}\opens}p'\in O'\limp O'\between \atopen 
\qquad\text{and}
\\
p'\in\closure{O}\liff&\Forall{O'{\in}\opens}p'\in O'\limp O'\between O
.
\end{array}
$$
It would suffice to prove $O'\between \atopen\limp O'\between O$ for any $O'\in\opens$.

So suppose $O'\between \atopen$.
By assumption $\atopen\between O$ and by transitivity of $\atopen$ (Definition~\ref{defn.transitive}) $O'\between O$.
\end{proof}

\begin{prop}
\label{prop.open.consensus}
\label{prop.open.strong-consensus}
Suppose $(\ns P,\opens)$ is a semitopology and $\atopen\in\topens$ and $O\in\opens$.
Then the following are equivalent:
$$
\atopen\between O
\quad\text{if and only if}\quad
\atopen\subseteq\closure{\atopen}\subseteq \closure{O}
.
$$
\end{prop}
\begin{proof}
We prove two implications:
\begin{itemize}
\item
Suppose $\atopen\between O$.
By Lemma~\ref{lemm.open.consensus} $\closure{\atopen}\subseteq\closure{O}$.
By Lemma~\ref{lemm.closure.monotone}(\ref{closure.increasing}) (as standard) $\atopen\subseteq\closure{\atopen}$. 
\item
Suppose $\atopen\subseteq\closure{\atopen}\subseteq\closure{O}$.
Then $\atopen\between\closure{O}$ and by Lemma~\ref{lemm.closure.open.char} (since $\atopen$ is nonempty (and transitive) and open) also $\atopen\between O$.
\qedhere\end{itemize}
\end{proof}

\begin{rmrk}
\label{rmrk.gradecast}
In retrospect we can see the imprint of topens (Definition~\ref{defn.transitive}) in previous work, if we look at things in a certain way.
Many consensus algorithms have the property that once consensus is established in a quorum $O$, it propagates to $\closure{O}$.

This is apparent (for example) in the Grade-Cast algorithm~\cite{feldman_optimal_1988}, in which participants assign a confidence grade of 0, 1 or 2 to their output and must ensure that if any participant outputs $v$ with grade 2 then all must output $v$ with grade at least 1.
In this algorithm, if a participant finds that all its quorums intersect some set $S$ that unanimously supports value $v$, then the participant assigns grade at least 1 to $v$.
From our point of view here, this is just taking a closure in the style we discussed in Remark~\ref{rmrk.computing.closures}.
If $T$ unanimously supports $v$ and participants communicate enough, then eventually every member of $\closure{T}$ assigns grade at least 1 to $v$.
Thus, Proposition~\ref{prop.open.strong-consensus} suggests that, to convince a topen to agree on a value, we can first convince an open neighbourhood that intersects the topen, and then use Grade-Cast to convince the closure of that open set and thus in particular the topen which we know must be contained in that closure. 
\end{rmrk}

We conclude with an easy observation which will be useful later.
Recall from Notation~\ref{nttn.intertwined.space} the notion of an intertwined space being one such that all nonempty open sets intersect.
Then we have:
\begin{lemm}
\label{lemm.intertwined.iff.closure}
Suppose $(\ns P,\opens)$ is a semitopology and suppose $\atopen\in\topens$.
Then the following are equivalent:
\begin{enumerate*}
\item
$\ns P$ is intertwined.
\item
$\closure{\atopen}=\ns P$.
\end{enumerate*}
\end{lemm}
\begin{proof}
Suppose $\closure{\atopen}=\ns P$ and consider any $O,O'\in\opens$.
Unpacking Definition~\ref{defn.closure}(\ref{item.closure}) it follows that $O\between\atopen\between O'$.
By transitivity of $\atopen$ (Definition~\ref{defn.transitive}(\ref{transitive.transitive})) $O\between O'$ as required.

Suppose $(\ns P,\opens)$ is intertwined.
By Lemma~\ref{lemm.intertwined.space} every nonempty open set is topen, thus $\ns P$ is topen, and $\ns P=\closure{\atopen}$ follows by Lemma~\ref{lemm.open.consensus}. 
\end{proof}

\jamiesubsection{Closed neighbourhoods and intertwined points}
\label{subsect.closed.neighbourhoods}

\jamiesubsubsection{Definition and basic properties}

\begin{defn}
\label{defn.cn}
Suppose $(\ns P,\opens)$ is a semitopology.
We generalise Definition~\ref{defn.open.neighbourhood} as follows:
\begin{enumerate*}
\item\label{item.neighbourhood.of.p}
Call $P\subseteq\ns P$ a \deffont{neighbourhood} when it contains an open set (i.e. when $\interior(P)\neq\varnothing$), and call $P$ a \deffont{neighbourhood of $p$} when $p\in\ns P$ and $P$ contains an open neighbourhood of $p$ (i.e. when $p\in\interior(P)$).
In particular:
\item\label{item.closed.neighbourhood.of.p}
$C\subseteq\ns P$ is a \deffont{closed neighbourhood of $p\in\ns P$} when $C$ is closed and $p\in\interior(C)$.
\item\label{item.closed.neighbourhood}
$C\subseteq\ns P$ is a \deffont{closed neighbourhood} when $C$ is closed and $\interior(C)\neq\varnothing$.
\end{enumerate*} 
\end{defn}

\begin{rmrk}
\leavevmode
\begin{enumerate}
\item
If $C$ is a closed neighbourhood of $p$ in the sense of Definition~\ref{defn.cn}(\ref{item.closed.neighbourhood.of.p}) then $C$ is a closed neighbourhood in the sense of Definition~\ref{defn.cn}(\ref{item.closed.neighbourhood}), just because if $p\in\interior(C)$ then $\interior(C)\neq\varnothing$. 
\item
$p\in C$ is not enough for $C$ to be a closed neighbourhood of $p$;
we require the stronger condition $p\in\interior(C)$.

For instance take $\ns P=\{0,1\}$ and $\opens=\{\varnothing,\{1\},\ns P\}$ (the Sierpi\'nski space; see Figure~\ref{fig.sierpinski}), and consider $p=0$ and $C=\{0\}$.
Then $p\in C$ but $p\not\oldin\interior(C)=\varnothing$, so that $C$ is not a closed neighbourhood of $p$. 
\end{enumerate}
\end{rmrk}

Recall from Definition~\ref{defn.intertwined.points} the notions of $p\intertwinedwith p'$ and $\intertwined{p}$.
Proposition~\ref{prop.intertwined.as.closure} packages up our material for convenient use in later results. 
\begin{prop}
\label{prop.intertwined.as.closure}
Suppose $(\ns P,\opens)$ is a semitopology and $p,p'\in\ns P$.
Then:
\begin{enumerate*}
\item\label{item.intertwined.as.closure.1}
We can characterise when $p'$ is intertwined with $p$ as follows: 
$$
p\intertwinedwith p' 
\quad\text{if and only if}\quad
\Forall{O{\in}\opens} p\in O\limp p'\in\closure{O} .
$$
\item\label{item.intertwined.as.intersection.of.closures}
As a corollary,
$$
\intertwined{p} = \bigcap\{\closure{O} \mid p\in O\in\opens\}.
$$
\item\label{intertwined.as.closure.closed}
Equivalently:
$$
\begin{array}{r@{\ }l@{\qquad}l}
\intertwined{p}
=& \bigcap\{C\in\closed \mid p\in \interior(C) \}
\\
=&
\bigcap\{C\in\tf{Closed} \mid C\text{ a closed neighbourhood of }p\}
&\text{Definition~\ref{defn.cn}}.
\end{array}
$$
Thus in particular, if $C$ is a closed neighbourhood of $p$ then $\intertwined{p}\subseteq C$.
\item\label{intertwined.p.closed}
$\intertwined{p}$ is closed and $\ns P\setminus\intertwined{p}$ is open.
\end{enumerate*}
\end{prop}
\begin{proof}
\leavevmode
\begin{enumerate}
\item
We just rearrange Definition~\ref{defn.intertwined.points}.
So
$$
\Forall{O,O'\in\opens}((p\in O\land p'\in O') \limp O\between O')
$$
rearranges to
$$
\Forall{O\in\opens}(p\in O\limp \Forall{O'\in\opens} (p'\in O' \limp O\between O')) . 
$$
We now observe from Definition~\ref{defn.closure} that this is precisely
$$
\Forall{O\in\opens}(p\in O\limp p'\in\closure{O}).
$$
\item
We just rephrase part~\ref{item.intertwined.as.closure.1} of this result.
\item
Using part~\ref{item.intertwined.as.intersection.of.closures} of this result it would suffice to prove
$$
\bigcap\{\closure{O}\mid p\in O\in\opens\} = \bigcap\{C\in\closed \mid p\in \interior(C) \} .
$$
We will do this by proving that for each $O$-component on the left there is a $C$ on the right with $C\subseteq\closure{O}$; and for each $C$-component on the right there is an $O$ on the left with $\closure{O}\subseteq C$:
\begin{itemize}
\item
Consider some $O\in\opens$ with $p\in O$.

We set $C=\closure{O}$, so that trivially $C\subseteq\closure{O}$.
By Lemma~\ref{lemm.closure.interior}(\ref{item.closure.interior.open}) $O\subseteq\interior(\closure{O})$, so $p\in\interior(C)$.
\item
Consider some $C\in\closed$ such that $p\in\interior(C)$.

We set $O=\interior(C)$.
Then $p\in O$, and by Lemma~\ref{lemm.closure.interior}(\ref{item.closure.interior.closed}) $\closure{O}\subseteq C$.
\end{itemize}
\item
Part~\ref{intertwined.as.closure.closed} of this result exhibits $\intertwined{p}$ as an intersection of closed sets, and by Corollary~\ref{corr.closure.closure}(\ref{closure.closure.cap}) this is closed.
By Lemma~\ref{lemm.closed.complement.open}(\ref{item.closed.complement.open.1}) its complement $\ns P\setminus\intertwined{p}$ is open.
\qedhere\end{enumerate}
\end{proof}

\begin{defn}
\label{defn.nbhd.system}
\label{defn.nbhd}
Suppose $(\ns P,\opens)$ is a semitopology and $p\in\ns P$.
\begin{enumerate*}
\item
Write $\nbhd(p)=\{O\in\opens\mid p\in\opens\}$ and call this the \deffont[open neighbourhood system $\nbhd(p)$]{open neighbourhood system} of $p\in\ns P$. 
\item
Write $\nbhd^c(p)=\{C\in\closed\mid p\in\closed\}$ and call this the \deffont[closed neighbourhood system $\nbhd^c(p)$]{closed neighbourhood system}\index{$\nbhd^c(p)$ (closed neighbourhood system of a point)} of $p\in\ns P$.
\end{enumerate*}
\end{defn}

\begin{rmrk}
\label{rmrk.nbhd.concise}
As standard, we can use Definition~\ref{defn.nbhd} to rewrite the definition of $\avaluation$ being continuous at $p$ (Definition~\ref{defn.continuity}(\ref{item.continuous.function.at.p})) as
$$
\Forall{O'{\in}\nbhd(f(p))}\Exists{O{\in}\nbhd(p)} O\subseteq f^\mone(O') .
$$
\end{rmrk}

\begin{rmrk}
\label{rmrk.nbhd.filter}
If $(\ns P,\opens)$ is a topology, then $\nbhd(p)$ is a filter (a nonempty up-closed down-directed set) and this is often called the \emph{neighbourhood filter} of $p$.

We are working with semitopologies, so $\opens$ is not necessarily closed under intersections, and $\nbhd(p)$ is not necessarily a filter.
Figure~\ref{fig.nbhd} illustrates examples of this: e.g. in the left-hand example $\{0,1\},\{0,2\}\in \nbhd(0)$ but $\{0\}\notin\nbhd(0)$, since $\{0\}$ is not an open set.
\end{rmrk}

\begin{figure}
\vspace{-1em}
\centering
\includegraphics[align=c,width=0.3\columnwidth,trim={50 0 50 0},clip]{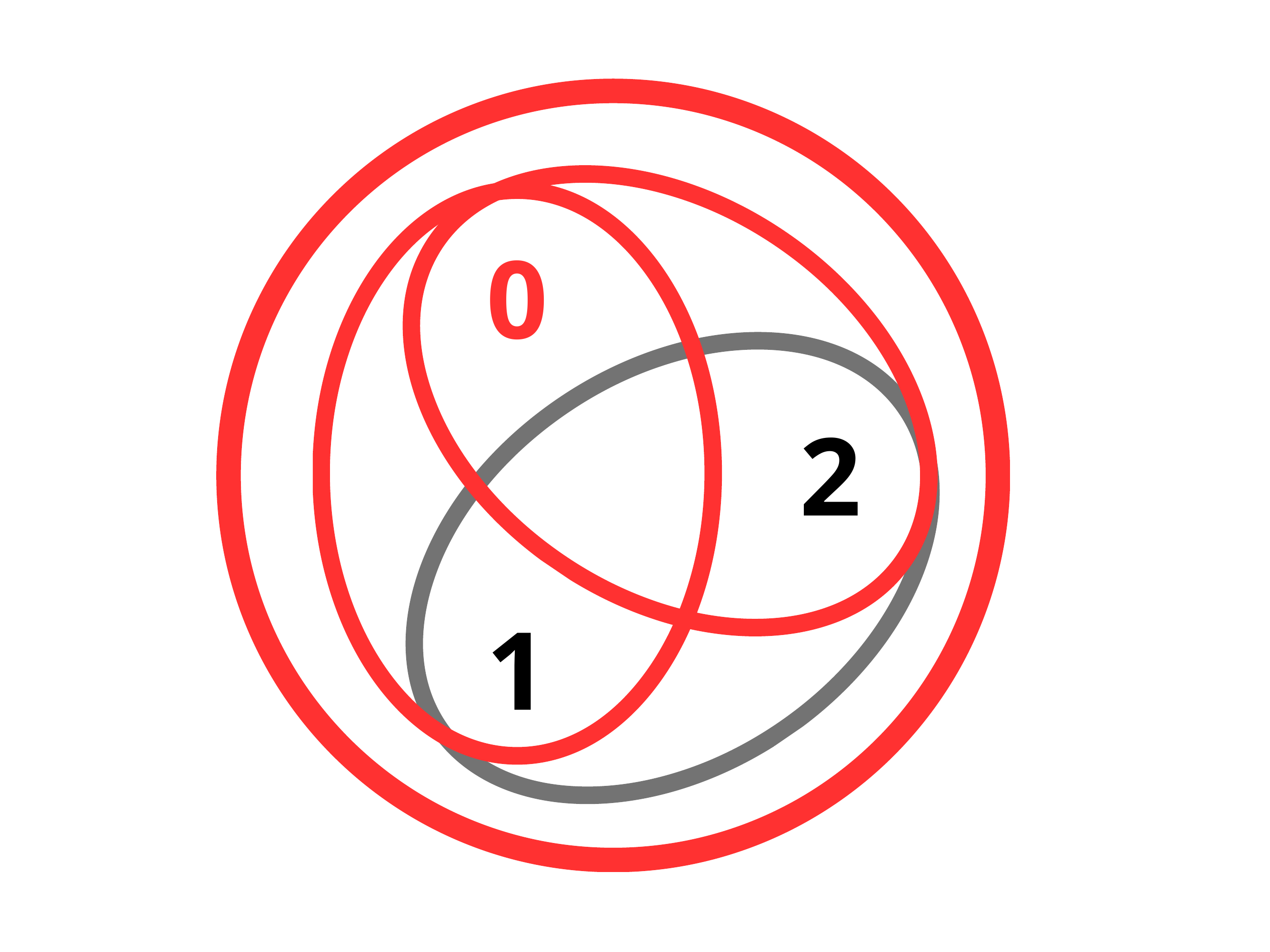}
\quad
\includegraphics[align=c,width=0.32\columnwidth,trim={50 0 50 0},clip]{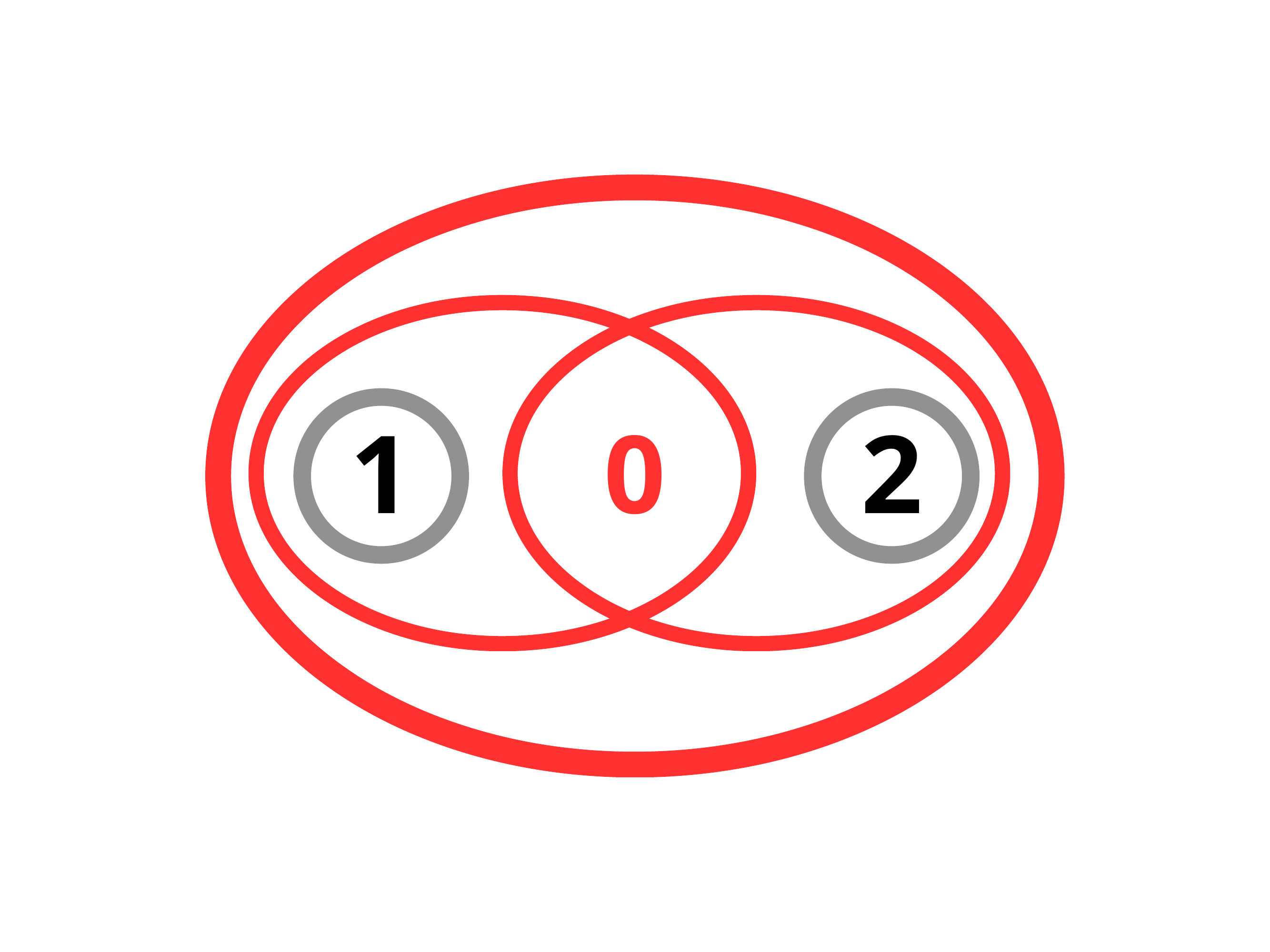}
\quad
\includegraphics[align=c,width=0.28\columnwidth,trim={50 0 50 0},clip]{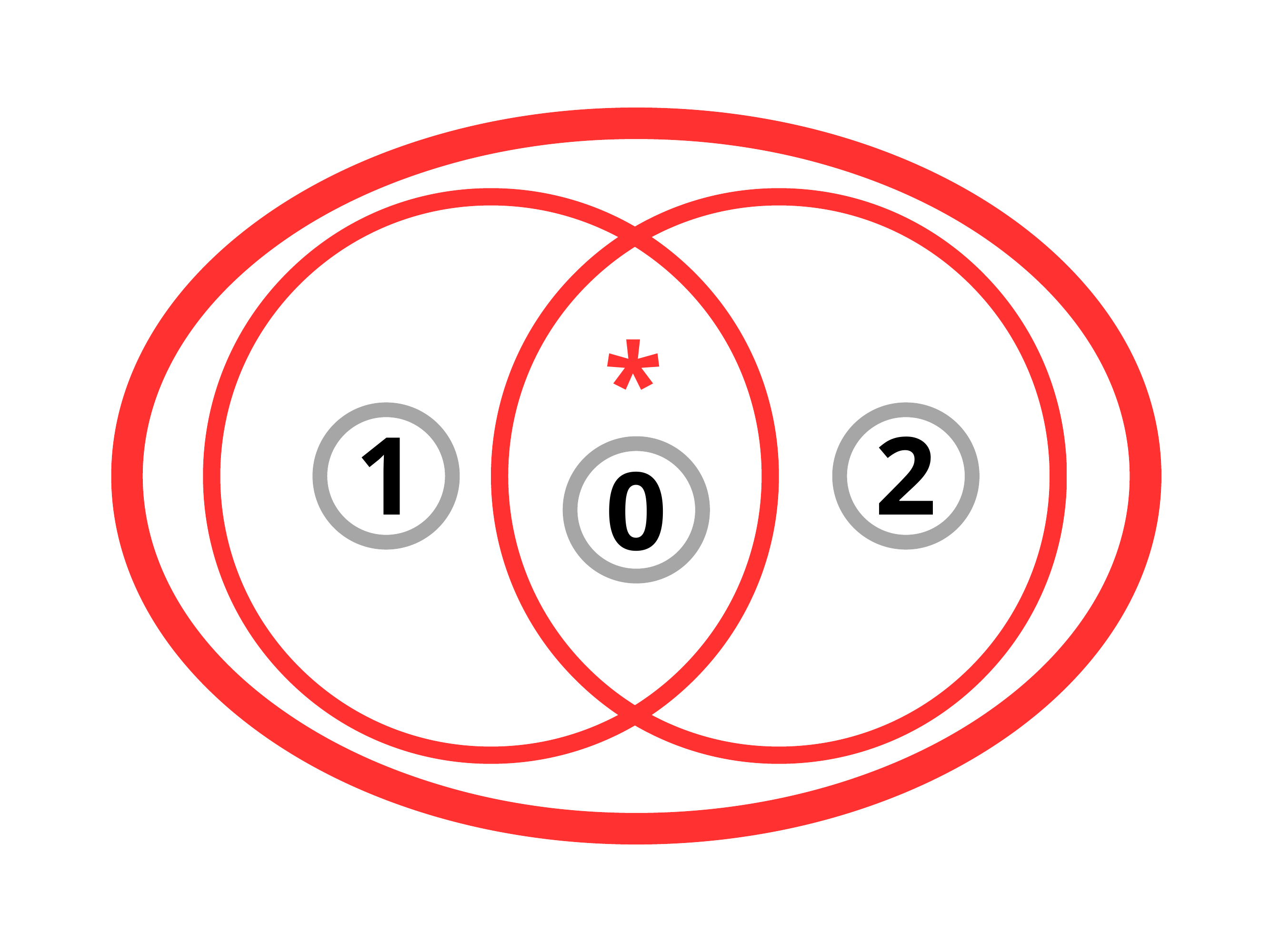}
\caption{Examples of open neighbourhoods (Remark~\ref{rmrk.nbhd.filter})}
\label{fig.nbhd}
\end{figure}

\begin{rmrk}
\label{rmrk.cluster.convergence.2}
We can relate Proposition~\ref{prop.intertwined.as.closure} to a concept from topology. 
Following standard terminology (\cite[Definition~2, page~69]{bourbaki:gent1} or \cite[page~52]{engelking:gent}), a \deffont{cluster point} $p\in\ns P$ of $\mathcal O\subseteq\opens$ is one such that every open neighbourhood of $p$ intersects every $O\in\mathcal O$.
Then Proposition~\ref{prop.intertwined.as.closure}(\ref{item.intertwined.as.intersection.of.closures}) identifies $\intertwined{p}$ as the set of cluster points of $\nbhd(p)\subseteq\opens$.
\end{rmrk}

\jamiesubsubsection{Application to characterise (quasi/weak) regularity}

\begin{rmrk}
\label{rmrk.how.weakly.regular}
Recall that Theorem~\ref{thrm.max.cc.char} characterised regularity in multiple ways, including as the existence of a greatest topen neighbourhood. 
Proposition~\ref{prop.views.of.regularity} below does something similar, for quasiregularity and weak regularity and the existence of closed neighbourhoods (Definition~\ref{defn.cn}), and Theorem~\ref{thrm.up.down.char} is a result in the same style, for regularity.

Here, for the reader's convenience, is a summary of the relevant results:
\begin{enumerate*}
\item
Proposition~\ref{prop.views.of.quasiregularity}:\ 
$p$ is quasiregular when $\intertwined{p}$ is a closed neighbourhood.
\item
Proposition~\ref{prop.views.of.regularity}:\ 
$p$ is weakly regular when $\intertwined{p}$ is a closed neighbourhood of $p$.
\item
Theorem~\ref{thrm.up.down.char}:\ 
$p$ is regular when $\intertwined{p}$ is a closed neighbourhood of $p$ and is a minimal closed neighbourhood.
\end{enumerate*}
\end{rmrk}

\begin{prop}
\label{prop.views.of.quasiregularity}
Suppose $(\ns P,\opens)$ is a semitopology and $p\in\ns P$.
Then the following are equivalent:
\begin{enumerate*}
\item
$p$ is quasiregular, or in full: $\community(p)\neq\varnothing$ (Definition~\ref{defn.tn}(\ref{item.quasiregular.point})).
\item
$\intertwined{p}$ is a closed neighbourhood (Definition~\ref{defn.cn}(\ref{item.closed.neighbourhood})).
\end{enumerate*}
\end{prop}
\begin{proof}
By construction in Definition~\ref{defn.tn}(\ref{item.tn}), $\community(p)=\interior(\intertwined{p})$.
So $\community(p)\neq\varnothing$ means precisely that $\intertwined{p}$ is a closed neighbourhood.
\end{proof}

\begin{prop}
\label{prop.views.of.regularity}
Suppose $(\ns P,\opens)$ is a semitopology and $p\in\ns P$.
Then the following are equivalent:
\begin{enumerate*}
\item\label{item.views.of.regularity.wr}
$p$ is weakly regular, or in full: $p\in\community(p)$ (Definition~\ref{defn.tn}(\ref{item.weakly.regular.point})).
\item\label{item.intertwined.p.closed.neighbourhood.of.p}
$\intertwined{p}$ is a closed neighbourhood of $p$ (Definition~\ref{defn.cn}(\ref{item.closed.neighbourhood.of.p})).
\item\label{item.views.of.regularity.cn}
The poset of closed neighbourhoods of $p$ ordered by subset inclusion, has a least element.
\item\label{item.intertwined.p.least.in.poset.closed.neighbourhoods.of.p}
$\intertwined{p}$ is least in the poset of closed neighbourhoods of $p$ ordered by subset inclusion.
\end{enumerate*}
\end{prop}
\begin{proof}
We prove a cycle of implications:
\begin{itemize}
\item
Suppose 
$p\in\interior(\intertwined{p})$.
By Proposition~\ref{prop.intertwined.as.closure}(\ref{intertwined.p.closed}) $\intertwined{p}$ is closed, so this makes it a closed neighbourhood of $p$ as per Definition~\ref{defn.cn}.
\item
Suppose $\intertwined{p}$ is a closed neighbourhood of $p$.
By Proposition~\ref{prop.intertwined.as.closure}(\ref{intertwined.as.closure.closed}) 
$\intertwined{p}$ is the intersection of \emph{all} closed neighbourhoods of $p$, and it follows that this poset has $\intertwined{p}$ as a least element.
\item
Assume the poset of closed neighbourhoods of $p$ has a least element; write it $C$.
So $C=\bigcap\{C'\in\tf{Closed}\mid C'\text{ is a closed neighbourhood of }p\}$ and thus by Proposition~\ref{prop.intertwined.as.closure}(\ref{intertwined.as.closure.closed}) $C=\intertwined{p}$.
\item
If $\intertwined{p}$ is least in the poset of closed neighbourhoods of $p$ ordered by subset inclusion, then in particular $\intertwined{p}$ is a closed neighbourhood of $p$ and it follows from Definition~\ref{defn.cn} that $p\in\interior(\intertwined{p})$. 
\qedhere\end{itemize}
\end{proof}

Recall from Definition~\ref{defn.tn} that $\community(p)=\interior(\intertwined{p})$:
\begin{lemm}
\label{lemm.closure.community.subset}
Suppose $(\ns P,\opens)$ is a semitopology and $p\in\ns P$.
Then $\closure{\community(p)}\subseteq\intertwined{p}$.
\end{lemm}
\begin{proof}
By Proposition~\ref{prop.intertwined.as.closure}(\ref{intertwined.p.closed}) $\intertwined{p}$ is closed; we use Lemma~\ref{lemm.closure.interior}(\ref{item.closure.interior.closed}).
\end{proof}

\begin{thrm}
\label{thrm.pKp}
Suppose $(\ns P,\opens)$ is a semitopology and $p\in\ns P$.
Then:
\begin{enumerate*}
\item\label{item.pKp.1}
If $p$ weakly regular then $\closure{\community(p)}=\intertwined{p}$.
In symbols:
$$
p\in\community(p)
\quad\text{implies}\quad \closure{\community(p)}=\intertwined{p}.
$$
\item\label{item.closure.community.p.intertwined}
As an immediate corollary, if $p$ is regular then $\closure{\community(p)}=\intertwined{p}$.
\end{enumerate*}
\end{thrm}
\begin{proof}
We consider each part in turn:
\begin{enumerate}
\item
If $p\in\community(p)=\interior(\intertwined{p})$ then $\closure{\community(p)}$ is a closed neighbourhood of $p$, so by Proposition~\ref{prop.intertwined.as.closure}(\ref{intertwined.as.closure.closed}) $\intertwined{p}\subseteq\closure{\community(p)}$.
By Lemma~\ref{lemm.closure.community.subset} $\closure{\community(p)}\subseteq\intertwined{p}$.
\item
By Lemma~\ref{lemm.wr.r}(\ref{item.r.implies.wr}) if $p$ is regular then it is weakly regular.
We use part~\ref{item.pKp.1} of this result. 
\qedhere\end{enumerate}
\end{proof}

We can combine Theorem~\ref{thrm.pKp} with Corollary~\ref{corr.regular.is.regular}: 
\begin{corr}
\label{corr.corr.pKp}
Suppose $(\ns P,\opens)$ is a semitopology and $p\in\ns P$. 
Then the following are equivalent:
\begin{enumerate*}
\item
$p$ is regular.
\item
$p$ is weakly regular and $\intertwined{p}=\intertwined{p'}$ \ for every $p'\in\community(p)$.
\end{enumerate*} 
\end{corr}
\begin{proof}
Suppose $p$ is regular and $p'\in\community(p)$.
Then $p$ is weakly regular by Lemma~\ref{lemm.wr.r}(\ref{item.r.implies.wr}), and $\community(p)=\community(p')$ by Corollary~\ref{corr.regular.is.regular}, and $\intertwined{p}=\intertwined{p'}$ by Theorem~\ref{thrm.pKp}.

Suppose $p$ is weakly regular and $\intertwined{p}=\intertwined{p'}$ for every $p'\in\community(p)$.
By Definition~\ref{defn.tn}(\ref{item.tn}) also $\community(p)=\interior(\intertwined{p})=\interior(\intertwined{p'})=\community(p')$ for every $p'\in\community(p)$, and by Corollary~\ref{corr.regular.is.regular} $p$ is regular.
\end{proof}

\begin{rmrk}
Note a subtlety to Corollary~\ref{corr.corr.pKp}: it is possible for $p$ to be regular, yet it is not the case that $\intertwined{p}=\intertwined{p'}$ for every $p'\in\intertwined{p}$ (rather than for every $p'\in\community(p)$).
For an example consider the top-left semitopology in Figure~\ref{fig.012}, taking $p=0$ and $p'=1$; then $1\in\intertwined{0}$ but $\intertwined{0}=\{0,1\}$ and $\intertwined{1}=\{0,1,2\}$.

To understand why this happens the interested reader can look ahead to Subsection~\ref{subsect.reg.tra.int}: in the terminology of that Subsection, $p'$ needs to be \emph{unconflicted} in Corollaries~\ref{corr.regular.is.regular} and~\ref{corr.corr.pKp}. 
\end{rmrk}

\jamiesubsection{Intersections of communities with open sets}

\begin{rmrk}[An observation about consensus]
\label{rmrk.fundamental.consensus}
Proposition~\ref{prop.regular.closure} and Lemma~\ref{lemm.regular.between} tell us some interesting and useful things: 
\begin{itemize*}
\item
Suppose a weakly regular $p$ wants to convince its community $\community(p)$ of some belief.
How might it proceed?

By Proposition~\ref{prop.regular.closure} it would suffice to seed one of the open neighbourhoods in its community with that belief, and then compute a \emph{topological closure} of that open set; in Remark~\ref{rmrk.why.top.closure} we discuss why topological closures are particularly interesting. 
\item
Suppose $p$ is regular, so it is a member of a transitive open neighbourhood, and $p$ wants to convince its community $\community(p)$ of some belief.

By Lemma~\ref{lemm.regular.between} $p$ need only convince \emph{some} open set that intersects its community (this open set need not even contain $p$), and then compute a topological closure as in the previous point.
\end{itemize*}
\end{rmrk}

\begin{lemm}
\label{lemm.regular.between}
Suppose $(\ns P,\opens)$ is a semitopology and $p\in\ns P$ is regular (so $p\in\community(p)\in\topens$).
Suppose $O\in\opens$.
Then
$$
p\in O\between \community(p)
\quad\text{implies}\quad 
\community(p)\subseteq\intertwined{p}\subseteq\closure{O}.
$$
In word:
\begin{quote}
If an open set intersects the community of a regular point, then that community is included in the closure of the open set.
\end{quote}
\end{lemm}
\begin{proof}
Suppose $p$ is regular, so $p\in\community(p)\in\topens$, and suppose $p\in O\between\community(p)$.
By Proposition~\ref{prop.open.strong-consensus} $\community(p)\subseteq\closure{\community(p)}\subseteq\closure{O}$.
By Theorem~\ref{thrm.pKp} $\closure{\community(p)}=\intertwined{p}$, and putting this together we get 
$$
\community(p)\subseteq\intertwined{p}\subseteq\closure{O}
$$ 
as required.
\end{proof}

Proposition~\ref{prop.regular.closure} generalises Theorem~\ref{thrm.pKp}, and is proved using it.
We regain Theorem~\ref{thrm.pKp} as the special case where $O=\community(p)$: 
\begin{prop}
\label{prop.regular.closure}
Suppose $(\ns P,\opens)$ is a semitopology and $p\in\ns P$ is weakly regular (so $p\in\community(p)\in\opens$).
Suppose $O\in\opens$.
Then:
\begin{enumerate*}
\item\label{item.regular.closure.1}
$p\in O\subseteq\community(p)$ implies
$\intertwined{p}=\closure{O}$.
\item\label{item.regular.closure.2}
As a corollary, $p\in O\subseteq\intertwined{p}$ implies
$\intertwined{p}=\closure{O}$.
\end{enumerate*}
\end{prop}
\begin{proof}
If $p\in O\subseteq\community(p)$ then $p\in\community(p)$ and using Theorem~\ref{thrm.pKp} $\closure{\community(p)}\subseteq\intertwined{p}$.
Since $O\subseteq\community(p)$ also $\closure{O}\subseteq\intertwined{p}$.
Also, by Proposition~\ref{prop.intertwined.as.closure}(\ref{item.intertwined.as.intersection.of.closures}) (since $p\in O\in\opens$) $\intertwined{p}\subseteq\closure{O}$.

For the corollary, we note that if $O$ is open then $O\subseteq\interior(\intertwined{p})=\community(p)$ if and only if $O\subseteq\intertwined{p}$.
\end{proof}

\begin{rmrk}
Note in Proposition~\ref{prop.regular.closure} that it really matters that $p\in O$ --- that is, that $O$ is an open neighbourhood \emph{of $p$} and not just an open set in $\intertwined{p}$.

To see why, consider the example in Lemma~\ref{lemm.two.intertwined} (illustrated in Figure~\ref{fig.012}, top-left diagram): so $\ns P=\{0,1,2\}$ and $\opens=\{\varnothing,\ns P,\{0\},\{2\}\}$.
Note that:
\begin{itemize*}
\item
$\intertwined{1}=\{0,1,2\}$.
\item
If we set $O=\{0\}\subseteq\{0,1,2\}$ then this is open, but $\closure{O}=\{0,1\}\neq\{0,1,2\}$.
\item
If we set $O=\{0,1,2\}\subseteq\{0,1,2\}$ then $\closure{O}=\{0,1,2\}$.
\end{itemize*}
\end{rmrk}

\begin{rmrk}
\label{rmrk.why.top.closure}
Topological closures will matter because we will develop a theory of computable semitopologies which will (amongst other things) deliver a distributed algorithm to compute closures.

Thus, we can say that from the point of view of a regular participant $p$, Proposition~\ref{prop.regular.closure} and Lemma~\ref{lemm.regular.between} reduce the problem of 
\begin{quote}
$p$ wishes to progress with value $v$
\end{quote}
to the simpler problem of 
\begin{quote}
$p$ wishes to find an open set that intersects with the community of $p$, and work with this open set to agree on $v$ (which open set does not matter; $p$ can try several until one works).
\end{quote}
Once this is done, the distributed algorithm will safely propagate the belief across the network.

Note that no forking is possible above (this is when a distributed system that was in agreement, partitions into subsets that are committed to incompatible values); all the action is in finding and convincing the $O\between \community(p)$, and then the rest is automatic.
\end{rmrk}

\jamiesubsection{Regularity, maximal topens, \& minimal closed neighbourhoods}
\label{subsect.reg.max.min}

\begin{rmrk}
\label{rmrk.arc}
Recall we have seen an arc of results which 
\begin{itemize*}
\item
started with Theorem~\ref{thrm.max.cc.char} and Corollary~\ref{corr.regular.is.regular} --- characterisations of regularity %
in terms of maximal topens --- and 
\item
passed through Proposition~\ref{prop.views.of.regularity} --- characterisation of weak regularity $p\in\community(p)\in\opens$ in terms of minimal closed neighbourhoods.
\end{itemize*}
We are now ready to complete this arc by stating and proving Theorem~\ref{thrm.up.down.char}.
This establishes a pleasing --- and not-at-all-obvious --- duality between `has a maximal topen neighbourhood' and `has a minimal closed neighbourhood'.
\end{rmrk}

\begin{thrm}
\label{thrm.up.down.char}
Suppose $(\ns P,\opens)$ is a semitopology and $p\in\ns P$.
Then the following are equivalent:
\begin{enumerate*}
\item\label{item.up.down.char.regular}
$p$ is regular.
\item\label{item.up.down.char.max}
$\community(p)$ is a maximal/greatest topen neighbourhood of $p$.
\item\label{item.up.down.char.wr.mcn}
$p$ is weakly regular (meaning that $p\in\community(p)=\interior(\intertwined{p})$) and $\intertwined{p}$ is a minimal closed neighbourhood (Definition~\ref{defn.cn}).\footnote{We really do mean ``$\intertwined{p}$ is minimal amongst closed neighbourhoods'' and \emph{not} the weaker condition ``$\intertwined{p}$ is minimal amongst closed neighbourhoods of $p$''!  That weaker condition is treated in Proposition~\ref{prop.views.of.regularity}.  See Remark~\ref{rmrk.don't.misread}.}
\end{enumerate*}
\end{thrm}
\begin{proof}
Equivalence of parts~\ref{item.up.down.char.regular} and~\ref{item.up.down.char.max} is just Theorem~\ref{thrm.max.cc.char}(\ref{char.Kp.greatest.topen}).

For equivalence of parts~\ref{item.up.down.char.max} and~\ref{item.up.down.char.wr.mcn} we prove two implications:
\begin{itemize}
\item
Suppose $p$ is regular.
By Lemma~\ref{lemm.wr.r}(\ref{item.r.implies.wr}) $p$ is weakly regular.
Now consider a closed neighbourhood $C'\subseteq \intertwined{p}$.
Note that $C'$ has a nonempty interior by Definition~\ref{defn.cn}(\ref{item.closed.neighbourhood}), so pick any $p'$ such that
$$
p'\in\interior(C')\subseteq C'\subseteq\intertwined{p} .
$$
It follows that $p'\in\interior(\intertwined{p})=\community(p)$, and $p$ is regular, so by Corollary~\ref{corr.corr.pKp} $\intertwined{p'}=\intertwined{p}$, 
and then by Proposition~\ref{prop.views.of.regularity}(\ref{item.intertwined.p.closed.neighbourhood.of.p}\&\ref{item.intertwined.p.least.in.poset.closed.neighbourhoods.of.p}) (since $p'{\in}\interior(C')$) $\intertwined{p'}\subseteq C'$.
Putting this all together we have
$$
\intertwined{p}=\intertwined{p'} \subseteq C' \subseteq\intertwined{p},
$$
so that $C'=\intertwined{p}$ as required.
\item
Suppose $p$ is weakly regular and suppose $\intertwined{p}$ is minimal in the poset of closed neighbourhoods ordered by subset inclusion.

Consider some $p'\in\community(p)$.
By Proposition~\ref{prop.intertwined.as.closure}(\ref{intertwined.as.closure.closed}) $\intertwined{p'}\subseteq\intertwined{p}$, and by minimality it follows that $\intertwined{p'}=\intertwined{p}$.
Thus also $\community(p')=\community(p)$.

Now $p'\in\community(p)$ was arbitrary, so by Corollary~\ref{corr.regular.is.regular} $p$ is regular as required.  
\qedhere\end{itemize}
\end{proof}

\begin{rmrk}
\label{rmrk.indeed.two.closed.neighbourhoods}
Recall Example~\ref{xmpl.p.not.regular}(\ref{item.p.not.regular.01234b}), as illustrated in Figure~\ref{fig.irregular} (right-hand diagram).
This has a point $0$ whose community $\community(0)=\{1,2\}$ is not a single topen (it contains two topens: $\{1\}$ and $\{2\}$).

A corollary of Theorem~\ref{thrm.up.down.char} is that $\intertwined{0}=\{0,1,2\}$ cannot be a minimal closed neighbourhood, because if it were then $0$ would be regular and $\community(0)$ would be a maximal topen neighbourhood of $0$.

We check, and see that indeed, $\intertwined{0}$ contains \emph{two} distinct minimal closed neighbourhoods: $\{0,1\}$ and $\{0,2\}$.
\end{rmrk}

\begin{rmrk}
\label{rmrk.don't.misread}
Theorem~\ref{thrm.up.down.char}(\ref{item.up.down.char.wr.mcn}) looks like Proposition~\ref{prop.views.of.regularity}(\ref{item.intertwined.p.least.in.poset.closed.neighbourhoods.of.p}), but
\begin{itemize*}
\item
Proposition~\ref{prop.views.of.regularity}(\ref{item.intertwined.p.least.in.poset.closed.neighbourhoods.of.p}) regards the \emph{poset of closed neighbourhoods of $p$} (closed sets with a nonempty open interior that contains $p$),
\item
Theorem~\ref{thrm.up.down.char}(\ref{item.up.down.char.wr.mcn}) regards the \emph{poset of all closed neighbourhoods} (closed sets with a nonempty open interior, not necessarily including $p$).
\end{itemize*}
So the condition used in Theorem~\ref{thrm.up.down.char}(\ref{item.up.down.char.wr.mcn}) is strictly stronger than the condition used in Proposition~\ref{prop.views.of.regularity}(\ref{item.intertwined.p.least.in.poset.closed.neighbourhoods.of.p}).
Correspondingly, the regularity condition in Theorem~\ref{thrm.up.down.char}(\ref{item.up.down.char.regular}) can be written as $p\in\community(p)\in\topens$, and (as noted in Lemma~\ref{lemm.wr.r} and Example~\ref{xmpl.wr}(\ref{item.wr.2})) this is strictly stronger than the condition $p\in\community(p)$ used in Proposition~\ref{prop.views.of.regularity}(\ref{item.views.of.regularity.wr}). 
\end{rmrk}

Corollary~\ref{corr.anti-hausdorff} makes Remark~\ref{rmrk.not.hausdorff} (intertwined is the opposite of Hausdorff) a little more precise:
\begin{corr}
\label{corr.anti-hausdorff}
Suppose $(\ns P,\opens)$ is a Hausdorff semitopology (so every two points have a pair of disjoint neighbourhoods).
Then if $p\in\ns P$ is regular, then $\{p\}$ is clopen.
\end{corr}
\begin{proof}
Suppose $\ns P$ is Hausdorff and consider $p\in \ns P$.
By Remark~\ref{rmrk.not.hausdorff} $\intertwined{p}=\{p\}$. 
From Theorem~\ref{thrm.up.down.char}(\ref{item.up.down.char.wr.mcn}) $\{p\}$ is closed and has a nonempty open interior which must therefore also be equal to $\{p\}$.
By Corollary~\ref{corr.when.singleton.topen} (or from Theorem~\ref{thrm.up.down.char}(\ref{item.up.down.char.max})) this interior is transitive.
\end{proof}

\begin{prop}
\label{prop.max.topen.min.closed}
Suppose $(\ns P,\opens)$ is a semitopology.
Then:
\begin{enumerate*}
\item\label{item.max.topen.min.closed.1}
Every maximal topen is equal to the interior of a minimal closed neighbourhood.
\item\label{item.max.topen.min.closed.2}
The converse implication holds if $(\ns P,\opens)$ is a topology, but need not hold in the more general case that $(\ns P,\opens)$ is a semitopology: there may exist a minimal closed neighbourhood whose interior is not topen.
\end{enumerate*}
\end{prop}
\begin{proof}
\leavevmode
\begin{enumerate}
\item
Suppose $\atopen$ is a maximal topen.
By Definition~\ref{defn.transitive}(\ref{transitive.cc}) $\atopen$ is nonempty, so choose $p\in \atopen$.
By Proposition~\ref{prop.intertwined.as.closure}(\ref{intertwined.p.closed}) $\intertwined{p}$ is closed, and using Theorem~\ref{thrm.max.cc.char} 
$$
p\in \atopen=\community(p)=\interior(\intertwined{p})\subseteq\intertwined{p}.
$$
Thus $p$ is weakly regular and by Proposition~\ref{prop.views.of.regularity}(\ref{item.views.of.regularity.wr}\&\ref{item.intertwined.p.least.in.poset.closed.neighbourhoods.of.p}) $\intertwined{p}$ is a least closed neighbourhood of $p$.
\item
It suffices to provide a counterexample.
This is Example~\ref{xmpl.not.intertwined} below.
However, we also provide here a breaking `proof', which throws light on precisely what Example~\ref{xmpl.not.intertwined} is breaking, and illustrates what the difference between semitopology and topology can mean in practical proof.

Suppose $\atopen=\interior(C)$ is the nonempty open interior of some minimal closed neighbourhood $C$: we will try (and fail) to show that this is transitive.
By Proposition~\ref{prop.cc.char} it suffices to prove that $p\intertwinedwith p'$ for every $p,p'\in \atopen$.

So suppose $p\in O$ and $p'\in O'$ and $O\notbetween O'$.
By Definition~\ref{defn.closure}(\ref{item.closure}) $p'\notin\closure{O}$, so that $\closure{O}\cap C\subseteq C$ is a strictly smaller closed set.
Also, $O\cap C$ is nonempty because it contains $p$.

If $(\ns P,\opens)$ is a topology then we are done, because $O\cap\atopen=\interior(O\cap C)$ would necessarily be open, contradicting our assumption that $C$ is a minimal closed neighbourhood. 

However, if $(\ns P,\opens)$ is a semitopology then this does not necessarily follow: $O\cap\atopen$ need not be open, and we cannot proceed.
\qedhere\end{enumerate}
\end{proof}

\begin{figure}
\vspace{-1em}
\centering
\includegraphics[width=0.4\columnwidth]{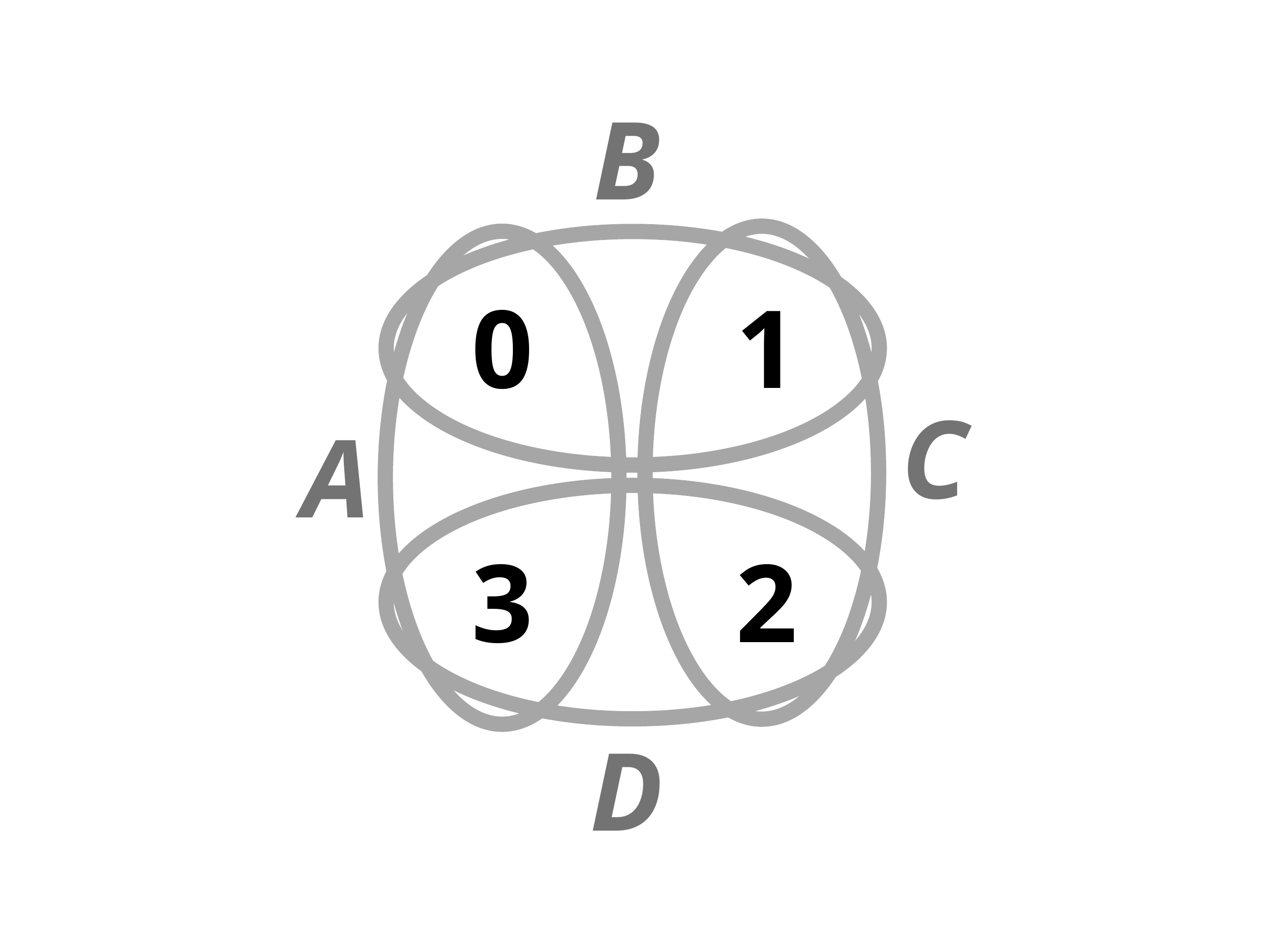}
\caption{An unconflicted, irregular space (Proposition~\ref{prop.unconflicted.irregular}) in which every minimal closed neighbourhood has a non-transitive open interior (Example~\ref{xmpl.not.intertwined})}
\label{fig.square.diagram}
\end{figure}

\begin{lemm}
\label{lemm.square.diagram.not.qr}
Consider the semitopology illustrated in Figure~\ref{fig.square.diagram}.
So:
\begin{itemize}
\item
$\ns P = \{0, 1, 2, 3\}$.
\item
$\opens$ is generated by $\{A,B,C,D\}$ where: 
$$
A=\{3, 0\}, 
\quad
B=\{0, 1\},
\quad
C=\{1, 2\},
\quad\text{and}\quad
D=\{2, 3\}.
$$
\end{itemize}
Then for every $p\in\ns P$ we have:
\begin{enumerate*}
\item\label{item.square.diagram.not.qr.1}
$p$ is intertwined only with itself.
\item\label{item.square.diagram.not.qr.2}
$\community(p)=\varnothing$.
\end{enumerate*}
\end{lemm}
\begin{proof}
Part~\ref{item.square.diagram.not.qr.1} is by routine calculations from Definition~\ref{defn.intertwined.points}(\ref{intertwined.defn}).
Part~\ref{item.square.diagram.not.qr.2} follows, noting that $\interior(\{p\})=\varnothing$ for every $p\in\ns P$.
\end{proof}

\begin{xmpl}
\label{xmpl.not.intertwined}
The semitopology illustrated in Figure~\ref{fig.square.diagram}, and specified in Lemma~\ref{lemm.square.diagram.not.qr},
contains sets that are minimal amongst closed sets with a nonempty interior, yet that interior is not topen:
\begin{itemize*}
\item
$A$, $B$, $C$, and $D$ are clopen, because $C$ is the complement of $A$ and $D$ is the complement of $B$, so they are their own interior.
\item
$A$ is a minimal closed neighbourhood (which is also open, being $A$ itself), because 
\begin{itemize*}
\item
$A=\{3, 0\}$ is closed because it is the complement of $C$, and it is its own interior, and 
\item
its two nonempty subsets $\{3\}$ and $\{0\}$ are closed (being the complement of $B\cup C$ and $C\cup D$ respectively) but they have empty open interior because $\{3\}$ and $\{0\}$ are not open.
\end{itemize*} 
\item
$A$ is not transitive because $3$ and $0$ are not intertwined: $3\in D$ and $0\in B$ and $B\cap D=\varnothing$.
\item
Similarly $B$, $C$, and $D$ are minimal closed neighbourhoods, which are also open, and they are not transitive.
\end{itemize*}
We further note that:
\begin{enumerate*}
\item
$\closure{0}=\{0\}$, because its complement is equal to $C\cup D$ (Definition~\ref{defn.closure}; Lemma~\ref{lemm.closed.complement.open}).
Similarly for every other point in $\ns P$.
\item
$\intertwined{0}=\{0\}$, as noted in Lemma~\ref{lemm.square.diagram.not.qr}.
Similarly for every other point in $\ns P$.
\item\label{item.square.diagram.not.regular}
$\community(0)=\interior(\intertwined{0})=\varnothing$ as noted in Lemma~\ref{lemm.square.diagram.not.qr},
so that $0$ is not regular (Definition~\ref{defn.tn}(\ref{item.tn})), and $0$ is not even weakly regular or quasiregular.
Similarly for every other point in $\ns P$.
\item
$0$ has \emph{two} minimal closed neighbourhoods: $A$ and $B$.
Similarly for every other point in $\ns P$.
\end{enumerate*}
This illustrates that $\intertwined{p}\subsetneq C$ is possible, where $C$ is a minimal closed neighbourhood of $p$.
\end{xmpl}

\begin{rmrk}
The results and discussions above tell us something interesting above and beyond the specific mathematical facts which they express.

They demonstrate that points being intertwined (the $p\intertwinedwith p'$ from Definition~\ref{defn.intertwined.points}) is a distinct \emph{semitopological} notion. 
A reader familiar with topology might be tempted to identify maximal topens with interiors of minimal closed neighbourhood (so that in view of Proposition~\ref{prop.cc.char}, being intertwined would be topologically characterised just as two points being in the interior of the same minimal closed neighbourhood).

This works in topologies, but we see from Example~\ref{xmpl.not.intertwined} that in semitopologies being intertwined has its own distinct identity.
\end{rmrk}

We conclude with one more example, showing how an (apparently?) slight change to a semitopology can make a big difference to its intertwinedness:
\begin{xmpl}
\label{xmpl.two.topen.examples}
\leavevmode
\begin{enumerate*}
\item\label{item.two.topen.examples.1}
$\mathbb Q^2$ with open sets generated by any covering collection of pairwise non-parallel \deffont{rational lines} --- meaning a set of solutions to a linear equation $a.x\plus b.y=c$ for $a$, $b$, and $c$ integers --- is a semitopology.

This consists of a single (maximal) topen: lines are pairwise non-parallel, so any two lines intersect and (looking to Proposition~\ref{prop.cc.char}) all points are intertwined.
There is only one closed set with a nonempty open interior, which is the whole space.
\item\label{item.two.topen.examples.2}
$\mathbb Q^2$ with open sets generated by all (possibly parallel) rational lines, is a semitopology.
It has no topen sets and (looking to Proposition~\ref{prop.cc.char}) no two distinct points are intertwined.

For any line $l$, its complement $\mathbb Q^2\setminus l$ is a closed set, given by the union of all the lines parallel to $l$.
Thus every closed set is also an open set, and vice versa, and every line $l$ is an example of a minimal closed neighbourhood (itself), whose interior is not a topen. 
\end{enumerate*}
\end{xmpl}

\jamiesubsection{More on minimal closed neighbourhoods}

We make good use of closed neighbourhoods, and in particular minimal closed neighbourhoods, in Subsection~\ref{subsect.reg.max.min} and elsewhere.
We take a moment to give a pleasing alternative characterisation of this useful concept. 

\jamiesubsubsection{Regular open/closed sets}

\begin{rmrk}
The terminology `regular open/closed set' is from the topological literature.
It is not directly related to terminology `regular point' from Definition~\ref{defn.tn}(\ref{item.regular.point}), which comes from semitopologies.
However, it turns out that a mathematical connection does exist between these two notions. 
We outline some theory of regular open/closed sets, and then demonstrate the connections to what we have seen in our semitopological world. 
\end{rmrk}

\begin{defn}
\label{defn.regular.open.set}
Suppose $(\ns P,\opens)$ is a semitopology.
Recall some standard terminology from topology~\cite[Exercise~3D, page~29]{willard:gent}:
\begin{enumerate*}
\item
We call an open set $O\in\opens$ a \deffont{regular open set} when $O=\interior(\closure{O})$.
\item
We call a closed set $C\in\closed$ a \deffont{regular closed set} when $C=\closure{\interior(C)}$.
\item
Write $\regularOpens$ and $\regularClosed$ for the sets of regular open and regular closed sets respectively.
\end{enumerate*}
\end{defn}

\begin{lemm}
\label{lemm.ic.ci.regular}
Suppose $(\ns P,\opens)$ is a semitopology and $O\in\opens$ and $C\in\closed$.
Then:
\begin{enumerate*}
\item\label{item.ic.ci.regular.open}
$\interior(C)$ is a regular open set.
\item\label{item.ic.ci.regular.closed}
$\closure{O}$ is a regular closed set.
\end{enumerate*}
\end{lemm}
\begin{proof}
Direct from Definition~\ref{defn.regular.open.set} and Corollary~\ref{corr.ic.ci}.
\end{proof}

\begin{corr}
\label{corr.community.regular.open}
Suppose $(\ns P,\opens)$ is a semitopology and $p\in\ns P$.
Then $\community(p)\in\regularOpens$. 
\end{corr}
\begin{proof}
We just combine Lemma~\ref{lemm.ic.ci.regular}(\ref{item.ic.ci.regular.open}) with Proposition~\ref{prop.intertwined.as.closure}(\ref{intertwined.p.closed}).
\end{proof}

\begin{corr}
\label{corr.interior.closure.regular}
Suppose $(\ns P,\opens)$ is a semitopology and $O\in\opens$.
Then $\interior(\closure{O})$ is a regular open set.
\end{corr}
\begin{proof}
By Lemma~\ref{lemm.closure.closed} $\closure{O}$ is closed, and by Lemma~\ref{lemm.ic.ci.regular} $\interior(\closure{O})$ is regular open. 
\end{proof}

The regular open and the regular closed sets are the same thing, up to an easy and natural bijection: 
\begin{corr}
\label{corr.ro=rc}
Suppose $(\ns P,\opens)$ is a semitopology.
Then 
\begin{itemize*}
\item
the topological closure map $\closure{\text{-}}$ and 
\item
the topological interior map $\interior(\text{-})$ 
\end{itemize*}
define a bijection of posets between $\regularOpens$ and $\regularClosed$ ordered by subset inclusion. 
\end{corr}
\begin{proof}
By Lemma~\ref{lemm.ic.ci.regular}, $\closure{\text{-}}$ and $\interior(\text{-})$ map between $\regularOpens$ to $\regularClosed$.
Now we note that the regularity property from Definition~\ref{defn.regular.open.set}, which states that $\interior(\closure{O})=O$ when $O\in\regularOpens$ and $\closure{\interior(C)}=C$ when $C\in\regularClosed$, expresses precisely that these maps are inverse.

They are maps of posets by Corollary~\ref{corr.interior.monotone} and Lemma~\ref{lemm.closure.monotone}(\ref{closure.increasing}). 
\end{proof}

\begin{lemm}
\label{lemm.regular.open.closed}
Suppose $(\ns P,\opens)$ is a semitopology and $O\in\opens$ and $C\in\closed$.
Then:
\begin{enumerate*}
\item
$O$ is a regular open set if and only if $\ns P\setminus O$ is a regular closed set if and only if $\closure{O}$ is a regular closed set.
\item
$C$ is a regular closed set if and only if $\ns P\setminus C$ is a regular open set if and only if $\interior(C)$ is a regular open set.
\end{enumerate*}
\end{lemm} 
\begin{proof}
By routine calculations from the definitions using parts~\ref{item.closure.interior.complement.closure} and~\ref{item.closure.interior.complement.interior} of Lemma~\ref{lemm.closure.interior}.
\end{proof}

\jamiesubsubsection{Intersections of regular open sets}

An easy observation about open sets will be useful:
\begin{lemm}
\label{lemm.clint.between}
Suppose $(\ns P,\opens)$ is a semitopology and $O,O'\in\opens$.
Then the following are equivalent:
\begin{enumerate*}
\item\label{item.client.between.1} 
$O\between O'$.
\item\label{item.client.between.2} 
$O\between\interior(\closure{O'})$.
\item\label{item.client.between.3} 
$\interior(\closure{O})\between\interior(\closure{O'})$.
\end{enumerate*}
\end{lemm}
\begin{proof}
First we prove the equivalence of parts~\ref{item.client.between.1} and~\ref{item.client.between.2}:
\begin{enumerate}
\item
Suppose $O\between O'$.
By Lemma~\ref{lemm.closure.interior}(\ref{item.closure.interior.open}) $O\between \interior(\closure{O'})$.
\item
Suppose there is some $p\in O\cap\interior(\closure{O'})$.
Then $O$ is an open neighbourhood of $p$ and $p\in\closure{O'}$, so by Definition~\ref{defn.closure}(\ref{item.closure}) $O\between O'$ as required.\footnote{Lemma~\ref{lemm.closure.using.nbhd.intersections} packages this argument up nicely with some slick notation, which we have not yet set up.}
\end{enumerate}
Equivalence of parts~\ref{item.client.between.1} and~\ref{item.client.between.3} then follows easily by two applications of the equivalence of parts~\ref{item.client.between.1} and~\ref{item.client.between.2}.
\end{proof}

\begin{rmrk}
\label{rmrk.pi-base}
Lemma~\ref{lemm.clint.between} is true in topologies as well, but it is not prominent in the literature.
Two standard reference works~\cite{engelking:gent,willard:gent} do not seem to mention it.
It appears as equation~10 in Theorem~1.37 of~\cite{koppelberg:hanba1}, and as a lemma in $\pi$-base\footnoteref{https://topology.pi-base.org/theorems/T000420}{https://web.archive.org/web/20240108192930/https://topology.pi-base.org/theorems/T000420} (we thank the mathematics StackExchange community for the pointers).  
We mention this to note an interesting contrast: this result is as true in topologies as it is in semitopologies, but somehow, it \emph{matters} more in the latter than the former.
\end{rmrk}

\begin{corr}
\label{corr.nonintersect.nonintersect.regular}
Suppose $(\ns P,\opens)$ is a semitopology and $p,p'\in\ns P$.
Then the following conditions are equivalent:
\begin{enumerate*}
\item\label{item.nonintersect.nonintersect.regular.1}
$p$ and $p'$ have a nonintersecting pair of open neighbourhoods.
\item\label{item.nonintersect.nonintersect.regular.2}
$p$ and $p'$ have a nonintersecting pair of regular open neighbourhoods.
\end{enumerate*}
\end{corr}
\begin{proof}
Part~\ref{item.nonintersect.nonintersect.regular.2} clearly implies part~\ref{item.nonintersect.nonintersect.regular.1}, since a regular open set is an open set.
Part~\ref{item.nonintersect.nonintersect.regular.1} implies part~\ref{item.nonintersect.nonintersect.regular.2} using Lemma~\ref{lemm.clint.between} and Corollary~\ref{corr.interior.closure.regular}.
\end{proof}

\begin{rmrk}
\label{rmrk.intertwined.with.regular.opens}
In Definition~\ref{defn.intertwined.points}(\ref{item.p.intertwinedwith.p'}) we defined $p\intertwinedwith p'$ in terms of open neighbourhoods of $p$ and $p'$ as follows:
$$
\Forall{O,O'{\in}\opens} (p\in O\land p'\in O') \limp O\between O' .
$$ 
In the light of Corollary~\ref{corr.nonintersect.nonintersect.regular}, we could just as well have defined it just in terms of regular open neighbourhoods: 
$$
\Forall{O,O'{\in}\regularOpens} (p\in O\land p'\in O') \limp O\between O' .
$$ 
Mathematically, for what we have needed so far, this latter characterisation is not needed.
However, it is easy to think of scenarios in which it might be useful.
In particular, \emph{computationally} it could make sense to restrict to the regular open sets, simply because there are fewer of them. 
\end{rmrk}

\jamiesubsubsection{Minimal nonempty regular closed sets are precisely the minimal closed neighbourhoods}

\begin{lemm}
\label{lemm.lcn.nrc}
Suppose $(\ns P,\opens)$ is a semitopology and $C\in\closed$.
Then:
\begin{enumerate*}
\item\label{item.lcn.nrc.1}
If $C$ is a minimal closed neighbourhood (a closed set with a nonempty open interior), then $C$ is a nonempty regular closed set (Definition~\ref{defn.regular.open.set}).
\item\label{item.lcn.nrc.2}
If $C$ is a nonempty regular closed set then $C$ is a closed neighbourhood (Definition~\ref{defn.cn}).
\end{enumerate*}
\end{lemm}
\begin{proof}
We consider each part in turn:
\begin{enumerate}
\item
\emph{Suppose $C$ is a minimal closed neighbourhood.}

Write $O'=\interior(C)$ and $C'=\closure{O'}=\closure{\interior(C)}$.
Because $C$ is a closed neighbourhood, by Definition~\ref{defn.cn} $O'\neq\varnothing$.
By Lemma~\ref{lemm.closure.closed} $C'\in\closed$.
Using Corollary~\ref{corr.ic.ci} $\interior(C')=\interior(\closure{\interior(C)})=\interior(C)=O'\neq\varnothing$, so that $C'$ is a closed neighbourhood, and by minimality $C'=C$.
But then $C=\closure{\interior(C)}$ so $C$ is regular, as required.
\item
\emph{Suppose $C$ is a nonempty regular closed set,} so that $\varnothing\neq C=\closure{\interior(C)}$.

It follows that $\interior(C)\neq\varnothing$ and this means precisely that $C$ is a closed neighbourhood. 
\qedhere\end{enumerate}
\end{proof}

In Theorem~\ref{thrm.up.down.char} we characterised the point $p$ being regular in terms of minimal closed neighbourhoods.
We can now characterise the minimal closed neighbourhoods in terms of something topologically familiar:
\begin{prop}
\label{prop.lnrc.lcn}
Suppose $(\ns P,\opens)$ is a semitopology and $C\in\closed$.
Then the following are equivalent:
\begin{enumerate*}
\item
$C$ is a minimal nonempty regular closed set. 
\item
$C$ is a minimal closed neighbourhood. 
\end{enumerate*}
\end{prop}
\begin{proof}
We prove two implications:
\begin{itemize}
\item
\emph{Suppose $C$ is a minimal closed neighbourhood.}

By Lemma~\ref{lemm.lcn.nrc}(\ref{item.lcn.nrc.1}) $C$ is a nonempty regular closed set.
Furthermore by Lemma~\ref{lemm.lcn.nrc}(\ref{item.lcn.nrc.2}) if $C'\subseteq C$ is any other nonempty regular closed set contained in $C$, then it is a closed neighbourhood, and by minimality it is equal to $C$.
Thus, $C$ is minimal.
\item
\emph{Suppose $C$ is a minimal nonempty regular closed set.}

By Lemma~\ref{lemm.lcn.nrc}(\ref{item.lcn.nrc.2}) $C$ is a closed neighbourhood.
Furthermore by Lemma~\ref{lemm.lcn.nrc}(\ref{item.lcn.nrc.1}) if $C'\subseteq C$ is any other closed neighbourhood then it is a nonempty regular closed set, and by minimality it is equal to $C$.
\qedhere\end{itemize}
\end{proof}

\jamiesubsection{How are $\intertwined{p}$ and $\closure{p}$ related?}

\begin{rmrk}
\label{rmrk.re-read.closure}
Recall the definitions of $\intertwined{p}$ and $\closure{p}$:
\begin{itemize*}
\item
The set $\closure{p}$ is the \emph{closure} of $p$.

By Definition~\ref{defn.closure} this is the set of $p'$ such that every open neighbourhood $O'\ni p'$ intersects with $\{p\}$.
By Definition~\ref{defn.closed} $\closure{p}$ is closed.
\item
The set $\intertwined{p}$ is the set of points \emph{intertwined} with $p$.

By Definition~\ref{defn.intertwined.points}(\ref{intertwined.defn}) this is the set of $p'$ such that every open neighbourhood $O'\ni p'$ intersects with every open neighbourhood $O \ni p$. 
By Proposition~\ref{prop.intertwined.as.closure}(\ref{intertwined.p.closed}) $\intertwined{p}$ is closed.
\end{itemize*}
So we see that $\closure{p}$ and $\intertwined{p}$ give us two canonical ways of generating a closed set from a point $p\in \ns P$. 
This invites a question: 
\begin{quote}
\emph{How are $\intertwined{p}$ and $\closure{p}$ related?}
\end{quote}
\end{rmrk}

Lemma~\ref{lemm.char.not.intertwined} rephrases Remark~\ref{rmrk.re-read.closure} more precisely by looking at it through sets complements.
\begin{lemm}
\label{lemm.char.not.intertwined}
Suppose $(\ns P,\opens)$ is a semitopology and $p\in\ns P$.
Then:
\begin{enumerate*}
\item
$\ns P\setminus\closure{p} = \bigcup \{O\in\opens \mid p\notin O\}\oldin\opens$.
\item\label{item.intertwined.open.avoid}
$\ns P\setminus\intertwined{p} = \bigcup\{O'\in\opens \mid \Exists{O{\in}\opens} p\in O\land O'\notbetween O\}\oldin\opens$.
\item
$\ns P\setminus\intertwined{p} = \bigcup\{O\in\opens \mid p\notin \closure{O}\}\oldin\opens$.
\end{enumerate*}
In words, we can say: $\ns P\setminus\closure{p}$ is the union of the open sets such that $p$ avoids them, and $\ns P\setminus\intertwined{p}$ is the union of the open sets such that $p$ avoids their closures.
\end{lemm} 
\begin{proof}
\leavevmode
\begin{enumerate*}
\item
Immediate from Definitions~\ref{defn.intertwined.points} and~\ref{defn.closure}.\footnote{A longer proof via Corollary~\ref{corr.closure.closure}(\ref{item.closure.as.intersection}) and Lemma~\ref{lemm.closed.complement.open} is also possible.}
Openness is from Definition~\ref{defn.semitopology}(\ref{semitopology.unions}).
\item
By a routine argument direct from Definition~\ref{defn.intertwined.points}. 
Openness is from Definition~\ref{defn.semitopology}(\ref{semitopology.unions}).
\item
Rephrasing part~\ref{item.intertwined.open.avoid} of this result using Definition~\ref{defn.closure}(\ref{item.closure}).
\qedhere\end{enumerate*}
\end{proof}

\begin{prop}
\label{prop.closure.intertwined}
Suppose $(\ns P,\opens)$ is a semitopology and $p\in\ns P$.
Then:
\begin{enumerate*}
\item\label{item.closure.intertwined.1}
$\closure{p}\subseteq \intertwined{p}$.
\item\label{item.closure.intertwined.2}
The subset inclusion may be strict; that is, $\closure{p}\subsetneq\intertwined{p}$ is possible --- even if $p$ is regular (Definition~\ref{defn.tn}(\ref{item.regular.point})).
\item\label{item.closure.intertwined.3}
If $\interior(\closure{p})\neq\varnothing$ (so $\closure{p}$ has a nonempty interior)
then 
$\closure{p}=\intertwined{p}$.
\end{enumerate*}
\end{prop}
\begin{proof}
\leavevmode
\begin{enumerate}
\item
We reason as follows:
$$
\begin{array}{r@{\ }l@{\quad}l}
\closure{p}=&
\closure{\{p\}}
&\text{Definition~\ref{defn.closure}(\ref{item.closure.p})}
\\
=&
\bigcap\{C\in\closed \mid p\in C\}
&\text{Corollary~\ref{corr.closure.closure}(\ref{item.closure.as.intersection})}
\\
\subseteq&
\bigcap\{C\in\closed \mid p\in\interior(C)\}
&\text{Fact of intersections}
\\
=&
\intertwined{p} 
&\text{Proposition~\ref{prop.intertwined.as.closure}(\ref{intertwined.as.closure.closed})}
\end{array}
$$
\item
Example~\ref{xmpl.closure.101} below shows that $\closure{p}\subsetneq\intertwined{p}$ is possible for $p$ regular. 
\item
Write $O=\interior(\closure{p})$.
By standard topological reasoning, $\closure{p}$ is the complement of the union of the open sets that do not contain $p$, and $O=\interior(\closure{p})$ is the greatest open set such that $\Forall{O'{\in}\opens}O\between O'\limp p\in O'$.  
We assumed that $O$ is nonempty, so $O\between O$, thus $p\in O$.

Then by part~\ref{item.closure.intertwined.1} of this result $p\in O\subseteq\closure{p}\subseteq\intertwined{p}$, and by Proposition~\ref{prop.regular.closure}(\ref{item.regular.closure.2}) $\intertwined{p}=\closure{O}$.
Using more standard topological reasoning (since $O\neq\varnothing$) $\closure{O}=\closure{p}$, and the result follows.
\qedhere\end{enumerate}
\end{proof}

\begin{figure}
\centering
\includegraphics[width=0.4\columnwidth,trim={50 150 50 150},clip]{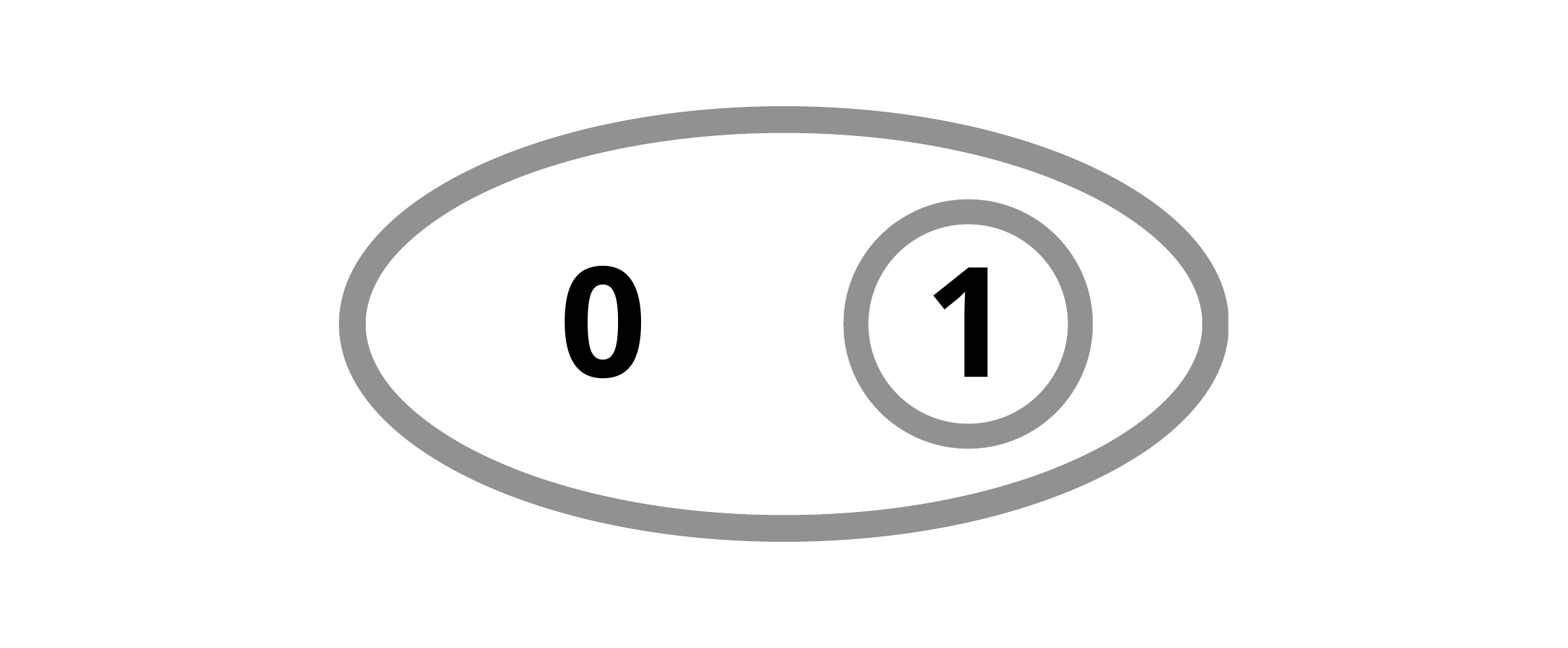}
\caption{The Sierpi\'nski space $\tf{Sk}$ (Example~\ref{xmpl.sk})}
\label{fig.sierpinski}
\end{figure}

\begin{xmpl}
\label{xmpl.closure.101}
\label{xmpl.sk}
Define $\tf{Sk}$ the \deffont{Sierpi\'nski space}~\cite[Example~3.2(e)]{willard:gent} by $\ns P=\{0,1\}$ and $\opens=\{\varnothing,\{1\},\{0,1\}\}$, as illustrated in Figure~\ref{fig.sierpinski}. 
Then:
\begin{itemize*}
\item
$\closure{0}=\{0\}$ (because $\{1\}$ is open), but
\item
$\intertwined{0}=\{0,1\}$ (because every open neighbourhood of $0$ intersects with every open neighbourhood of $1$). 
\end{itemize*}
Thus we see that $\closure{0}=\{0\}\subsetneq\{0,1\}=\intertwined{0}$, and $0$ is regular since $0\in\interior(\intertwined{0})=\{0,1\}\in\topens$.
\end{xmpl}

\begin{rmrk}
We have one loose end left.
We know from Theorem~\ref{thrm.up.down.char}(\ref{item.up.down.char.wr.mcn}) that $\intertwined{p}$ is a minimal closed neighbourhood (closed set with nonempty open interior) when $p$ is regular. 
We also know from Proposition~\ref{prop.closure.intertwined} that $\closure{p}\subsetneq\intertwined{p}$ is possible, even if $p$ is regular.

So a closed \emph{neighbourhood} in between $\closure{p}$ and $\intertwined{p}$ is impossible by minimality, but can there be any closed \emph{sets} (not necessarily having a nonempty open interior) in between $\closure{p}$ and $\intertwined{p}$?

Somewhat counterintuitively perhaps, this is possible: 
\end{rmrk}

\begin{lemm}
Suppose $(\ns P,\opens)$ is a semitopology and $p\in\ns P$. 
Then it is possible for there to exist a closed set $C\subseteq\ns P$ with $\closure{p}\subsetneq C\subsetneq\intertwined{p}$, even if $p$ is regular.
\end{lemm}
\begin{proof}
It suffices to provide an example.
Consider $\mathbb N$ with the semitopology whose open sets are generated by 
\begin{itemize*}
\item
final segments $n_\geq=\{n'\in\mathbb N\mid n'\geq n\}$ for $n\in\mathbb N$ (cf. Example~\ref{xmpl.meet-irreducible}(\ref{item.final.N})), and 
\item
$\{0,1,2,3,4,5,6,7,8,9\}$.
\end{itemize*} 
The reader can check that $\closure{0}=\{0\}$ and $\intertwined{0}=\{0,1,2,3,4,5,6,7,8,9\}$.
However, there are also eight closed sets $\{0,1\}$, $\{0,1,2\}$, \dots, $\{0,1,2,3,\dots,8\}$ in between $\closure{0}$ and $\intertwined{0}$. 
\end{proof}

We will study $\intertwined{p}$ further but to make more progress we need the notion of a(n un)conflicted point.
This is an important idea in its own right:

\jamiesection{(Un)conflicted points: transitivity of $\intertwinedwith$}
\label{sect.unconflicted.point}

\jamiesubsection{The basic definition} 
\label{subsect.reg.tra.int}

In Lemma~\ref{lemm.intertwined.not.transitive} we asked whether the `is intertwined with' relation $\intertwinedwith$ from Definition~\ref{defn.intertwined.points}(\ref{item.p.intertwinedwith.p'}) is transitive --- answer: not necessarily.

Transitivity of $\intertwinedwith$ is a natural condition.
We now have enough machinery to study it in more detail, and this will help us gain a deeper understanding of the properties of not-necessarily-regular points.

\begin{defn}
\label{defn.conflicted}
Suppose $(\ns P,\opens)$ is a semitopology.
\begin{enumerate*}
\item\label{item.conflicted.point}
Call $p$ a \deffont{conflicted point} when there exist $p'$ and $p''$ such that $p'\intertwinedwith p$ and $p\intertwinedwith p''$ yet $\neg(p'\intertwinedwith p'')$.
\item\label{item.unconflicted}
If $p'\intertwinedwith p\intertwinedwith p''$ implies $p'\intertwinedwith p''$ always, then call $p$ an \deffont{unconflicted point}.
\item
Continuing Definition~\ref{defn.tn}(\ref{item.regular.S}), if $P\subseteq\ns P$ and every $p\in P$ is conflicted/unconflicted, then we may call $P$ a \deffont{conflicted/unconflicted set} respectively. 
\end{enumerate*}
\end{defn}

\begin{xmpl}
\label{xmpl.conflicted.points}
We consider some examples:
\begin{enumerate*}
\item\label{item.example.of.conflicted.point}
In Figure~\ref{fig.012} top-left diagram, $0$ and $2$ are unconflicted and intertwined with themselves, and $1$ is conflicted (being intertwined with $0$, $1$, and $2$).

If the reader wants to know what a conflicted point looks like: it looks like $1$. 
\item 
In Figure~\ref{fig.012} top-right diagram, $0$ and $2$ are unconflicted and intertwined with themselves, and $1$ is conflicted (being intertwined with $0$, $1$, and $2$).
\item
In Figure~\ref{fig.012} lower-left diagram, $0$ and $1$ are unconflicted and intertwined with themselves, and $3$ and $4$ are unconflicted and intertwined with themselves, and $2$ is conflicted (being intertwined with $0$, $1$, $2$, $3$, and $4$).
\item
In Figure~\ref{fig.012} lower-right diagram, all points are unconflicted, and $0$ and $2$ are intertwined just with themselves, and $1$ and $\ast$ are intertwined with one another.
\item
In Figure~\ref{fig.square.diagram}, all points are unconflicted and intertwined only with themselves.
\end{enumerate*}
\end{xmpl}

So $p$ is conflicted when it witnesses a counterexample to $\intertwinedwith$ being transitive.
We start with an easy lemma (we will use this later, but we mention it now for Remark~\ref{rmrk.intertwined.unconflicted.in.context}):
\begin{lemm}
\label{lemm.unconflicted.char}
Suppose $(\ns P,\opens)$ is a semitopology and $p\in\ns P$.
Then the following are equivalent:
\begin{enumerate*}
\item\label{item.unconflicted.char.1}
$p$ is unconflicted.
\item\label{item.unconflicted.p.in.q}
If $q\in\ns P$ and $p\in\intertwined{q}$ then $\intertwined{p}\subseteq\intertwined{q}$. 
\item\label{item.p'.in.unconflicted.p}
$\intertwined{p}\subseteq\intertwined{p'}$ for every $p'\in\intertwined{p}$.
\item\label{item.unconflicted.as.least}
$\intertwined{p}$ is least in the set $\{\intertwined{p'}\mid p\intertwinedwith p'\}$ ordered by subset inclusion.
\end{enumerate*}
\end{lemm}
\begin{proof}
The proof is just by pushing definitions around in a cycle of implications.
\begin{itemize}
\item
\emph{Part~\ref{item.unconflicted.char.1} implies part~\ref{item.unconflicted.p.in.q}.}

Suppose $p$ is unconflicted.
Consider $q\in\ns P$ such that $p\in\intertwined{q}$, and consider any $p'\in\intertwined{p}$.
Unpacking definitions we have that $p'\intertwinedwith p\intertwinedwith q$ and so $p'\intertwinedwith q$, thus $p'\in\intertwined{q}$ as required.
\item
\emph{Part~\ref{item.unconflicted.p.in.q} implies part~\ref{item.p'.in.unconflicted.p}.}

From the fact that $p'\in\intertwined{p}$ if and only if $p'\intertwinedwith p$ if and only if $p\in\intertwined{p'}$.
\item
\emph{Part~\ref{item.p'.in.unconflicted.p} implies part~\ref{item.unconflicted.as.least}.}

Part~\ref{item.unconflicted.as.least} just rephrases part~\ref{item.p'.in.unconflicted.p}.
\item
\emph{Part~\ref{item.unconflicted.as.least} implies part~\ref{item.unconflicted.char.1}.}

Suppose $\intertwined{p}$ is $\subseteq$-least in $\{\intertwined{p'}\mid p\intertwinedwith p'\}$ and suppose $p''\intertwinedwith p\intertwinedwith p'$.
Then $p''\in\intertwined{p}\subseteq\intertwined{p'}$, so $p''\intertwinedwith p'$ as required.
\qedhere\end{itemize}
\end{proof}

\begin{rmrk}
\label{rmrk.intertwined.unconflicted.in.context}
Lemma~\ref{lemm.unconflicted.char} is just an exercise in reformulating definitions, but part~\ref{item.unconflicted.as.least} of the result helps us to contrast the property of being unconflicted, with structurally similar 
characterisations of \emph{weak regularity} and of \emph{regularity} in Proposition~\ref{prop.views.of.regularity} and Theorem~\ref{thrm.up.down.char} respectively.
For the reader's convenience we collect them here --- all sets below are ordered by subset inclusion:
\begin{enumerate}
\item
$p$ is unconflicted when \emph{$\intertwined{p}$ is least in $\{\intertwined{p'}\mid p\intertwinedwith p'\}$}. 
\item
$p$ is weakly regular when \emph{$\intertwined{p}$ is least amongst closed neighbourhoods of $p$}.

See Proposition~\ref{prop.views.of.regularity} and recall from Definition~\ref{defn.cn} that a closed neighbourhood of $p$ is a closed set $C$ such that $p\in\interior(C)$.
\item 
$p$ is regular when \emph{$\intertwined{p}$ is a closed neighbourhood of $p$ and minimal amongst all closed neighbourhoods}.

See Theorem~\ref{thrm.up.down.char} and recall that a closed neighbourhood is any closed set with a nonempty interior (not necessarily containing $p$).
\end{enumerate}
We know from Lemma~\ref{lemm.wr.r}(\ref{item.r.implies.wr}) that regular implies weakly regular. 
We now consider how these properties relate to being unconflicted.
\end{rmrk}

\jamiesubsection{Regular = weakly regular + unconflicted}
\label{subsect.r=wr+uc}

\begin{prop}
\label{prop.unconflicted.irregular}
Suppose $(\ns P,\opens)$ is a semitopology and $p\in\ns P$.
Then:
\begin{enumerate*}
\item\label{item.reg.implies.unconflicted}
If $p$ is regular then it is unconflicted.

Equivalently by the contrapositive: if $p$ is conflicted then it is not regular.
\item\label{item.unconflicted.irregular.2}
$p$ may be unconflicted and neither quasiregular, weakly regular, nor regular.
\item\label{item.unconflicted.irregular.3}
There exists a semitopological space such that 
\begin{itemize*}
\item
every point is unconflicted (so $\intertwinedwith$ is a transitive relation) yet 
\item
every point has empty community, so that the space is irregular, not weakly regular, and not quasiregular.%
\footnote{See also Proposition~\ref{prop.conflicted.weakly.regular}.}
\end{itemize*}
\end{enumerate*}
\end{prop}
\begin{proof}
We consider each part in turn:
\begin{enumerate}
\item
So consider $q\intertwinedwith p \intertwinedwith q'$.
We must show that $q\intertwinedwith q'$, so consider open neighbourhoods $Q\ni q$ and $Q'\ni q'$.
By assumption $p$ is regular, so unpacking Definition~\ref{defn.tn}(\ref{item.regular.point}) $p\in\community(p)\in\topens$.
From
$$
q\intertwinedwith p\intertwinedwith q'
\quad\text{if follows that}\quad
Q\between \community(p)\between Q',
$$
and by transitivity of $\community(p)$ (Definition~\ref{defn.transitive}(\ref{transitive.transitive})) we have $Q\between Q'$ as required.
\item
Consider the semitopology illustrated in Figure~\ref{fig.square.diagram}.
By Lemma~\ref{lemm.square.diagram.not.qr} the point $0$ is trivially unconflicted (because it is intertwined only with itself), but it is also neither quasiregular, weakly regular, nor regular, because its community is the empty set. 
See also Example~\ref{xmpl.boundary.examples}. 
\item
As for the previous part, noting that the same holds of points $1$, $2$, and $3$ in Figure~\ref{fig.square.diagram}.
\qedhere\end{enumerate}
\end{proof}

We can combine Proposition~\ref{prop.unconflicted.irregular} with a previous result Lemma~\ref{lemm.wr.r} to get a precise and attractive relation between being 
\begin{itemize*}
\item
regular (Definition~\ref{defn.tn}(\ref{item.regular.point})), 
\item
weakly regular (Definition~\ref{defn.tn}(\ref{item.weakly.regular.point})), and 
\item
unconflicted (Definition~\ref{defn.conflicted}), 
\end{itemize*}
as follows:
\begin{thrm}
\label{thrm.r=wr+uc}
Suppose $(\ns P,\opens)$ is a semitopology and $p\in\ns P$.
Then the following are equivalent:
\begin{itemize*}
\item
$p$ is regular.
\item
$p$ is weakly regular and unconflicted.
\end{itemize*}
More succinctly we can write: \emph{regular = weakly regular + unconflicted}.\footnote{See also a similar result Theorem~\ref{thrm.regular=qr+sc}, and a discussion in Remark~\ref{rmrk.two.char.r}.}
\end{thrm}
\begin{proof}
We prove two implications:
\begin{itemize}
\item
If $p$ is regular then it is weakly regular by Lemma~\ref{lemm.wr.r} and unconflicted by Proposition~\ref{prop.unconflicted.irregular}(\ref{item.reg.implies.unconflicted}). 
\item
Suppose $p$ is weakly regular and unconflicted.
By Definition~\ref{defn.tn}(\ref{item.weakly.regular.point}) $p\in\community(p)$ and by Lemma~\ref{lemm.three.transitive} it would suffice to show that $q\intertwinedwith q'$ for any $q,q'\in\community(p)$.

So consider $q,q'\in\community(p)$.
Now by Definition~\ref{defn.tn}(\ref{item.tn}) $\community(p)=\interior(\intertwined{p})$ so in particular $q,q'\in\intertwined{p}$.
Thus $q\intertwinedwith p\intertwinedwith q'$, and since $p$ is unconflicted $q\intertwinedwith q'$ as required.
\qedhere\end{itemize}
\end{proof}

We can use Theorem~\ref{thrm.r=wr+uc} to derive simple global well-behavedness conditions on spaces, as follows: 
\begin{corr}
Suppose $(\ns P,\opens)$ is a semitopology.
Then:
\begin{enumerate*}
\item
If the $\intertwinedwith$ relation is transitive (i.e. if every point is unconflicted) then a point is regular if and only if it is weakly regular.
\item
If every point is weakly regular (i.e. if $p\in\community(p)$ always) then a point is regular if and only if it is unconflicted.
\end{enumerate*} 
\end{corr}
\begin{proof}
Immediate from Theorem~\ref{thrm.r=wr+uc}. 
\end{proof}

\jamiesubsection{The boundary of $\intertwined{p}$}
\label{subsect.boundary.intertwined}

In this short Subsection we ask what points on the topological boundary of $\intertwined{p}$ can look like:
\begin{nttn}
\label{nttn.boundary}
Suppose $(\ns P,\opens)$ is a semitopology and $P\subseteq\ns P$.
\begin{enumerate*}
\item
As standard, we define 
$$
\f{boundary}(P) = P\setminus\interior(P)
$$ 
and we call this the \deffont{boundary of $P$}.
\item
In the case that $P=\intertwined{p}$ for $p\in\ns P$ then 
$$
\f{boundary}(\intertwined{p})=\intertwined{p}\setminus\interior(\intertwined{p})=\intertwined{p}\setminus\community(p).
$$
\end{enumerate*}
\end{nttn}

Points in the boundary of $\intertwined{p}$ are \emph{not} regular points:
\begin{prop}
\label{prop.boundary.points.not.regular}
\label{prop.char.boundary}
Suppose $(\ns P,\opens)$ is a semitopology and $p,q\in\ns P$ and $q\in\intertwined{p}$.
Then:
\begin{enumerate*}
\item\label{item.char.boundary.1}
If $q$ is regular then $q\in\community(p)=\interior(\intertwined{p})$.
\item\label{item.char.boundary.2}
If $q$ is regular then $q\notin\boundary(\intertwined{p})$.
\item\label{item.char.boundary.3}
If $q\in\boundary(\intertwined{p})$ then $q$ is either conflicted or not weakly regular (or both).
\end{enumerate*}
\end{prop}
\begin{proof}
We consider each part in turn:
\begin{enumerate}
\item
Suppose $q$ is regular.
By Theorem~\ref{thrm.r=wr+uc} $q$ is unconflicted, so that by Lemma~\ref{lemm.unconflicted.char}(\ref{item.p'.in.unconflicted.p}) $\intertwined{q}\subseteq\intertwined{p}$; and also $q$ is weakly regular, so that $q\in\community(q)\in\opens$ and $\community(q)\subseteq\intertwined{q}\subseteq\intertwined{p}$.
Thus $\community(q)$ is an open neighbourhood of $q$ that is contained in $\intertwined{p}$ and thus $q\in\interior(\intertwined{p})$ as required.
\item
This just repeats part~\ref{item.char.boundary.2} of this result, recalling from Notation~\ref{nttn.boundary} that $q\in\boundary(\intertwined{p})$ if and only if $q\notin\interior(\intertwined{p})$.
\item
This is just the contrapositive of part~\ref{item.char.boundary.2}, combined with Theorem~\ref{thrm.r=wr+uc}.
\qedhere\end{enumerate}
\end{proof}

\begin{figure}
\vspace{-1em}
\centering
\includegraphics[width=0.32\columnwidth,trim={50 20 50 20},clip]{diagrams/counterexample-1.pdf}
\includegraphics[width=0.32\columnwidth,trim={50 20 50 20},clip]{diagrams/012a.pdf}
\includegraphics[width=0.30\columnwidth,trim={50 20 50 20},clip]{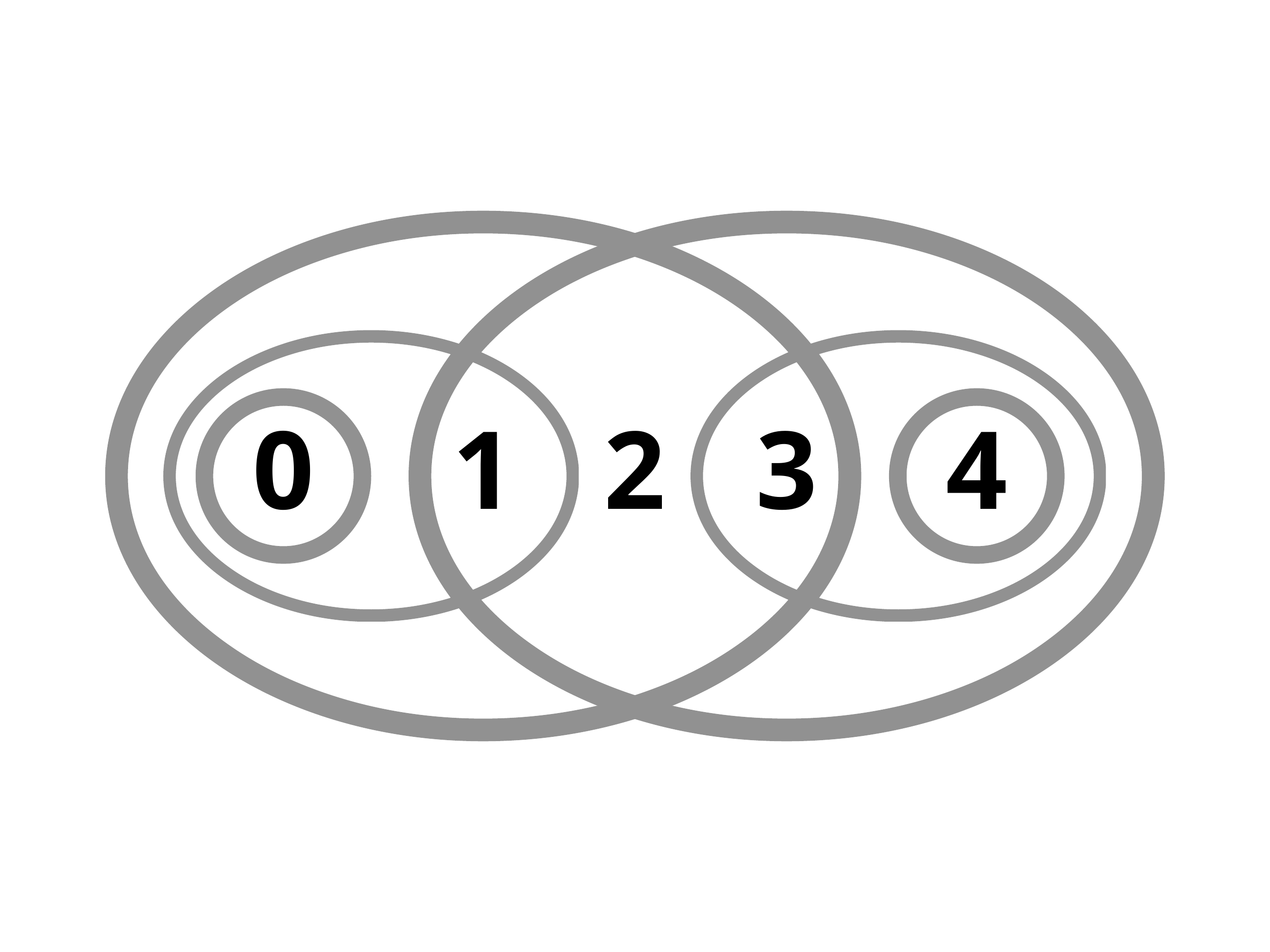}
\caption{Examples of boundary points (Example~\ref{xmpl.boundary.examples}).}
\label{fig.boundaries}
\end{figure}

\begin{xmpl}
\label{xmpl.boundary.examples}
Proposition~\ref{prop.char.boundary}(\ref{item.char.boundary.3}) tells us that points on the topological boundary of $\intertwined{p}$ are either conflicted, or not weakly regular, or perhaps both.
It remains to show that all options are possible.
It suffices to provide examples: 
\begin{enumerate*}
\item\label{item.boundary.examples.1}
In Figure~\ref{fig.boundaries} (left-hand diagram) the point $\ast$ is on the boundary of $\intertwined{1}=\{\ast,1\}$.
It is unconflicted (being intertwined just with itself and $1$), and not weakly regular (since $\ast\notin\community(\ast)=\{1\}$). 
\item\label{item.boundary.examples.2}
In Figure~\ref{fig.boundaries} (middle diagram) the point $1$ is on the boundary of $\intertwined{0}=\{0,1\}$.
It is conflicted (since $0\intertwinedwith 1\intertwinedwith 2$ yet $0\notintertwinedwith 2$) and it is weakly regular (since $1\in\community(1)=\{0,1,2\}$).\footnote{This semitopology is also in Figure~\ref{fig.012}.  We reproduce it here for the reader's convenience so that the examples are side-by-side.  
}
\item\label{item.boundary.examples.3}
In Figure~\ref{fig.boundaries} (right-hand diagram) the point $2$ is conflicted (since $1\intertwinedwith 2\intertwinedwith 3$ yet $1\notintertwinedwith 3$) and it is not weakly regular, or even quasiregular (since $\community(2)=\interior(\{1,2,3\})=\varnothing$).
\end{enumerate*} 
\end{xmpl}

We consider the special case of \emph{regular} spaces (we will pick this thread up again in Subsection~\ref{subsect.boundaries.of.closed.sets} after we have built more machinery):
\begin{corr}
\label{corr.bgp}
Suppose $(\ns P,\opens)$ is a semitopology and $p\in\ns P$. 
Then:
\begin{enumerate*}
\item\label{item.bgp.1}
If the set $\intertwined{p}$ is regular, then $\boundary(\intertwined{p})=\varnothing$ and $\intertwined{p}$ is clopen (closed and open) and transitive.
\item\label{item.bgp.2}
If $\ns P$ is a regular space (so every point in it is regular) then $\ns P$ partitions into clopen transitive components given by $\{\intertwined{p} \mid p\in\ns P\}$.
\end{enumerate*}
\end{corr}
\begin{proof}
\leavevmode
\begin{enumerate}
\item
By Proposition~\ref{prop.char.boundary} $\intertwined{p}=\interior(\intertwined{p})$, so by Lemma~\ref{lemm.interior.open} $\intertwined{p}$ is open.
By Proposition~\ref{prop.intertwined.as.closure}(\ref{intertwined.p.closed}) $\intertwined{p}$ is closed.
By Definition~\ref{defn.tn}(\ref{item.regular.point}) $p\in\community(p)=\interior(\intertwined{p})\in\topens$.
It follows that $\intertwined{p}$ is (topen and therefore) transitive.
\item
By part~\ref{item.bgp.1} of this result each $\intertwined{p}$ is a clopen transitive set.
Using Theorem~\ref{thrm.r=wr+uc} every point is unconflicted and it follows that if $\intertwined{p}\between\intertwined{p'}$ then $\intertwined{p}=\intertwined{p'}$. 
\qedhere\end{enumerate}
\end{proof}

\jamiesubsection{The intertwined preorder}

\jamiesubsubsection{Definition and properties}

\begin{rmrk}
Recall the \emph{specialisation preorder} on points from topology, defined by 
$$
p\leq p'
\quad\text{when}\quad
\closure{p}\subseteq\closure{p'}.
$$
In words: we order points $p$ by subset inclusion on their closure $\closure{p}$.

This can also be defined on semitopologies of course, but we will also find a similar preorder interesting, which is defined using $\intertwined{p}$ instead of $\closure{p}$ (Definition~\ref{defn.intertwined.preorder}).
Recall that:
\begin{itemize*}
\item
$\closure{p}$ is a closed set and is equal to the intersection of all the closed sets containing $p$, and 
\item
$\intertwined{p}$ is also a closed set (Proposition~\ref{prop.intertwined.as.closure}(\ref{intertwined.p.closed}))
and it is the intersection of all the closed neighbourhoods of $p$ (closed sets with an interior that contains $p$; see Definition~\ref{defn.cn} and Proposition~\ref{prop.intertwined.as.closure}(\ref{intertwined.as.closure.closed})).
\end{itemize*}
\end{rmrk}

\begin{defn}
\label{defn.intertwined.preorder}
Suppose $(\ns P,\opens)$ is a semitopology.
\begin{enumerate}
\item
Define the \deffont[intertwined preorder $p\leqk p'$]{intertwined preorder}\index{$p\leqk p'$ (intertwined preorder on points)} on points $p,p'\in\ns P$ by:
$$
p\leqk p'
\quad\text{when}\quad
\intertwined{p}\subseteq\intertwined{p'}.
$$
As standard, we may write $p'\geqk p$ when $p\leqk p'$ (pronounced `$p'$ is intertwined-less / intertwined-greater than $p$').

Calling this the `intertwined preorder' does not refer to the ordering being intertwined in any sense; it just means that we order on $\intertwined{p}$ (which is read `intertwined-$p$').
\item\label{item.intertwined-bounded}
Call $(\ns P,\opens)$ an \deffont{$\intertwinedwith$-complete semitopology}\index{intertwined-complete semitopology} (read `\deffont{intertwined-complete}') when 
for every subset $P\subseteq\ns P$ that is totally ordered by $\leqk$, 
there exists some $p\in\ns P$ such that $\intertwined{p}\subseteq \bigcap_i\{\intertwined{p}\mid p\in P\}$.
\end{enumerate}
\end{defn}

\begin{rmrk}
\label{rmrk.intertwinedwith-bounded.natural}
Being $\intertwinedwith$-complete (Definition~\ref{defn.intertwined.preorder}(\ref{item.intertwined-bounded})) is a plausible well-behavedness condition, because 
finite semitopologies are $\intertwinedwith$-complete, since a descending chain of subsets of a finite set is terminating.
Real systems are finite (though participants in the system may not be able to access all of them, so they may look infinite `from the inside'), so assuming that a semitopology is $\intertwinedwith$-complete is a reasonable abstraction of actual finiteness.
\end{rmrk}

\begin{rmrk}
There is also the \deffont[community preorder $p\leq_K p'$]{community preorder}\index{$p\leq_K p'$ (community preorder on points)} defined such that $p\leq_K p'$ when $\community(p)\subseteq\community(p')$, which is related to $p\leq p'$ via the fact that by definition $\community(p)=\interior(\intertwined{p})$, so that $\leq_K$ is a coarser relation (meaning: it relates more points).
There is an argument that this would sit more nicely with the condition $q\in\community(p)$ in Lemma~\ref{lemm.weakly.regular.community}, but ordering on $\community(p)$ would relate all points with empty community, e.g. all of the points in Figure~\ref{fig.square.diagram}, and would slightly obfuscate the parallel with the specialisation preorder. 
This strikes us as unintuitive, so we prefer to preorder on $\intertwined{p}$. 
\end{rmrk}

\begin{lemm}
\label{lemm.weakly.regular.community}
Suppose $(\ns P,\opens)$ is a semitopology and $p,q\in\ns P$. 
Then:
\begin{enumerate*}
\item\label{item.weakly.regular.community.1}
If $q\in\community(p)$ then $q\leqk p$ (meaning that $\intertwined{q}\subseteq\intertwined{p}$).
\item\label{item.weakly.regular.community.2}
If $q\in\community(p)$ then $\community(q)\subseteq \community(p)$.
\end{enumerate*}
\end{lemm}
\begin{proof}
We consider each part in turn:
\begin{enumerate}
\item
Suppose $q\in\community(p)$ and recall from Lemma~\ref{lemm.two.intertwined}(\ref{item.two.intertwined.1})
that $\community(p)\in\opens$, which means that $\closure{\community(p)}$ is a closed neighbourhood of $q$.
We use Proposition~\ref{prop.intertwined.as.closure}(\ref{item.intertwined.as.intersection.of.closures}) and Lemma~\ref{lemm.closure.community.subset}:\footnote{If the reader prefers a proof by concrete calculations, it runs as follows:
Suppose $p'\in\community(p)$, so that in particular $p'\intertwinedwith p$.
We wish to prove that $\intertwined{p'}\subseteq\intertwined{p}$.
So consider $p''\intertwinedwith p'$; we will show that $p''\intertwinedwith p$, i.e. that every pair of open neighbourhoods of $p''$ and $p$ must intersect.
Consider a pair of open neighbourhoods $p''\in O''\in\opens$ and $p\in O\in\opens$.
We note that $O''\between \community(p)$, because $p'\in\community(p)\in\opens$ and $p''\intertwinedwith p'$.
Choose $q\in\community(p)\cap O''$. 
Now $q\intertwinedwith p$ and $q\in O''$ and $p\in O$, and we conclude that $O''\between O$ as required.
}
$$
\intertwined{q} 
\stackrel{P\ref{prop.intertwined.as.closure}(\ref{item.intertwined.as.intersection.of.closures})}{\subseteq} 
\closure{\community(p)} 
\stackrel{L\ref{lemm.closure.community.subset}}{\subseteq} 
\intertwined{p}.
$$
\item
Suppose $q\in\community(p)$.
By part~\ref{item.weakly.regular.community.1} of this result and Definition~\ref{defn.intertwined.preorder} $\intertwined{q}\subseteq\intertwined{p}$.
It is a fact that then $\interior(\intertwined{q})\subseteq\interior(\intertwined{p})$.
By Definition~\ref{defn.tn}(\ref{item.tn}) therefore $\community(q)\subseteq\community(p)$ as required.
\qedhere\end{enumerate}
\end{proof}

In the rest of this Subsection we develop corollaries of Lemma~\ref{lemm.weakly.regular.community} (and compare this with Proposition~\ref{prop.community.partition}):
\begin{corr}
\label{corr.community.intersects.community}
Suppose $(\ns P,\opens)$ is a semitopology and $q,q'\in\ns P$.
Then:
\begin{enumerate*}
\item\label{item.community.intersects.community.1}
If $\community(q)\between\community(q')$ then $q\intertwinedwith q'$.
\item\label{item.community.intersects.community.2}
If $q$ and $q'$ are weakly regular (so that $q\in\community(q)$ and $q'\in\community(q')$) then
$$
q\intertwinedwith q'
\quad\text{if and only if}\quad
\community(q)\between\community(q').
$$
\end{enumerate*}
\end{corr}
\begin{proof} 
We consider each part in turn:
\begin{enumerate}
\item
Suppose $r\in\community(q)\cap\community(q')$.
Then $\intertwined{r}\subseteq\intertwined{q}\cap\intertwined{q'}$ using Lemma~\ref{lemm.weakly.regular.community}(\ref{item.weakly.regular.community.1}).
But $q\in\intertwined{r}$, so $q\in\intertwined{q'}$, and thus $q\intertwinedwith q'$.
\item
If $q$ and $q'$ are weakly regular and $q\intertwinedwith q'$ then $\community(q)\between\community(q')$ follows from Definition~\ref{defn.intertwined.points}(\ref{item.p.intertwinedwith.p'}).
The result follows from this and from part~\ref{item.community.intersects.community.1} of this result.
\qedhere\end{enumerate}
\end{proof}

Theorem~\ref{thrm.K-regular} is somewhat reminiscent of the \emph{hairy ball theorem}:\footnote{This famous result states that every tangent vector field on a sphere of even dimension --- this being the surface of a ball of odd dimension --- must vanish at at least one point.  Intuitively, if we consider a `hairy ball' in three-dimensional space and we try to comb its hairs so they all lie smoothly flat (with no discontinuities in direction), then at least one of the hairs is pointing straight up (i.e. its projection onto the ball is zero).  A nice combinatorial proof is in \cite{doi:10.1080/00029890.2004.11920120}.} 
\begin{thrm}
\label{thrm.K-regular}
Suppose $(\ns P,\opens)$ is an $\intertwinedwith$-complete quasiregular semitopology.\footnote{Definition~\ref{defn.tn}(\ref{item.quasiregular.point}): a semitopology that is $\intertwinedwith$-complete and whose every point has a nonempty community.}
Then:
\begin{enumerate*}
\item\label{item.K-regular.1}
For every $p\in\ns P$ there exists some regular $q\in\community(p)$.
\item\label{item.K-regular.2}
$\ns P$ contains a regular point.
\end{enumerate*}
\end{thrm}
\begin{proof}
We consider each part in turn:
\begin{enumerate}
\item
Consider the subset $\{p'\in\ns P \mid p'\leqk p\}\subseteq\ns P$ ordered by $\leqk $.
Using Zorn's lemma (on $\geqk$), this contains a $\leqk$-minimal element $q'$.
By assumption of quasiregularity $\community(q')\neq\varnothing$, so choose $q\in\community(q')$.
By Lemma~\ref{lemm.weakly.regular.community}(\ref{item.weakly.regular.community.1}) $\intertwined{q}\subseteq\intertwined{q'}$ and by $\leqk$-minimality $\intertwined{q}=\intertwined{q'}$ and it follows that $q\in\community(q)$.
Thus $q$ is weakly regular.
Applying similar reasoning to $p'\in\community(q)$ we deduce that $\intertwined{p'}=\intertwined{q}$, and thus $\community(p')=\community(q)$, for every $p'\in\community(q)$, and so by Corollary~\ref{corr.corr.pKp} $q$ is regular.
\item
Choose any $p\in\ns P$, and use part~\ref{item.K-regular.2} of this result.
\qedhere\end{enumerate}
\end{proof}

\begin{rmrk}
We care about the existence of regular points as these are the ones that are well-behaved with respect to our semitopological model. 
A semitopology with a regular point is one that --- in some idealised mathematical sense --- is capable of some collaboration somewhere to take some action.

So Theorem~\ref{thrm.K-regular} can be read as a guarantee that, provided the semitopology is $\intertwinedwith$-complete and quasiregular, there exists somebody, somewhere, who can make sense of their local network and progress to act.
This a mathematical guarantee and not an engineering one, much as is the hairy ball theorem of which the result reminds us. 
\end{rmrk} 

\jamiesubsubsection{Application to quasiregular conflicted spaces}

In Proposition~\ref{prop.unconflicted.irregular}(\ref{item.unconflicted.irregular.3}) we saw an example of an unconflicted irregular space (illustrated in Figure~\ref{fig.square.diagram}): this is a space in which every point is unconflicted but not weakly regular.
In this subsection we consider a dual case, of a conflicted quasiregular space: a space in which every point is conflicted yet quasiregular.

One question is: does such a creature even exist?
The answer is: 
\begin{itemize*}
\item
no, in the finite case (Corollary~\ref{corr.no.finite.wr.c}); and 
\item
yes, in the infinite case (Proposition~\ref{prop.conflicted.weakly.regular}).
\end{itemize*}

\begin{prop}
\label{prop.weakly.regular.to.regular}
Suppose $(\ns P,\opens)$ is a finite quasiregular semitopology (so $\ns P$ is finite and every $p\in\ns P$ is quasiregular) --- in particular this holds if the semitopology is weakly regular.
Then:
\begin{enumerate*}
\item
For every $p\in\ns P$ there exist some regular $q\in\community(p)$. 
\item
$\ns P$ contains a regular point.
\end{enumerate*}
In words we can say: every finite quasiregular semitopology contains a regular point.
\end{prop}
\begin{proof}
From Theorem~\ref{thrm.K-regular}, since `is finite' implies `is $\intertwinedwith$-complete'.%
\footnote{The proof of Theorem~\ref{thrm.K-regular} uses Zorn's lemma.  A longer, direct proof of Proposition~\ref{prop.weakly.regular.to.regular} is also possible, by explicit induction on size of sets.}
\end{proof}

\begin{corr}
\label{corr.no.finite.wr.c}
There exists no finite quasiregular conflicted semitopology (i.e. a semitopology with finitely many points, each of which is quasiregular but conflicted).
\end{corr}
\begin{proof}
Suppose $(\ns P,\opens)$ is finite and quasiregular.
By Proposition~\ref{prop.weakly.regular.to.regular} it contains a regular $q\in\ns P$ and by Proposition~\ref{prop.unconflicted.irregular}(\ref{item.reg.implies.unconflicted}) $q$ is unconflicted. 
\end{proof}

\begin{figure}
\centering
\includegraphics[width=0.6\columnwidth]{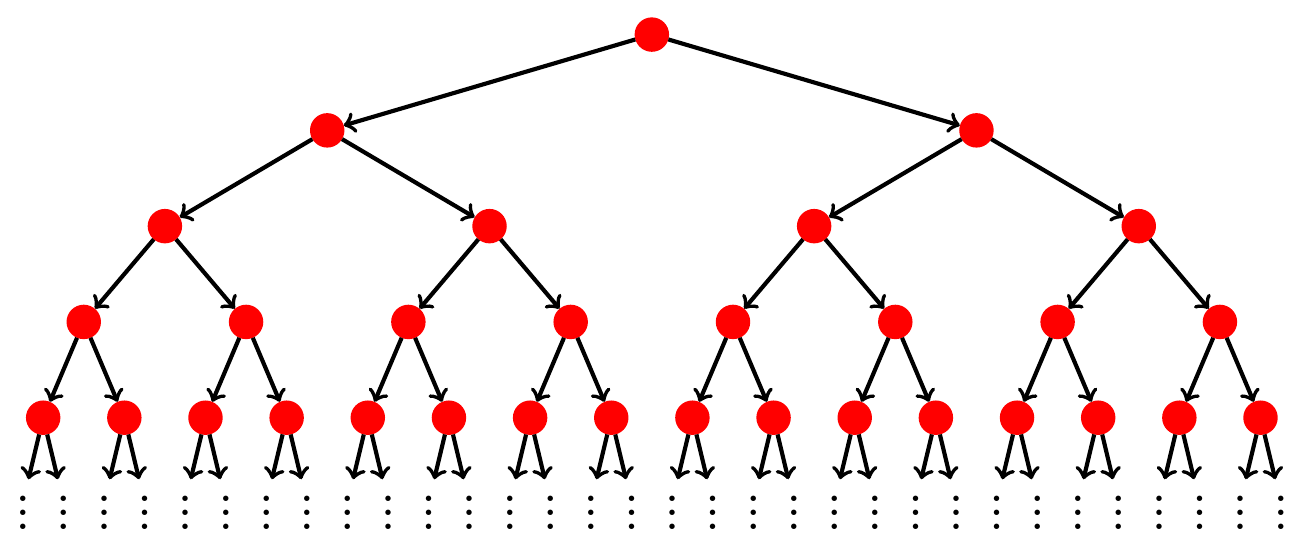}
\caption{A weakly regular, conflicted space (Proposition~\ref{prop.conflicted.weakly.regular}); the opens are the down-closed sets}
\label{fig.weakly-regular.conflicted}
\end{figure}

Corollary~\ref{corr.no.finite.wr.c} applies to finite semitopologies because these are necessarily $\intertwinedwith$-complete.
The infinite case is different, as we shall now observe:
\begin{prop}
\label{prop.conflicted.weakly.regular}
There exists an infinite quasiregular --- indeed it is also weakly regular --- conflicted semitopology $(\ns P,\opens)$.

In more detail:
\begin{itemize*}
\item
every $p\in\ns P$ is weakly regular (so $p\in\community(p)\in\opens$; see Definition~\ref{defn.tn}(\ref{item.weakly.regular.point})) yet 
\item
every $p\in\ns P$ is conflicted (so $\intertwinedwith$ is not transitive at $p$; Definition~\ref{defn.conflicted}(\ref{item.conflicted.point})).
\end{itemize*}
Furthermore: $\ns P$ is a topology
and contains no topen sets.
\end{prop}
\begin{proof}
Take $\ns P=[01]^*$ to be the set of words (possibly empty finite lists) from $0$ and $1$.
For $w,w'\in\ns P$ write $w\leq w'$ when $w$ is an initial segment of $w'$ and define 
$$
w_\geq = \{w' \mid w\leq w'\}
\quad\text{and}\quad
w_\leq = \{w' \mid w'\leq w\}.
$$
Let open sets be generated as (possibly empty) unions of the $w_\geq$.
This space is illustrated in Figure~\ref{fig.weakly-regular.conflicted}; open sets are down-closed subsets. 

The reader can check that $\neg(w0\intertwinedwith w1)$, because $w0_\geq\cap w1_\geq=\varnothing$, and that $w\intertwinedwith w'$ when $w\leq w'$ or $w'\leq w$.
It follows from the above that 
$$
\intertwined{w}=w_\geq\cup w_\leq
\quad\text{and}\quad 
\community(w)=\interior(\intertwined{w})=w_\geq,
$$
and since $w\in w_\geq$ every $w$ is weakly regular. 
Yet every $w$ is also conflicted, because $w0\intertwinedwith w \intertwinedwith w1$ yet $\neg(w0\intertwinedwith w1)$. 

This example is a topology, because an intersection of down-closed sets is still down-closed.
It escapes the constraints of Theorem~\ref{thrm.K-regular} by not being $\intertwinedwith$-complete.
It contains no topen sets because if it did contain some topen $\atopen$ then by Theorem~\ref{thrm.max.cc.char}(\ref{char.p.regular}\&\ref{char.some.topen}) there would exist a regular $p\in\atopen$ in $\ns P$.
\end{proof}

\jamiesubsubsection{(Un)conflicted points and boundaries of closed sets}
\label{subsect.boundaries.of.closed.sets}

Recall from Definition~\ref{defn.cn} that a closed neighbourhood is a closed set with a nonempty interior, and recall that $\intertwined{p}$ --- the set of points intertwined with $p$ from Definition~\ref{defn.intertwined.points} --- is characterised using closed neighbourhoods in Proposition~\ref{prop.closure.intertwined}, as the intersection of all closed neighbourhoods that have $p$ in their interior.

This leads to the question of whether the theory of $\intertwined{p}$ might \emph{be} a theory of closed neighbourhoods.
The answer seems to be no: $\intertwined{p}$ has its own distinct character, as the results and counterexamples below will briefly illustrate. 

For instance: in view of Proposition~\ref{prop.closure.intertwined} characterising $\intertwined{p}$ as an intersection of closed neighbourhoods of $p$, might it be the case that for $C$ a closed neighbourhood, $C=\bigcup\{\intertwined{p} \mid p\in\interior(C)\}$.
In words: is a closed neighbourhood $C$ the union of the points intertwined with its interior? 
This turns out to be only half true:
\begin{lemm}
\label{lemm.ab12}
Suppose $(\ns P,\opens)$ is a semitopology and $C\in\closed$ is a closed neighbourhood.
Then: 
\begin{enumerate*}
\item\label{item.ab12.1}
$\bigcup\{\intertwined{p} \mid p\in\interior(C)\}\subseteq C$.
\item\label{item.ab12.2}
This subset inclusion may be strict: it is possible for $p\in\ns P$ to be on the boundary of a closed neighbourhood $C$, but not intertwined with any point in that neighbourhood's interior.
This is true even if $\ns P$ is a regular space (meaning that every $p\in\ns P$ is regular).
\end{enumerate*}
\end{lemm}
\begin{proof}
We consider each part in turn:
\begin{enumerate}
\item
If $p\in\interior(C)$ then $\intertwined{p}\subseteq C$ by Proposition~\ref{prop.intertwined.as.closure}(\ref{intertwined.as.closure.closed}).
\item
We provide a counterexample, as illustrated in Figure~\ref{fig.Ast12} (left-hand diagram): 
\begin{itemize*}
\item
$\ns P=\{\ast, 1, 2\}$.
\item
Open sets are generated by $\{1\}$, $\{2\}$, and $\{\ast,2\}$.
\item
We set $p=\ast$ and $C=\{1,\ast\}$.
\end{itemize*}
Then the reader can check that $\interior(C)=\{2\}$ $\intertwined{\ast}=\{\ast,2\}$ and $\ast\notintertwinedwith 2$ and every point in $\ns P$ is regular.
\qedhere\end{enumerate}
\end{proof}

\begin{figure}
\vspace{-2em}
\centering
\subcaptionbox{Regular boundary point of closed neighbourhood that is not intertwined with its interior (Lemma~\ref{lemm.ab12}(\ref{item.ab12.2}))}{\includegraphics[width=0.4\columnwidth,trim={50 60 50 50},clip]{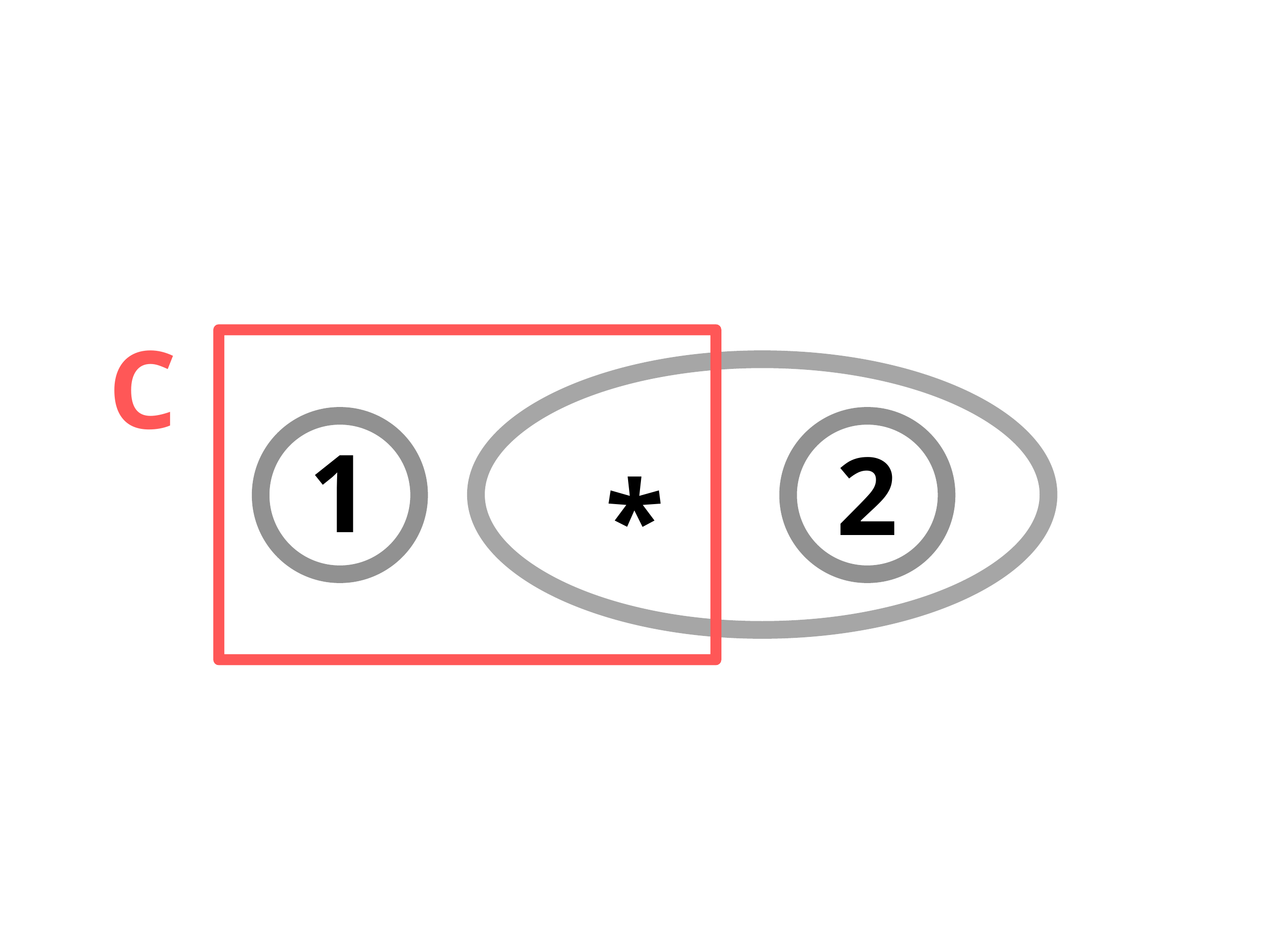}}
\qquad
\subcaptionbox{Regular point in kissing set of closed neighbourhoods, not intertwined with interiors (Corollary~\ref{corr.ab123}(\ref{item.ab123.2}))}{\includegraphics[width=0.4\columnwidth,trim={50 20 50 50},clip]{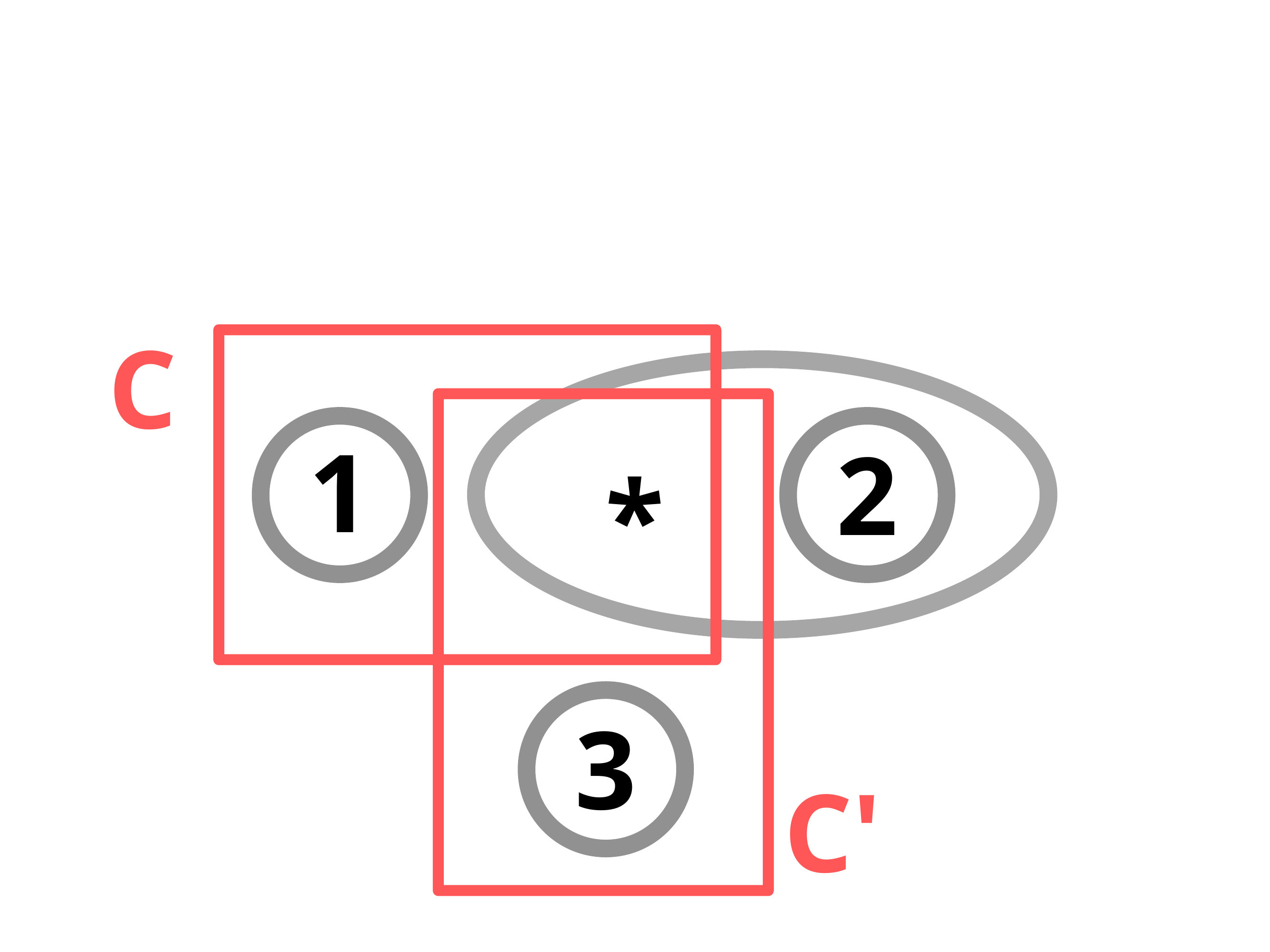}}
\caption{Two counterexamples}
\label{fig.Ast12}
\end{figure}

\begin{defn}
Suppose $(\ns P,\opens)$ is a semitopology and $P,P'\subseteq\ns P$.
Then
define 
$$
\f{kiss}(P,P')=\boundary(P)\cap \boundary(P')
$$ 
and call this the \deffont{kissing set of $P$ and $P'$}.
\end{defn}

\begin{lemm}
\label{lemm.kissing.conflict}
Suppose $(\ns P,\opens)$ is a semitopology.
Then the following are equivalent:
\begin{itemize*}
\item
$p$ is conflicted.
\item
There exist $q,q'\in\ns P$ such that $q\notintertwinedwith q'$ and $p\in\kiss(\intertwined{q},\intertwined{q'})$.
\item
There exist $q,q'\in\ns P$ such that $q\notintertwinedwith q'$ and $p\in\intertwined{q}\cap\intertwined{q'}$.
\end{itemize*}
\end{lemm}
\begin{proof}
We prove a cycle of implications:
\begin{itemize}
\item
\emph{Suppose $p$ is conflicted.}\quad

Then there exist $q,q'\in\ns P$ such that $q\intertwinedwith p\intertwinedwith q'$ yet $q\notintertwinedwith q'$.
Rephrasing this, we obtain that $p\in\intertwined{q}\cap\intertwined{q'}$.

We need to check that $p\notin\community(q)$ and $p\notin\community(q')$.
We prove $p\notin\community(q)$ by contradiction ($p\notin\community(q')$ follows by identical reasoning).
Suppose $p\in\community(q)$.
Then by Lemma~\ref{lemm.weakly.regular.community}(\ref{item.weakly.regular.community.1}) $\intertwined{p}\subseteq\intertwined{q}$.
But $q'\in\intertwined{p}$, so $q'\in\intertwined{q}$, so $q'\intertwinedwith q$, contradicting our assumption.
\item
\emph{Suppose $q\notintertwinedwith q'$ and $p\in\boundary(\intertwined{q})\cap\boundary(\intertwined{q'})$.}

Then certainly $p\in\intertwined{q}\cap\intertwined{q'}$.
\item
\emph{Suppose $q\notintertwinedwith q'$ and $p\in\intertwined{q}\cap\intertwined{q'}$.}

Then $q\intertwinedwith p\intertwinedwith q'$ and $q\notintertwinedwith q'$, which is precisely what it means to be conflicted.
\qedhere\end{itemize}
\end{proof}

We can look at Definition~\ref{defn.conflicted} and Lemma~\ref{lemm.kissing.conflict} and conjecture that a point $p$ is conflicted if and only if it is in the kissing set of a pair of distinct closed sets.
Again, this is half true:
\begin{corr}
\label{corr.ab123}
Suppose $(\ns P,\opens)$ is a semitopology and $p\in\ns P$.
Then:
\begin{enumerate*}
\item\label{item.ab123.1}
If $p$ is conflicted then there exist a pair of closed sets such that $p\in\kiss(C,C')$.
\item\label{item.ab123.2}
The reverse implication need not hold: it is possible for $p$ to be in the kissing set of a pair of closed sets $C$ and $C'$, yet $p$ is unconflicted.
This is even possible if the space is regular (meaning that every point in the space is regular, including $p$) and $C$ and $C'$ are closed neighbourhoods.
\end{enumerate*}
\end{corr}
\begin{proof}
We consider each part in turn:
\begin{enumerate}
\item
If $p$ is conflicted then we use Lemma~\ref{lemm.kissing.conflict} and Proposition~\ref{prop.intertwined.as.closure}(\ref{intertwined.p.closed}).
\item
We provide a counterexample, as illustrated in Figure~\ref{fig.Ast12} (right-hand diagram): 
\begin{itemize*}
\item
$\ns P=\{\ast, 1, 2, 3\}$.
\item
Open sets are generated by $\{1\}$, $\{2\}$, $\{3\}$, and $\{\ast, 2\}$. 
\item
We set $p=\ast$ and $C=\{\ast,1\}$ and $C'=\{\ast, 3\}$.
\end{itemize*}
Note that $\ast$ is regular (being intertwined with itself and $2$), and $C$ and $C'$ are closed neighbourhoods that kiss at $\ast$, and $1$, $2$, and $3$ are also regular. 
\qedhere\end{enumerate}
\end{proof}

\jamiesubsection{Regular = quasiregular + hypertransitive}

\begin{rmrk}
In Theorem~\ref{thrm.r=wr+uc} we characterised regularity in terms of weak regularity and being unconflicted.
Regularity and weak regularity are two of the regularity properties considered in Definition~\ref{defn.tn}, but there is also a third: \emph{quasiregularity}.
This raises the question whether there might be some other property $X$ such that regular = quasiregular + $X$?\footnote{By Lemma~\ref{lemm.wr.r}(\ref{item.wr.implies.qr}) being weakly regular is a stronger condition than being quasiregular, thus we would expect $X$ to be stronger than being unconflicted.  And indeed this will be so: see Lemma~\ref{lemm.regular.sc}(\ref{item.sc.implies.uc}).}

Yes there is, and we develop it in this Subsection, culminating with Theorem~\ref{thrm.regular=qr+sc}.
\end{rmrk}

\jamiesubsubsection{Hypertransitivity}

\begin{nttn}
\label{nttn.between.nbhd}
Suppose $(\ns P,\opens)$ is a semitopology and $O'\in\opens$ and $\mathcal O\subseteq\opens$.
\begin{enumerate*}
\item\label{item.between.nbhd.1}
Write $O'\between\mathcal O$, or equivalently $\mathcal O\between O'$, when $O'\between O$ for every $O\in\mathcal O$.
In symbols:
$$
O'\between\mathcal O
\quad\text{when}\quad
\Forall{O{\in}\mathcal O}O'\between O .
$$
\item\label{item.between.nbhd}
As a special case of part~\ref{item.between.nbhd.1} above taking $\mathcal O=\nbhd(p)$ (Definition~\ref{defn.nbhd.system}), if $p\in\ns P$ then write $O'\between\nbhd(p)$, or equivalently $\nbhd(p)\between O'$, when $O'\between O$ for every $O\in\opens$ such that $p\in O$. 
\end{enumerate*}
\end{nttn}

\begin{lemm}
\label{lemm.closure.using.nbhd.intersections}
Suppose $(\ns P,\opens)$ is a semitopology and $p\in\ns P$ and $O'\in\opens$.
Then 
$$
p\in\closure{O'}
\quad\text{if and only if}\quad 
O'\between\nbhd(p) .
$$
\end{lemm}
\begin{proof}
This just rephrases Definition~\ref{defn.closure}(\ref{item.closure}). 
\end{proof}

\begin{defn}
\label{defn.sc}
Suppose $(\ns P,\opens)$ is a semitopology.
Call $p\in\ns P$ a \deffont{hypertransitive point} when for every $O',O''\in\opens$, 
$$
O'\between\nbhd(p)\between O''
\quad\text{implies}\quad O'\between O''.
$$
Call $(\ns P,\opens)$ a \deffont{hypertransitive semitopology} when every $p\in\ns P$ is hypertransitive.
\end{defn}

\begin{lemm}
\label{lemm.sc.op.reg.op}
Suppose $(\ns P,\opens)$ is a semitopology and $p\in\ns P$.
Then the following are equivalent:
\begin{enumerate*}
\item\label{item.sc.op.reg.op.1}
$p$ is hypertransitive.
\item\label{item.sc.op.reg.op.2}
For every pair of open sets $O',O''\in\opens$, $p\in \closure{O'}\cap \closure{O''}$ implies $O'\between O''$.
\item\label{item.sc.op.reg.op.3}
For every pair of \emph{regular} open sets $O',O''\in\regularOpens$, $p\in \closure{O'}\cap \closure{O''}$ implies $O'\between O''$ (cf. Remark~\ref{rmrk.intertwined.with.regular.opens}).
\end{enumerate*}
\end{lemm}
\begin{proof}
For the equivalence of parts~\ref{item.sc.op.reg.op.1} and~\ref{item.sc.op.reg.op.2} we reason as follows:
\begin{itemize*}
\item
Suppose $p$ is hypertransitive and suppose $p\in\closure{O'}$ and $p\in\closure{O''}$.
By Lemma~\ref{lemm.closure.using.nbhd.intersections} it follows that $O'\between\nbhd(p)\between O''$.
By hypertransitivity, $O'\between O''$ as required.
\item
Suppose for every $O,O'\in\opens$, $p\in\closure{O}\cap\closure{O'}$ implies $O'\between O''$, and suppose $O'\between\nbhd(p)\between O''$.
By Lemma~\ref{lemm.closure.using.nbhd.intersections} $p\in\closure{O}\cap\closure{O'}$ and therefore $O'\between O''$.
\end{itemize*}
For the equivalence of parts~\ref{item.sc.op.reg.op.2} and~\ref{item.sc.op.reg.op.3} we reason as follows: 
\begin{itemize*}
\item
Part~\ref{item.sc.op.reg.op.2} implies part~\ref{item.sc.op.reg.op.3} follows since every open regular set is also an open set.
\item
To show part~\ref{item.sc.op.reg.op.3} implies part~\ref{item.sc.op.reg.op.2}, suppose for every pair of regular opens $O',O''\in\regularOpens$, $p\in \closure{O'}\cap \closure{O''}$ implies $O'\between O''$, and suppose $O',O''\in\opens$ are two open sets that are not necessarily regular, and suppose $p\in\closure{O'}\cap\closure{O''}$.
We must show that $O'\between O''$.

Write $P'=\interior(\closure{O'})$ and $P''=\interior(\closure{O''})$ and note by Lemmas~\ref{lemm.ic.ci.regular} and~\ref{lemm.closure.closed} that $P'$ and $P''$ are regular open sets and $\closure{P'}=\closure{O'}$ and $\closure{P''}=\closure{O''}$.
Then $\closure{P'}\between\closure{P''}$, so $P'\between P''$, and $O'\between O''$ follows from Lemma~\ref{lemm.clint.between}
\qedhere\end{itemize*}
\end{proof}

\jamiesubsubsection{The equivalence}

\begin{lemm}
\label{lemm.regular.sc}
Suppose $(\ns P,\opens)$ is a semitopology and $p\in\ns p$.
Then:
\begin{enumerate*}
\item\label{item.r.implies.sc}
If $p$ is regular then it is hypertransitive.
\item\label{item.sc.implies.uc}
If $p$ is hypertransitive then it is unconflicted.
\item
The reverse implication need not hold: it is possible for $p$ to be unconflicted but not hypertransitive.
\item
It is possible for $p$ to be hypertransitive (and unconflicted), but not quasiregular (and thus not weakly regular or regular).
\end{enumerate*}
\end{lemm}
\begin{proof}
We consider each part:
\begin{enumerate}
\item
Suppose $p$ is regular and $O,O'\in\opens$ and $O\between\nbhd(p)\between O'$.
By Definition~\ref{defn.tn}(\ref{item.regular.point}) (since $p$ is regular) $\community(p)$ is a topen (= open and transitive) neighbourhood of $p$.
Therefore by transitivity $O\between O'$ as required. 
\item
Suppose $p$ is hypertransitive and suppose $p',p''\in\ns P$ and $p'\intertwinedwith p\intertwinedwith p''$.
Now consider $p'\in O'\in\opens$ and $p''\in O''\in\opens$.
By our intertwinedness assumptions we have that $O'\between\nbhd(p)\between O''$.
But $p$ is hypertransitive, so $O'\between O''$ as required.
\item
It suffices to provide a counterexample.
Consider the bottom right semitopology in Figure~\ref{fig.012}, and take $p=\ast$ and $O'=\{1\}$ and $O''=\{0,2\}$.
Note that:
\begin{itemize*}
\item
$\ast$ is unconflicted, since it is intertwined only with itself and $1$.
\item
$O'$ and $O'$ intersect every open neighbourhood of $\ast$, but $O'\notbetween O''$, so $\ast$ is not strongly compatible.
\end{itemize*} 
\item
It suffices to provide an example.
Consider the semitopology illustrated in Figure~\ref{fig.012}, top-right diagram; so $\ns P=\{0,1,2\}$ and $\opens=\{\varnothing,\{0\},\{2\},\{1,2\},\{0,1\},\{0,1,2\}\}$.
The reader can check that $p=1$ is hypertransitive, but $\intertwined{1}=\{1\}$ and $\community(1)=\varnothing$ so $p$ is not quasiregular.
\qedhere\end{enumerate}
\end{proof}

(Yet) another characterisation of being quasiregular will be helpful:
\begin{lemm}
\label{lemm.quasiregular.iff.between}
Suppose $(\ns P,\opens)$ is a semitopology and $p\in\ns P$.
Then the following conditions are equivalent:
\begin{enumerate*}
\item\label{item.quasiregular.iff.between.1}
$p$ is quasiregular (meaning by Definition~\ref{defn.tn}(\ref{item.quasiregular.point}) that $\community(p)\neq\varnothing$).
\item\label{item.quasiregular.iff.between.2}
$\community(p)\between\nbhd(p)$ (meaning by Notation~\ref{nttn.between.nbhd}(\ref{item.between.nbhd}) that $\community(p)\between O$ for every $O\in\nbhd(p)$).
\item\label{item.quasiregular.iff.between.3}
$p\in\closure{\community(p)}$.
\end{enumerate*}
\end{lemm}
\begin{proof}
Equivalence of parts~\ref{item.quasiregular.iff.between.2} and~\ref{item.quasiregular.iff.between.3} is immediate from Lemma~\ref{lemm.closure.using.nbhd.intersections}.

For equivalence of parts~\ref{item.quasiregular.iff.between.1} and~\ref{item.quasiregular.iff.between.2}, we prove two implications:
\begin{itemize}
\item
Suppose $p$ is quasiregular, meaning by Definition~\ref{defn.tn}(\ref{item.quasiregular.point}) that $\community(p)\neq\varnothing$.
Pick some $p'\in\community(p)$ (it does not matter which).
It follows by construction in Definitions~\ref{defn.intertwined.points}(\ref{intertwined.defn}) and~\ref{defn.tn}(\ref{item.tn}) and Lemma~\ref{lemm.interior.open} that $p'\intertwinedwith p$, so that $p'\in\community(p)$. 
Using Definition~\ref{defn.intertwined.points}(\ref{item.p.intertwinedwith.p'}) it follows that $\community(p)\between O$ for every $O\in\nbhd(p)$, as required.
\item
Suppose $\community(p)\between\nbhd(p)$.
Then in particular $\community(p)\between\ns P$ (because $p\in\ns P\in\opens$), and by Notation~\ref{nttn.between}(\ref{item.between}) it follows that $\community(p)\neq\varnothing$.
\qedhere\end{itemize}
\end{proof}

Compare and contrast Theorem~\ref{thrm.regular=qr+sc} with Theorem~\ref{thrm.r=wr+uc}:
\begin{thrm}
\label{thrm.regular=qr+sc}
Suppose $(\ns P,\opens)$ is a semitopology and $p\in\ns P$.
Then the following are equivalent:
\begin{enumerate*}
\item
$p$ is regular.
\item
$p$ is quasiregular and hypertransitive.
\end{enumerate*}
\end{thrm}
\begin{proof}
We consider two implications:
\begin{itemize}
\item
\emph{Suppose $p$ is regular.}\quad

Then $p$ is quasiregular by Lemma~\ref{lemm.wr.r}(\ref{item.r.implies.wr}\&\ref{item.wr.implies.qr}), and hypertransitive by Lemma~\ref{lemm.regular.sc}(\ref{item.r.implies.sc}). 
\item
\emph{Suppose $p$ is quasiregular and hypertransitive.}\quad

By Lemma~\ref{lemm.regular.sc}(\ref{item.sc.implies.uc}) $p$ is unconflicted.
If we can prove that $p$ is weakly regular (meaning by Definition~\ref{defn.tn}(\ref{item.weakly.regular.point}) that $p\in\community(p)$), then by Theorem~\ref{thrm.r=wr+uc} it would follow that $p$ is regular as required.
Thus, it would suffice to show that $p\in\community(p)$, thus that there is an open neighbourhood of points with which $p$ is intertwined.

Write $O''=\interior(\ns P\setminus\community(p))$.
We have two subcases to consider:
\begin{itemize*}
\item
\emph{Suppose $\nbhd(p)\between O''$.}\quad

By Lemma~\ref{lemm.quasiregular.iff.between} (since $p$ is quasiregular) we have that $\community(p)\between\nbhd(p)$.
Thus $\community(p)\between\nbhd(p)\between O''$, and by hypertransitivity of $p$ it follows that $\community(p)\between O''$.
But this contradicts the construction of $O''$ as being a subset of $\ns P\setminus\community(p)$, so this case is impossible.
\item
\emph{Suppose $\nbhd(p)\notbetween O''$.}\quad
Then there exists some $O\in\nbhd(p)$ such that $O\notbetween O''$, and it follows that $O\subseteq\community(p)$ so that $p\in\community(p)$ as required.
\end{itemize*}
Thus $p$ is weakly regular, as required.
\qedhere\end{itemize}
\end{proof}

\begin{rmrk}
\label{rmrk.two.char.r}
So we have obtained two nice characterisations of regularity of points from Definition~\ref{defn.tn}(\ref{item.regular.point}):
\begin{enumerate*}
\item
Regular = weakly regular + unconflicted, by Theorem~\ref{thrm.r=wr+uc}. 
\item
Regular = quasiregular + hypertransitive, by Theorem~\ref{thrm.regular=qr+sc}. 
\end{enumerate*}
\end{rmrk}

\jamiesection{Conclusions}
\label{sect.conclusions}

We start by noticing that a notion of `actionable coalition' as discussed in the Introduction, leads to the topology-like structure which we call \emph{semitopologies}.

We simplified and purified our motivating examples --- having to do with understanding agreement and consensus in distributed systems --- to two precise mathematical questions: 
\begin{enumerate*}
\item
understand antiseparation properties, and 
\item
understand the implications of these for value assignments.\footnote{A value assignment is just a not-necessarily-continuous map from a semitopology to a discrete space.}
\end{enumerate*}
We have seen that the implications of these ideas are rich and varied.
Point-set semitopologies have an interesting theory which obviously closely resembles point-set topology, but is not identical to it.
In particular, dropping the condition that intersections of open sets must be open permits a wealth of new structure, which our taxonomy of antiseparation properties and its applications to value assignments explores.

\jamiesubsection{Topology vs. semitopology}
\label{subsect.vs}

We briefly compare and contrast topology and semitopology: 
\begin{enumerate}
\item
\emph{Topology:}\ 
Separation axioms are prominent in the topological literature; I could find no corresponding taxonomy of anti-separation properties.\footnote{The Wikipedia page on separation axioms (\href{https://web.archive.org/web/20221103233631/https://en.wikipedia.org/wiki/Separation_axiom}{permalink}) includes an excellent overview with over a dozen separation axioms; no anti-separation axioms are proposed.  Important non-Hausdorff spaces do exist; e.g. the \emph{Zariski topology}~\cite[Subsection~1.1.1]{hulek:eleag}.} 

\emph{Semitopology:}\ 
Antiseparation, not separation, is our primary interest.
We consider a taxonomy of antiseparation properties, including: points being intertwined (see Definition~\ref{defn.intertwined.points} and Remark~\ref{rmrk.not.hausdorff}), and points being quasiregular, %
weakly regular, and regular (Definition~\ref{defn.tn}), (un)conflicted (Definition~\ref{defn.conflicted}(\ref{item.unconflicted})), and hypertransitive (Definition~\ref{defn.sc}).\footnote{An extra word on the converse of this:  Our theory of semitopologies admits spaces whose points partition into distinct communities, as discussed in Theorem~\ref{thrm.topen.partition} and Remark~\ref{rmrk.partition}.  To a professional blockchain engineer it might seem terrible if two points points are \emph{not} intertwined, since this means they might not be in consensus in a final state. 
Should this not be excluded by the definition of semitopology, as is done in the literature on quorum systems, where it typically definitionally assumed that all quorums in a quorum system intersect?  
No! 
Separation is a fact of life which we permit not only so that we can mathematically analyse it (and we do), but also because we may need it for certain \emph{normal situations}.
For example, most blockchains have a \emph{mainnet} and several \emph{testnets} and it is understood that each should be coherent within itself, but different nets \emph{need not} be in consensus with one another.  Indeed, if the mainnet had to agree with a testnet then this would likely be a bug, not a feature.  So the idea of having multiple partitions is nothing new \emph{per se}.  It is a familiar idea, which semitopologies put in a powerfully general mathematical context.}
\item
\emph{Topology:}\quad 
If a minimal open neighbourhood of a point exists then it is least, because we can intersect two minimal neighbourhoods to get a smaller one which by minimality is equal to both.

Yet, in topology the existence of a least open neighbourhood is not guaranteed (e.g. $0\in\mathbb R$ has no least open neighbourhood).

\emph{Semitopology:}\ 
A point may have multiple minimal open neighbourhoods --- examples are very easy to generate, see e.g. the top-right example in Figure~\ref{fig.012}.
\item
\emph{Topology:}\quad 
We are typically interested in functions on topologies that are continuous (or mostly so, e.g. $f(x)=1/x$).
Thus for example, the definition of $\tf{Top}$ the category of topological spaces takes continuous functions as morphisms, essentially building in assumptions that continuous functions are of most interest and that finding them is enough of a solved problem that we can restrict to continuous functions in the definition.
 
\emph{Semitopology:}\quad 
For our intended application to consensus, 
we are explicitly interested in functions that may be discontinuous.
This models initial and intermediate states where local consensus has not yet been achieved, or final states on semitopologies that include disjoint topens and non-regular points (e.g. conflicted points), as well as adversarial or failing behaviour.
Thus, having continuity is neither a solved problem, nor even necessarily desirable.
\item
Sometimes, ideas that come from semitopology project carry over to topology, but they lose impact or become less interesting in doing so. 
For example: our theory of semitopologies considers notions of \emph{topen set} and \emph{strongly topen set} (Definitions~\ref{defn.transitive} and~\ref{defn.strongly.transitive}).
In topology these are equivalent to one another, and to a known and simpler topological property of being \emph{hyperconnected} (Definition~\ref{defn.tangled}).\footnote{\dots but (strong) topens are their own thing.  Analogy: a projection from $\mathbb C$ to $\mathbb R$ maps $a+bi$ to $a$; this is not evidence that $i$ is equivalent to $0$!} 
\item
Semitopological questions such as \emph{`is this a topen set'} or \emph{`are these two points intertwined'} or \emph{`does this point have a topen neighbourhood'} --- and many other definitions, such as our taxonomy of points into \emph{regular}, \emph{weakly regular}, %
\emph{quasiregular}, \emph{unconflicted}, and \emph{hypertransitive} %
--- appear to be novel.

Also in the background 
is that we are particularly interested in properties and algorithms that work well using local and possibly incomplete or even partially incorrect information.

Thus semitopologies have their own distinct character: because they are mathematically distinct, and because modern applications having to do with actionable coalitions and distributed systems motivate us to ask questions that have not necessarily been considered before.
\end{enumerate}

\jamiesubsection{Related work}
\label{subsect.related.work}

\paragraph*{Union sets and minimal structures}

There is a thread of research into \emph{union-closed families}; these are subsets of a finite powerset closed under unions, so that a union-closed family is precisely just a finite semitopology. 
The motivation is to study the combinatorics of finite subsemilattices of a powerset.
Some progress has been made in this~\cite{poonen:unicf}; the canonical reference for the relevant combinatorial conjectures is the `problem session' on page~525 (conjectures 1.9, 1.9', and 1.9") of~\cite{rival:grao}.
See also recent progress in a conjecture about union-closed families (\href{https://web.archive.org/web/20230330170701/https://en.wikipedia.org/wiki/Union-closed_sets_conjecture#Partial_results}{permalink}).

There is no direct connection to this work, though the combinatorial properties considered may yet become useful for proving properties of concrete algorithms.

A \emph{minimal structure} on a set $X$ is a subset of $\powerset(X)$ that contains $\varnothing$ and $X$.
Thus a semitopology is a minimal structure that is also closed under arbitrary unions.
There is a thread of research into minimal structures, studying how notions familiar from topology (such as continuity) fare in weak (minimal) settings~\cite{noiri:defsgf} and how this changes as axioms (such as closure under unions) are added or removed.
An accessible discussion is in~\cite{szaz:minsgt}, and see the brief but comprehensive references in Remark~3.7 of that paper.
Of course our focus is on properties of semitopologies 
which are not considered in that literature; but we share an observation with minimal structures that it is useful to study topology-like constructs, in the absence of closure under intersections. 

\paragraph*{Gradecast converges on a topen}

Many consensus algorithms have the property that once consensus is established in a quorum $O$, it propagates to $\closure{O}$.
For example, in the Grade-Cast algorithm~\cite{feldman_optimal_1988}, participants assign a confidence grade of 0, 1 or 2 to their output and must ensure that if any participant outputs $v$ with grade 2 then all must output $v$ with grade at least 1.
If all the quorums of a participant intersect some set $S$ that unanimously supports value $v$, then the participant assigns grade at least 1 to $v$.

From the view of our paper, this is just taking a closure, which suggests that, to convince a topen to agree on a value, it would suffice to first convince an open neighbourhood that intersects the topen, and then use Grade-Cast to convince the whole topen.
See also Proposition~\ref{prop.open.strong-consensus} and Remark~\ref{rmrk.gradecast}.

\paragraph*{Algebraic topology as applied to distributed computing tasks}

Continuing the discussion of tasks above, the reader may know that solvability results about distributed computing tasks have been obtained from algebraic topology, starting with the impossibility of wait-free $k$-set consensus using read-write registers and the Asynchronous Computability Theorem~\cite{herlihy_asynchronous_1993,borowsky_generalized_1993,saks_wait-free_1993} in 1993.
See~\cite{herlihy_distributed_2013} for numerous such results.

The basic observation is that the set of final states of a distributed algorithm forms a simplicial complex, called the \emph{protocol complex}, and topological properties of this complex, like connectivity, are constrained by the underlying communication and fault model.
These topological properties in turn can determine what tasks are solvable.
For example: every algorithm in the wait-free model with atomic read-write registers has a connected protocol complex, and because the consensus task's output complex is disconnected, consensus in this model is not solvable~\cite[Chapter~4]{herlihy_distributed_2013}.

This work is also topological, but in a different way: we use (semi)topologies to study consensus in and of itself, rather than the solvability of consensus or other tasks in particular computation models.
Put another way: the papers cited above use topology to study the solvability of distributed tasks, but we show here how the very idea of `distribution' can be viewed as having a semitopological foundation.

Of course we can imagine that these might be combined --- that in future work we may find interesting and useful things to say about the topologies of distributed algorithms when viewed as algorithms \emph{on} and \emph{in} a semitopology.

\paragraph*{Fail-prone systems and quorum systems}

Given a set of processes $\ns P$, a \emph{fail-prone} system~\cite{malkhi_byzantine_1998}  (or \emph{adversary structure}~\cite{hirt_player_2000}) is a set of \emph{fail-prone sets} $\mathcal{F}=\{F_1,...,F_n\}$ where, for every $1\leq i\leq n$, $F_i\subseteq \ns P$.
$\mathcal{F}$ denotes the assumptions that the set of processes that will fail (potentially maliciously) is a subset of one of the fail-prone sets.
A \emph{dissemination quorum system} for $\mathcal{F}$ is a set  $\{Q_1,..., Q_m\}$ of quorums where, for every $1\leq i\leq m$, $Q_i\subseteq \ns P$, and such that 
\begin{itemize*}
\item
for every two quorums $Q$ and $Q'$ and for every fail-prone set $F$, $\left(Q\cap Q'\right)\setminus F\neq\emptyset$ and 
\item
for every fail-prone set $F$, there exists a quorum disjoint from $F$.
\end{itemize*}
Several distributed algorithms, such as Bracha Broadcast~\cite{bracha_asynchronous_1987} and PBFT~\cite{castro_practical_2002}, rely on a quorum system for a fail-prone system $\mathcal{F}$ in order to solve problems such as reliable broadcast and consensus assuming (at least) that the assumptions denoted by $\mathcal{F}$ are satisfied.

Several recent works generalise the fail-prone system model.
Under the failure assumptions of a traditional fail-prone system, Bezerra et al.~\cite{bezerra_relaxed_2022} study reliable broadcast when participants each have their own set of quorums.
Asymmetric Fail-Prone Systems~\cite{cachin_asymmetric_2019} generalise fail-prone systems to allow participants to make different failure assumption and have different quorums.
In Permissionless Fail-Prone Systems~\cite{cachin_quorum_2023}, participants not only make assumptions about failures, but also make assumptions about the assumptions of other processes;
the resulting structure seems closely related to semitopologies, but the exact relationship still needs to be elucidated.

In Federated Byzantine Agreement Systems~\cite{mazieres2015stellar}, participants declare quorum slices and quorums emerge out of the collective quorum slices of their members.
García-Pérez and Gotsman~\cite{garcia2018federated} rigorously prove the correctness of broadcast abstractions in Stellar's Federated Byzantine Agreement model and investigate the model's relationship to dissemination quorum systems.
The Personal Byzantine Quorum System model~\cite{losa:stecbi} is an abstraction of Stellar's Federated Byzantine Agreement System model and accounts for the existence of disjoint consensus clusters (in the terminology of the paper) which can each stay in agreement internally but may disagree between each other.
Consensus clusters are closely related to the notion of topen in Definition~\ref{defn.transitive}(\ref{transitive.cc}).

Sheff et al. study heterogeneous consensus in a model called Learner Graphs~\cite{sheff_heterogeneous_2021} and propose a consensus algorithm called Heterogeneous Paxos.

Cobalt, the Stellar Consensus Protocol, Heterogeneous Paxos, and the Ripple Consensus Algorithm~\cite{macbrough_cobalt_2018,mazieres2015stellar,sheff_heterogeneous_2021,schwartz_ripple_2014} are consensus algorithms that rely on heterogeneous quorums or variants thereof.
The Stellar network~\cite{lokhafa:fassgp} and the XRP Ledger~\cite{schwartz_ripple_2014} are two global payment networks that use heterogeneous quorums to achieve consensus among an open set of participants.

Quorum systems and semitopologies are not the same thing.
Quorum systems are typically taken to be such that all quorums intersect (in our terminology: they are \emph{intertwined}), whereas semitopologies do not require this.
On the other hand, quorums are not always taken to be closed under arbitrary unions, whereas semitopologies are (see the discussion in Example~\ref{xmpl.semitopologies}(\ref{item.quorum.system})).

But there are also differences in how the maths has been used and understood.
This paper has been all about point-set topology flavoured ideas, whereas the literature on fail-prone systems and quorum systems has been most interested in synchronisation algorithms for distributed systems. 
We see these interests as complementary, and the difference in emphasis is a feature, not a bug.
Some work by the second author and others~\cite{losa:stecbi} gets as far as proving an analogue to Lemma~\ref{lemm.cc.unions} (though we think it is fair to say that the presentation in this paper is much simpler and more clear), but it fails to notice the connection with topology and the subsequent results which we present in this paper.

\jamiesubsection{Future work} 

We briefly outline some ways in which this work can be extended and improved:
\begin{enumerate}
\item
In Definition~\ref{defn.value.assignment} we define a \emph{value assignment} $f:\ns P\to\tf{Val}$ to be a function from a semitopology to a codomain $\tf{Val}$ that is given the discrete semitopology.
This is a legitimate starting point, but of course we should consider more general codomains.
This could include an arbitrary semitopology on the right (for greatest generality), but even for our intended special case of consensus it would be interesting to try to endow $\tf{Val}$ with a semilattice structure (or something like it), at least, e.g. to model merging of distinct updates to a ledger.\footnote{We write `something like it' because we might also consider, or consider excluding, possibly conflicting updates.}
We can easily generate a (semi)topology from a semilattice by taking points to be elements of the lattice and open sets to be up-closed sets, and this would be a natural generalisation of the discrete semitopologies we have used so far.
\item 
We have not considered what would correspond to the exponential (or \emph{Vietoris}) semitopology.
Semilattice representation results exist~\cite{bredhikin:repts}, but a design space exists here and we should look for representations well-suited to computationally verifying or refuting properties of semitopologies.
\item
We mentioned in Subsection~\ref{subsect.related.work} that semitopology is not about algebraic topology applied to solvability of distributed computing tasks.
These are distinct topics, and the fact that they share a word in their name does not make them any more equal than a Great Dane and a Danish pastry.

But, it is a very interesting question what algebraic \emph{semi}topology might look like.
To put this another way: what is the geometry of semitopological spaces?
We would very much like to know.
\item
It remains to consider Byzantine behaviour, by which we mean that some participants may misreport their view of the network in order to `invent' or sabotage quorums and so influence the outcome of consensus.

So for instance we can ask: ``What conditions can we put on a semitopology consisting of a single toppen to guarantee that changing it at one point $p$ will not make that topen split into two topens?''
Thus intuitively, given a semitopology $(\ns P,\opens)$ we are interested in asking how properties range over an `$\epsilon$-ball' of perturbed semitopologies --- as might be caused by various possible non-standard behaviours from a limited number of Byzantine points --- and in particular we are looking for criteria to guarantee that appropriately-chosen good properties be preserved under perturbation.
\item
We have studied how consensus, once achieved on an open set $O$, propagates to its closure $\closure{O}$; see Proposition~\ref{prop.open.strong-consensus} and Remark~\ref{rmrk.gradecast}. 
But this is just half of the problem of consensus: it remains to consider (within our semitopological framework) what it is to attain consensus on some open set in the first place.

That is: suppose $(\ns P,\opens)$ is a semitopology and $f:\ns P\to\tf{Var}$ is a value assignment.
Then what does it mean, in maths and algorithms, to find a value assignment $f':\ns P\to\tf{Var}$ that is `close' to $f$ but is continuous on some open set $O$?
In this paper we have constructed a theory of what it would then be to extend $f'$ to an $f''$ that continuously extends $f'$ to regular points; but we have not yet looked at how to build the $f'$. 
We speculate that unauthenticated Byzantine consensus algorithms (like Information-Theoretic HotStuff~\cite{abraham_information_2020}) can be understood in our setting; unlike authenticated algorithms, unauthenticated algorithms do not rely on one participant being able to prove to another, by exhibiting signed messages, that a quorum has acted in a certain way.
\item
We have not considered morphisms of semitopologies and how to organise semitopologies into a category, in this paper --- but see next paragraph.
\end{enumerate}

A more extensive treatment of semitopologies is also available~\cite{gabbay:semdca} (completed since this paper was first submitted): its first part includes and extends the material in this paper; its second part treats the category of semitopologies and constructs its categorical dual as a category of \emph{semiframes}; and its third part builds a three-valued logic on semitopologies which we use to study, and expand on, the antiseparation properties we consider here.
The reader who found this paper too brief can find further reading there.

\input{topoc.bbl}

\section*{Acknowledgements}

I am extremely grateful to an anonymous referee for their careful attention to this material and for their feedback. 
Thanks to Giuliano Losa and the Stellar Development Foundation for their generous support and funding.
Giuliano was instrumental in explaining the issues, reviewing material, and guiding me to interesting problems; thank you for the many interesting discussions and insightful comments.

\end{document}

%% file: topoc.bbl
\newcommand{\etalchar}[1]{$^{#1}$}
\hyphenation{Mathe-ma-ti-sche}
\providecommand{\bysame}{\leavevmode\hbox to3em{\hrulefill}\thinspace}
\providecommand{\MR}{\relax\ifhmode\unskip\space\fi MR }
\providecommand{\MRhref}[2]{%
  \href{http://www.ams.org/mathscinet-getitem?mr=#1}{#2}
}
\providecommand{\href}[2]{#2}

%% file: topoc.bbl
\begin{thebibliography}{SWRM21}

\bibitem[ACTZ24]{Alpos2024}
Orestis Alpos, Christian Cachin, Bj{\"o}rn Tackmann, and Luca Zanolini,
  \emph{Asymmetric distributed trust}, Distributed Computing (2024), Available
  online at \url{https://doi.org/10.1007/s00446-024-00469-1}.

\bibitem[AS21]{abraham_information_2020}
Ittai Abraham and Gilad Stern, \emph{{Information} {Theoretic} {HotStuff}},
  24th International Conference on Principles of Distributed Systems, {OPODIS}
  2020, December 14-16, 2020, Strasbourg, France (Virtual Conference)
  (Dagstuhl, Germany) (Quentin Bramas, Rotem Oshman, and Paolo Romano, eds.),
  LIPIcs, vol. 184, Schloss Dagstuhl - Leibniz-Zentrum f{\"{u}}r Informatik,
  2021, pp.~11:1--11:16.

\bibitem[BG93]{borowsky_generalized_1993}
Elizabeth Borowsky and Eli Gafni, \emph{Generalized {FLP} {Impossibility}
  {Result} for {T}-resilient {Asynchronous} {Computations}}, Proceedings of the
  {Twenty}-fifth {Annual} {ACM} {Symposium} on {Theory} of {Computing} (New
  York, NY, USA), {STOC} '93, ACM, 1993, pp.~91--100.

\bibitem[BKK22]{bezerra_relaxed_2022}
João~Paulo Bezerra, Petr Kuznetsov, and Alice Koroleva, \emph{Relaxed reliable
  broadcast for decentralized trust}, Networked Systems (Mohammed-Amine
  Koulali and Mira Mezini, eds.), Lecture Notes in Computer Science, Springer
  International Publishing, 2022, pp.~104--118.

\bibitem[Bou98]{bourbaki:gent1}
Nicolas Bourbaki, \emph{General topology: Chapters 1-4}, Elements of
  mathematics, Springer, 1998.

\bibitem[Bra87]{bracha_asynchronous_1987}
Gabriel Bracha, \emph{Asynchronous {Byzantine} agreement protocols},
  Information and Computation \textbf{75} (1987), no.~2, 130--143.

\bibitem[Bre84]{bredhikin:repts}
D.~A. Bredhikin, \emph{A representation theorem for semilattices}, Proceedings
  of the American Mathematical Society \textbf{90} (1984), no.~2, 219--220.

\bibitem[CL02]{castro_practical_2002}
Miguel Castro and Barbara Liskov, \emph{Practical {Byzantine} fault tolerance
  and proactive recovery}, ACM Transactions on Computer Systems (TOCS)
  \textbf{20} (2002), no.~4, 398--461.

\bibitem[CLZ23]{cachin_quorum_2023}
Christian Cachin, Giuliano Losa, and Luca Zanolini, \emph{Quorum systems in
  permissionless networks}, 26th International Conference on Principles of
  Distributed Systems ({OPODIS} 2022) (Eshcar Hillel, Roberto Palmieri, and
  Etienne Rivière, eds.), Leibniz International Proceedings in Informatics
  ({LIPIcs}), vol. 253, Schloss Dagstuhl – Leibniz-Zentrum für Informatik,
  2023, {ISSN}: 1868-8969, pp.~17:1--17:22.

\bibitem[CT19]{cachin_asymmetric_2019}
Christian Cachin and Björn Tackmann, \emph{Asymmetric {Distributed} {Trust}},
  arXiv:1906.09314 [cs] (2019), arXiv: 1906.09314.

\bibitem[DP02]{priestley:intlo}
B.~A. Davey and Hilary~A. Priestley, \emph{Introduction to lattices and order},
  2 ed., Cambridge University Press, 2002.

\bibitem[Eng89]{engelking:gent}
Ryszard Engelking, \emph{General topology}, Sigma Series in Pure Mathematics,
  Heldermann Verlag, 1989.

\bibitem[FHNS22]{florian_sum_2022}
Martin Florian, Sebastian Henningsen, Charmaine Ndolo, and Björn Scheuermann,
  \emph{The sum of its parts: {Analysis} of federated byzantine agreement
  systems}, Distributed Computing \textbf{35} (2022), no.~5, 399--417 (en).

\bibitem[FM88]{feldman_optimal_1988}
Paul Feldman and Silvio Micali, \emph{Optimal algorithms for {B}yzantine
  agreement}, Proceedings of the twentieth annual {ACM} symposium on Theory of
  computing, {STOC} '88, Association for Computing Machinery, 1 1988,
  pp.~148--161.

\bibitem[Gab24]{gabbay:semdca}
Murdoch~J. Gabbay, \emph{Semitopology: decentralised collaborative action via
  topology, algebra, and logic}, College Publications, August 2024, ISBN
  9781848904651.

\bibitem[GPG18]{garcia2018federated}
{\'A}lvaro Garc{\'\i}a-P{\'e}rez and Alexey Gotsman, \emph{Federated byzantine
  quorum systems}, 22nd International Conference on Principles of Distributed
  Systems (OPODIS 2018), Schloss Dagstuhl-Leibniz-Zentrum fuer Informatik,
  2018.

\bibitem[HKR13]{herlihy_distributed_2013}
Maurice Herlihy, Dmitry Kozlov, and Sergio Rajsbaum, \emph{Distributed
  computing through combinatorial topology}, Morgan Kaufmann, 2013.

\bibitem[HM00]{hirt_player_2000}
Martin Hirt and Ueli Maurer, \emph{Player {Simulation} and {General}
  {Adversary} {Structures} in {Perfect} {Multiparty} {Computation}}, Journal of
  Cryptology \textbf{13} (2000), no.~1, 31--60.

\bibitem[HS93]{herlihy_asynchronous_1993}
Maurice Herlihy and Nir Shavit, \emph{The asynchronous computability theorem
  for t-resilient tasks}, Proceedings of the twenty-fifth annual {ACM}
  symposium on {Theory} of computing, 1993, pp.~111--120.

\bibitem[Hul03]{hulek:eleag}
Klaus Hulek, \emph{Elementary algebraic geometry}, Student Mathematical
  Library, vol.~20, American Mathematical Society, 2003.

\bibitem[JT04]{doi:10.1080/00029890.2004.11920120}
Tyler Jarvis and James Tanton, \emph{The hairy ball theorem via {S}perner's
  lemma}, The American Mathematical Monthly \textbf{111} (2004), no.~7,
  599--603.

\bibitem[Kop89]{koppelberg:hanba1}
Sabine Koppelberg, \emph{Handbook of {B}oolean algebras, volume 1},
  North-Holland, 1989, Series editors Robert Bonnet and James Donald Monk.

\bibitem[Lam98]{lamport_part-time_1998}
Leslie Lamport, \emph{The part-time parliament}, {ACM} Transactions on Computer
  Systems \textbf{16} (1998), no.~2, 133--169.

\bibitem[LCL23]{li_quorum_2023}
Xiao Li, Eric Chan, and Mohsen Lesani, \emph{{Quorum Subsumption for
  Heterogeneous Quorum Systems}}, 37th International Symposium on Distributed
  Computing (DISC 2023) (Dagstuhl, Germany) (Rotem Oshman, ed.), Leibniz
  International Proceedings in Informatics (LIPIcs), vol. 281, Schloss Dagstuhl
  -- Leibniz-Zentrum f{\"u}r Informatik, 2023, pp.~28:1--28:19.

\bibitem[LGM19]{losa:stecbi}
Giuliano Losa, Eli Gafni, and David Mazi{\`e}res, \emph{Stellar consensus by
  instantiation}, 33rd International Symposium on Distributed Computing (DISC
  2019) (Dagstuhl, Germany) (Jukka Suomela, ed.), Leibniz International
  Proceedings in Informatics (LIPIcs), vol. 146, Schloss
  Dagstuhl--Leibniz-Zentrum fuer Informatik, 2019, pp.~27:1--27:15.

\bibitem[LL23]{li_open_2023}
Xiao Li and Mohsen Lesani, \emph{Open {Heterogeneous} {Quorum} {Systems}},
  April 2023, arXiv:2304.02156 [cs].

\bibitem[LLM{\etalchar{+}}19]{lokhafa:fassgp}
Marta Lokhava, Giuliano Losa, David Mazi\`eres, Graydon Hoare, Nicolas Barry,
  Eli Gafni, Jonathan Jove, Rafa\l{} Malinowsky, and Jed McCaleb, \emph{Fast
  and secure global payments with {S}tellar}, Proceedings of the 27th ACM
  Symposium on Operating Systems Principles (New York, NY, USA), SOSP '19,
  Association for Computing Machinery, 2019, p.~80–96.

\bibitem[LSP82]{lamport:byzgp}
Leslie Lamport, Robert Shostak, and Marshall Pease, \emph{{The Byzantine
  Generals Problem}}, ACM Trans. Program. Lang. Syst. \textbf{4} (1982), no.~3,
  382–401.

\bibitem[Mac18]{macbrough_cobalt_2018}
Ethan MacBrough, \emph{Cobalt: {BFT} {Governance} in {Open} {Networks}},
  February 2018, Available online at
  \url{https://doi.org/10.48550/arXiv.1802.07240}.

\bibitem[Maz15]{mazieres2015stellar}
David Mazi{\`e}res, \emph{The {S}tellar consensus protocol: a federated model
  for {I}nternet-level consensus}, Tech. report, {Stellar Development
  Foundation}, 2015,
  \url{https://www.stellar.org/papers/stellar-consensus-protocol.pdf}
  (permalink:
  \url{https://web.archive.org/web/20240629063518/https://stellar.org/learn/stellar-consensus-protocol}).

\bibitem[MR98]{malkhi_byzantine_1998}
Dahlia Malkhi and Michael Reiter, \emph{Byzantine quorum systems}, Distributed
  computing \textbf{11} (1998), no.~4, 203--213.

\bibitem[NW94]{naor:loacaq}
Moni Naor and Avishai Wool, \emph{The load, capacity and availability of quorum
  systems}, Proceedings 35th Annual Symposium on Foundations of Computer
  Science, vol.~27, IEEE, 1994, pp.~214--225.

\bibitem[PN01]{noiri:defsgf}
Valeriu Popa and Takashi Noiri, \emph{On the definitions of some generalized
  forms of continuity under minimal conditions}, Memoirs of the Faculty of
  Science. Series A. Mathematics \textbf{22} (2001), 9--18.

\bibitem[Poo92]{poonen:unicf}
Bjorn Poonen, \emph{Union-closed families}, Journal of Combinatorial Theory,
  Series A \textbf{59} (1992), no.~2, 253--268.

\bibitem[Rik62]{riker:thepc}
William~H. Riker, \emph{The theory of political coalitions}, Yale University
  Press, 1962.

\bibitem[Riv85]{rival:grao}
Ivan Rival (ed.), \emph{Graphs and order. the role of graphs in the theory of
  ordered sets and its applications}, Proceedings of NATO ASI, Series C (ASIC),
  vol. 147, Banff/Canada, 1985.

\bibitem[Rya10]{ryan:hisidf}
Johnny Ryan, \emph{A history of the internet and the digital future}, Reaktion
  Books, 2010.

\bibitem[Rya11]{ars-technica:howabg}
\bysame, \emph{How the atom bomb helped give birth to the internet},
  \url{https://arstechnica.com/tech-policy/2011/02/how-the-atom-bomb-gave-birth-to-the-internet/},
  2 2011, Permalink:
  \url{http://web.archive.org/web/20240622221756/https://arstechnica.com/tech-policy/2011/02/how-the-atom-bomb-gave-birth-to-the-internet/}.

\bibitem[SWRM21]{sheff_heterogeneous_2021}
Isaac Sheff, Xinwen Wang, Robbert~van Renesse, and Andrew~C. Myers,
  \emph{Heterogeneous {Paxos}}, 24th {International} {Conference} on
  {Principles} of {Distributed} {Systems} ({OPODIS} 2020) (Dagstuhl, Germany)
  (Quentin Bramas, Rotem Oshman, and Paolo Romano, eds.), Leibniz
  {International} {Proceedings} in {Informatics} ({LIPIcs}), vol. 184, Schloss
  Dagstuhl–Leibniz-Zentrum für Informatik, 2021, ISSN: 1868-8969,
  pp.~5:1--5:17.

\bibitem[SYB14]{schwartz_ripple_2014}
David Schwartz, Noah Youngs, and Arthur Britto, \emph{The {R}ipple {P}rotocol
  {C}onsensus {A}lgorithm}, Ripple Labs Inc White Paper \textbf{5} (2014),
  no.~8, 151.

\bibitem[SZ93]{saks_wait-free_1993}
Michael Saks and Fotios Zaharoglou, \emph{Wait-free k-set agreement is
  impossible: {The} topology of public knowledge}, Proceedings of the
  twenty-fifth annual {ACM} symposium on {Theory} of computing, 1993,
  pp.~101--110.

\bibitem[Sz{\'a}07]{szaz:minsgt}
{\'A}rp{\'a}d Sz{\'a}z, \emph{Minimal structures, generalized topologies, and
  ascending systems should not be studied without generalized uniformities},
  Filomat (Nis) \textbf{21} (2007), 87--97.

\bibitem[Wes11]{sep-generalized-quantifiers}
Dag Westerst\r{a}hl, \emph{Generalized quantifiers}, The Stanford Encyclopedia
  of Philosophy (Edward~N. Zalta, ed.), Metaphysics research lab, CSLI,
  Stanford University, summer 2011 ed., 2011, Available online at
  \url{https://plato.stanford.edu/archives/sum2011/entries/generalized-quantifiers/}
  (permalink: \url{https://archive.ph/YlP09}).

\bibitem[Wil70]{willard:gent}
Stephen Willard, \emph{General topology}, Addison-Wesley, 1970, Reprinted by
  Dover Publications.

\end{thebibliography}
